\documentclass[12pt]{article}
\usepackage{amsmath}
\usepackage{amsfonts}
\usepackage{makeidx}
\usepackage{amssymb}
\usepackage{times}

\newtheorem{satz}{Theorem}[section]
\newtheorem{definition}[satz]{Definition}
\newtheorem{lemma}[satz]{Lemma}
\newtheorem{koro}[satz]{Corollary}
\newtheorem{bemerkung}[satz]{Remark}

\newtheorem{notation}[satz]{Notation}
\newenvironment{proof}{\par\noindent {\it Proof:} \hspace{7pt}}{\hfill\hbox{\vrule width 7pt depth 0pt height 7pt}
\par\vspace{10pt}}

\input{tcilatex}
\begin{document}

\title{From the 2nd Law of Thermodynamics to AC--Conductivity Measures of
Interacting Fermions in Disordered Media}
\author{J.-B. Bru \and W. de Siqueira Pedra}
\date{\today}
\maketitle

\begin{abstract}
We study the dynamics of interacting lattice fermions with random hopping
amplitudes and random static potentials, in presence of time--dependent
electromagnetic fields. The interparticle interaction is short--range and
translation invariant. Electromagnetic fields are compactly supported in
time and space. In the limit of infinite space supports (macroscopic limit)
of electromagnetic fields, we derive Ohm and Joule's laws in the AC--regime.
An important outcome is the extension to interacting fermions of the notion
of macroscopic AC--conductivity measures, known so far only for free
fermions with disorder. Such excitation measures result from the 2nd law of
thermodynamics and turn out to be L\'{e}vy measures. As compared to the
Drude (Lorentz--Sommerfeld) model, widely used in Physics, the quantum
many--body problem studied here predicts a much smaller AC--conductivity at
large frequencies. This indicates (in accordance with experimental results)
that the relaxation time of the Drude model, seen as an effective parameter
for the conductivity, should be highly frequency--dependent.\ We conclude by
proposing an alternative effective description -- using L\'{e}vy Processes
in Fourier space -- of the pheno%
\-%
menon of electrical conductivity.
\end{abstract}

\tableofcontents%

\section{Introduction}

The present paper belongs to a succession of works on electrical
conductivity, starting with \cite{OhmI,OhmII,OhmIII,OhmIV,OhmV}.

As claimed in the famous paper \cite[p. 505]{meanfreepath},
\textquotedblleft \textit{it must be admitted that there is no entirely
rigorous quantum theory of conductivity}.\textquotedblright\ Concerning
AC--conductivity, however, in the last years significant mathematical
progress has been made. See \cite%
{Annale,JMP-autre,JMP-autre2,Cornean,OhmI,OhmII,OhmIII,OhmIV,W} for examples
of mathematically rigorous derivations of linear conductivity from first
principles of quantum mechanics in the AC--regime. These results indicate a
picture of the microscopic origin of Ohm and Joule's laws which differs from
usual explanations coming from the Drude (Lorentz--Sommerfeld) model. This
is discussed with more details in \cite{OhmII}.

As electrical resistance of conductors may result from the presence of
interactions between charge carriers, the main drawback of these studies is
their restriction to non--interacting systems.

A first attempt in this direction has been tried in the parallel work \cite%
{W}, which uses like us an algebraic approach to tackle such problems. With
regard to interacting systems, explicit constructions of KMS states are
obtained in the Ph.D. thesis \cite{W} for a one--dimensional model of
interacting fermions with a finite range pair interaction. But, the author
studies in \cite[Chap. 9]{W} the linear response theory only for
non--interacting fermions, keeping in mind possible generalizations to
interacting systems.

Therefore, we aim to extend \cite{OhmI,OhmII,OhmIII,OhmIV} to fermion
systems with interactions. As a first step, \cite{OhmV} proves all
assertions of \cite{OhmI,OhmII} for fermion systems with short range
interactions. We perform here the second part of this program by extending
the main results of \cite{OhmIII} to interacting fermion systems:

\begin{itemize}
\item Like in \cite{OhmV}, we investigate some non--auto%
\-%
nomous $C^{\ast }$--dyna%
\-%
mical system on the CAR $C^{\ast }$--algebra of cubic infinite lattices of
any dimension. The (non--auto%
\-%
nomous) dynamics is generated by short--range and translation invariant
interactions between particles, random static potentials, and also random
next neighbor hopping amplitudes in presence of local and time--dependent
electromagnetic fields. Disorder is here defined via ergodic distributions
of random potentials and hopping amplitudes.

\item We study the linear response of interacting fermions at thermal
equilibrium in disordered media to \emph{macroscopic} electric fields that
are time-- and space--dependent. In particular, we obtain Ohm's law with a
(charge) transport coefficient that is a continuous function (of time)
naturally called ma%
\-%
croscopic conductivity.

\item The Fourier transform of the conductivity is named macroscopic AC--con%
\-%
ductivity measure. The fact that this Fourier transform is indeed a measure
follows from the 2nd law of thermodynamics (see Remark \ref{AC--Conductivity
Measure from the 2nd Law}). The latter corresponds here to the
Kelvin--Planck statement while avoiding the concept of \textquotedblleft
cooling\textquotedblright\ \cite[p. 49]{lieb-yngvasonPhysReport}. In
particular, the 2nd law yields the positivity of the heat production for
cyclic processes on equilibrium states. The concept of conductivity measure,
introduced by Klein, Lenoble and M\"{u}ller, has already been used in \cite%
{Annale,JMP-autre,JMP-autre2,OhmII,OhmIII,OhmIV} for the \emph{non}%
--interacting case.

\item We give a comparison of our results and the Drude
(Lorentz--Sommerfeld) model, widely used in Physics \cite%
{meanfreepath,Anderson-physics} to describe the phenomenon of electrical
conductivity. See also \cite{brupedrahistoire} for a historical perspective
of this subject. In particular, we show that the Drude model and its
refinements (like the Drude--Lorentz and the Lorentz--Sommerfeld models)
\emph{always} overestimate\ the in--phase conductivity at high frequencies.
This indicates that the relaxation time of the Drude model, seen as an
effective parameter for the conductivity, should be frequency--dependent, as
already observed for instance in \cite{T,NS1,NS2,SE,Y}. In fact, it should
either vanish or diverge at large frequencies.\

\item We show that the AC--conducti%
\-%
vity measure of the system under consideration is always a L\'{e}vy measure.
An alternative effective description of the pheno%
\-%
menon of linear conductivity by using L\'{e}vy Processes in Fourier space is
discussed. This is reminiscent of Boltzmann equation for collective
oscillatory modes of charge excitations. It was recently shown that L\'{e}vy
statistics can efficiently describe quantum phenomena like (subrecoil) laser
cooling \cite{9780511755668}. As far as we know, there is no mathematically
rigorous proof of this fact.
\end{itemize}

\noindent Note that also new results not presented in \cite{OhmIII} are
obtained here, even for non--interacting fermions. For instance, in contrast
with \cite{OhmIII}, the hopping amplitudes are allowed to be
non--homogeneous in space.

Like in \cite{OhmV}, these results are made possible by \emph{Lieb--Robinson
bounds for multi--commutators}. See \cite{brupedraLR} for more details.
Indeed, we need to get error terms uniformly bounded with respect to
(w.r.t.) the random parameters and the volume $\Lambda _{l}$ of the box
where the electromagnetic field lives. This is a crucial step to get
valuable information in the limit $l\rightarrow \infty $\ of macroscopic
electromagnetic fields, otherwise the results presented here would loose
almost all its interest. To get such error terms, we apply in \cite%
{OhmI,OhmII,OhmIII,OhmIV} tree--decay bounds on multi--commutators in the
sense of \cite[Section 4]{OhmI}. The latter are based on combinatorial
results \cite[Theorem 4.1]{OhmI} already used before, for instance in \cite%
{FroehlichMerkliUeltschi}. Nevertheless, \cite{OhmI,OhmII,OhmIII,OhmIV} or
\cite{FroehlichMerkliUeltschi} require Bogoliubov automorphisms (see \cite[%
Theorem 5.2.5]{BratteliRobinson}). In other words, only \emph{non}%
--interacting fermion systems can be tackled with such combinatorial
arguments (like \cite[Theorem 4.1]{OhmI}).

A solution to that issue for the \emph{interacting} case has only been
recently given in \cite{brupedraLR} via Lieb--Robinson bounds for
multi--commutators\footnote{%
Only Lieb--Robinson bounds for multi--commutators of order 3 is necessary in
our study.}, which is \emph{not} an obvious extension of usual
Lieb--Robinson bounds known since 1972 \cite{liebrobinsonbounds}. This is
explained in \cite[Sections 3.3, 4.3]{brupedraLR} and \cite{OhmV}. Note that
Lieb--Robinson bounds have also been recognized as an important ingredient
in \cite{W} via the so--called strong localization criterion, see \cite[%
Definition 10.1]{W}. Nevertheless, they are not sufficient for our purpose.
Indeed, its extensions to multi--commutators turn out to be pivotal in
Theorem \ref{Theorem AC conductivity measure copy(1)}, which is used to
prove Theorem \ref{Theorem AC conductivity measure}. In fact, the
Lieb--Robinson bounds for multi--commutators make the present paper much
easier from the technical point of view, even for interacting systems.
Compare, for instance, Section \ref{section energetic study} with \cite[%
Section 5.4]{OhmIII}. As a consequence, important conceptual issues, like
the derivation of AC--conductivity measures\footnote{%
The conductivity measure can, indeed, be seen as a excitation measure
related to electric perturbations. Similar constructions can be performed
for many classes of pertubations because of the 2nd law.} for interacting
fermions on lattices by using the 2nd law as a postulate (see Remark \ref%
{AC--Conductivity Measure from the 2nd Law}), become more transparent.

As explained in \cite{OhmV}, note however that Lieb--Robinson bounds for
multi--commutators requires short--range interactions. Our setting includes
density--density interactions resulting from the second quantization of
two--body interactions defined via a real--valued and summable (in a
convenient sense) function $v\left( r\right) :[0,\infty )\rightarrow \mathbb{%
R}$. For instance, the celebrated Hubbard model (and any other system with
finite range interactions) or models with Yukawa--type potentials are all
possible choices, but the Coulomb potential is excluded because it is not
summable in space. For more details, see \cite[Section 2.4]{OhmV}.

Our main assertions are Theorems \ref{thm charged transport coefficient}, %
\ref{main 1 copy(8)}, \ref{main 1 copy(1)}, \ref{Theorem AC conductivity
measure} and \ref{Theorem AC conductivity measure copy(2)}. The paper is
organized as follows:

\begin{itemize}
\item Section \ref{Section - 2 law} is a preliminary conceptual review on
the notion of thermal equilibrium state in relation to the 2nd law of
thermodynamics. In this context, the mathematical results of \cite{PW} are
discussed.

\item Section \ref{sect 2.1 copy(1)} formulates the mathematical setting
used to study charge transport properties of fermions. We define in
particular a Banach space of short--range interactions.

\item Section \ref{Sect Classical Ohm's Law copy(1)} states Ohm's law for
macroscopic electromagnetic fields as well as Green--Kubo relations for
current Duhamel fluctuation increments.

\item In Section \ref{Sect Conductivity Measure From Joule's Law} we derive
the macroscopic AC--conducti%
\-%
vity measure from Joule's law and the 2nd law of thermodynamics. Its
relations with microscopic AC--conducti%
\-%
vity measures and the Drude model are discussed. In Section \ref{sect time
ref} we propose a notion of time--reversal symmetry for fermion systems on
the lattice in presence of disorder and discuss its consequences for the
corresponding charge transport coefficients.

\item Section \ref{epilogue} proposes an effective description of the pheno%
\-%
menon of linear conductivity by using L\'{e}vy Processes.

\item Section \ref{Section technical proof Ohm-VI} gathers technical proofs
on which Sections \ref{Sect Classical Ohm's Law copy(1)}--\ref{epilogue} are
based. The arguments strongly use the results of \cite{OhmV,brupedraLR}.
\end{itemize}

\begin{notation}
\label{remark constant copy(1)}\mbox{
}\newline
To simplify notation, we denote by $D$ positive and finite constants. These
constants do not need to be the same from one statement to another. A norm
on a generic vector space $\mathcal{X}$ is denoted by $\Vert \cdot \Vert _{%
\mathcal{X}}$ and the identity map of $\mathcal{X}$ by $\mathbf{1}_{\mathcal{%
X}}$. To avoid ambiguity, scalar products in $\mathcal{X}$ are sometimes
denoted by $\langle \cdot ,\cdot \rangle _{\mathcal{X}}$.
\end{notation}

\section{2nd Law of Thermodynamics and Thermal States\label{Section - 2 law}}

\noindent \textit{It is impossible, by means of inanimate material agency,
to derive mechanical effect from any portion of matter by cooling it below
the temperature of the coldest of the surrounding objects. }\smallskip

\hfill \lbrack Lord Kelvin, 1851]\bigskip

\noindent See \cite{thomson}. This is the celebrated \emph{2nd law of
thermodynamics}, the history of which starts with Carnot's works in 1824. It
is \textquotedblleft \textit{one of the most perfect laws in physics}%
\textquotedblright\ \cite[Section 1]{lieb-yngvasonPhysReport} and it has
never been faulted by reproducible experiments. As explained in \cite%
{lieb-yngvasonPhysReport,lieb-yngvason}, different popular formulations of
the same principle have been stated by Clausius, Kelvin (and Planck), and
Cara%
\-%
th\'{e}odory. Our study is based on the Kelvin--Planck statement while
avoiding the concept of \textquotedblleft cooling\textquotedblright\ \cite[%
p. 49]{lieb-yngvasonPhysReport}: \bigskip

\noindent \textit{No process is possible, the sole result of which is a
change in the energy of a simple system (without changing the work
coordinates) and the raising of a weight. }\bigskip

\noindent The celebrated formulations of Clausius, Kelvin--Planck and Cara%
\-%
th\'{e}odory are all about impossible processes and let largely open what is
possible. This is useful to define the concept of \emph{thermal equilibrium}
states in a simple way. Note that Lieb and Yngvason's work \cite%
{lieb-yngvasonPhysReport} on the 2nd law is an important structural approach
which involves possible processes, instead.

We mathematically implement the Kelvin--Planck principle by using algebraic
quantum mechanics like in \cite{PW}. Basically, we use some $C^{\ast }$%
--algebra $\mathcal{X}$, the self--adjoint elements of which are the
so--called observables of the physical system. States on the $C^{\ast }$%
--algebra $\mathcal{X}$ are, by definition, linear functionals $\rho \in
\mathcal{X}^{\ast }$ which are normalized and positive, i.e., $\rho (\mathbf{%
1})=1$ and $\rho (B^{\ast }B)\geq 0$ for all $B\in \mathcal{X}$. They
represent the state of the physical system. In the commutative case of
classical physics states are usual probability measures.

To define equilibrium states, \cite{PW} is pivotal because it mathematically
implements the Kelvin--Planck physical notion of equilibrium: \bigskip

\noindent \textit{Systems in the equilibrium are unable to perform
mechanical work in cyclic processes.}\bigskip

\noindent Note at this point that the above principle (2nd law) defining
equilibrium can possibly be violated.

As explained in \cite[p. 276]{PW}, the above formulation of the 2nd law of
thermodynamics is directly related to the notion of \emph{passive} states.
Indeed, one defines a (unperturbed) dynamics of the system by a strongly
continuous one--parameter group $\tau \equiv \{\tau _{t}\}_{t\in {\mathbb{R}}%
}$ of $\ast $--automorphisms of $\mathcal{X}$ with (generally unbounded)
generator $\delta $. The latter is a dissipative and closed derivation of $%
\mathcal{X}$. If the state of the system at $t=t_{0}\in {\mathbb{R}}$ is $%
\rho \in \mathcal{X}^{\ast }$, then it evolves as $\rho _{t}=\rho \circ \tau
_{t-t_{0}}$ for any $t\geq t_{0}$. On this system, one produces
\textquotedblleft excitations\textquotedblright\ by perturbing the generator
of dynamics with bounded time--dependant symmetric derivations%
\begin{equation*}
B\mapsto i\left[ A_{t},B\right] :=i\left( A_{t}B-BA_{t}\right) \ ,\qquad
B\in \mathcal{X}\ ,\ t\in \mathbb{R}\ ,
\end{equation*}%
for arbitrary differentiable families $\{A_{t}\}_{t\geq t_{0}}\subset
\mathcal{X}$ of self--adjoint elements of $\mathcal{X}$. In particular, this
defines a strongly continuous two--parameter family $\{\tau
_{t,t_{0}}\}_{t\geq t_{0}}$ of $\ast $--automorphisms of $\mathcal{X}$ as
the solution of a non--autonomous evolution equation defined, for any $B\in
\mathrm{Dom}(\delta )$, by
\begin{equation*}
\forall t_{0},t\in {\mathbb{R}},\ t\geq t_{0}:\quad \partial _{t}\tau
_{t,t_{0}}\left( B\right) =\tau _{t,t_{0}}\left( \delta \left( B\right) +i%
\left[ A_{t},B\right] \right) ,\quad \tau _{t_{0},t_{0}}\left( B\right) :=B\
.
\end{equation*}%
The state of the system evolves now as $\rho _{t}=\rho \circ \tau _{t,t_{0}}$
for any $t\geq t_{0}$.

As explained in \cite[p. 276]{PW}, the energy exchanged between the external
device and the perturbed system at time $t\geq t_{0}$ is equal to
\begin{equation}
L_{t}^{A}\left( \rho \right) :=\int_{t_{0}}^{t}\rho \circ \tau
_{t,t_{0}}\left( \partial _{t}A_{t}\right) \mathrm{d}t\ .  \label{work}
\end{equation}%
If $L_{t}^{A}\left( \rho \right) \geq 0$ then work is performed on the
system, while $L_{t}^{A}\left( \rho \right) <0$ means that one decreases the
energy of the system. A \emph{cyclic} process of time length $T\geq 0$ is,
by definition, a differentiable family $\{A_{t}\}_{t\geq t_{0}}\subset
\mathcal{X}$ of self--adjoint elements of $\mathcal{X}$ such that $A_{t}=0$
for all $t\leq t_{0}$ and $t\geq t_{1}:=T+t_{0}$. Then, the 2nd law of
thermodynamics can be formulated in this mathematical framework as follows
(cf. \cite[Definition 1.1]{PW}):

\begin{definition}[2nd law of thermodynamics -- Passivity]
\label{def second law1}\mbox{ }\newline
Let $(\mathcal{X},\tau )$ be a $C^{\ast }$--dynamical system. A state $\rho
\in \mathcal{X}^{\ast }$ is passive iff $L_{T}^{A}\left( \rho \right) \geq 0$
for all cyclic processes $\{A_{t}\}_{t\geq t_{0}}\subset \mathcal{X}$ of any
time length $T\geq 0$.
\end{definition}

\noindent By \cite[Theorem 2.1]{PW}, passive states $\rho $ of a dynamical
system $(\mathcal{X},\tau )$ can be equivalently defined as states
satisfying
\begin{equation*}
-i\rho (U^{\ast }\delta (U))\geq 0
\end{equation*}%
for all unitaries $U\in \mathcal{X}$ both in the domain of definition of the
generator $\delta $\ of the group $\tau $ and in the connected component of
the identity of the group of all unitary elements of $\mathcal{X}$ with the
norm topology. See, e.g., \cite[Definition 5.3.21]{BratteliRobinson}. This
last condition is strongly related with internal energy increments and the
1st law of thermodynamics, see, e.g., \cite[Theorem 3.2]{OhmV}.

By \cite[Theorem 1.1]{PW}, such states are invariant with respect to
(w.r.t.) the unperturbed dynamics: any passive state $\rho \in \mathcal{X}%
^{\ast }$ satisfies%
\begin{equation*}
\rho =\rho \circ \tau _{t}\ ,\qquad t\in {\mathbb{R}}\ .
\end{equation*}%
Physically, it means that the dynamics of the system at equilibrium \emph{%
cannot} be observed unless one performs external perturbations $%
\{A_{t}\}_{t\geq t_{0}}$ to extract some \emph{excitation spectrum}. This
last notion will be discussed in detail in a companion paper and the
conductivity measure is one notable example of application.

Moreover, for any $\beta \in \mathbb{R}_{0}^{+}$, all $(\tau ,\beta )$--KMS
states $\varrho ^{(\beta )}$ are passive, see \cite[Theorem 1.2]{PW}. The
same holds true for $\beta =\infty $, that is, for ground states of $(%
\mathcal{X},\tau )$. Any convex combination of passive states is also
passive. In particular, for any $n\in \mathbb{N}$, $\beta _{1},\ldots ,\beta
_{n},\mu _{1},\ldots ,\mu _{n}\in \mathbb{R}^{+}$ with $\Sigma _{j=1}^{n}\mu
_{j}=1$, the state
\begin{equation}
\rho =\overset{n}{\sum\limits_{j=1}}\mu _{j}\varrho ^{(\beta _{j})}
\label{combinations}
\end{equation}%
is passive, but it is neither a KMS nor a ground state of $(\mathcal{X},\tau
)$, in general.

We impose another natural condition related to the physical notion \cite[%
Definition p. 55]{lieb-yngvasonPhysReport} of thermal equilibrium in
thermodynamics that excludes such convex combinations. A minimal requirement
for the system to be in thermal equilibrium is indeed that it cannot produce
work by interacting with any of its copy. To be more precise, prepare $n\in
\mathbb{N}$ copies $(\mathcal{X}^{(1)},\tau ^{(1)},\rho ^{(1)}),\ldots ,(%
\mathcal{X}^{(n)},\tau ^{(n)},\rho ^{(n)})$ of the original system defined
by $(\mathcal{X},\tau ,\rho )$ and consider the compound system%
\begin{equation*}
(\otimes _{j=1}^{n}\mathcal{X}^{(j)},\otimes _{j=1}^{n}\tau ^{(j)},\otimes
_{j=1}^{n}\rho ^{(j)})\ .
\end{equation*}%
If $(\mathcal{X},\tau ,\rho )$ is at \emph{thermal equilibrium}, the
compound system should also be at equilibrium and it must not be possible to
extract any energy from cyclic processes, by the 2nd law of thermodynamics.
Therefore, $\otimes _{j=1}^{n}\rho $ should also be passive for all $n\in
\mathbb{N}$. Such states are named in the literature \emph{completely passive%
} states:

\begin{definition}[Thermal equilibrium states]
\label{def second law2}\mbox{ }\newline
Let $(\mathcal{X},\tau )$ be a $C^{\ast }$--dynamical system. A state $%
\varrho \in \mathcal{X}^{\ast }$ is completely passive iff $\otimes
_{j=1}^{n}\varrho $ is a passive state of $(\otimes _{j=1}^{n}\mathcal{X}%
^{(j)},\otimes _{j=1}^{n}\tau ^{(j)})$ for all $n\in \mathbb{N}$. We name
them thermal equilibrium states of $(\mathcal{X},\tau )$.
\end{definition}

\noindent \cite[Theorem 1.4]{PW} gives an explicit characterization of
thermal equilibrium states:

\begin{satz}[Pusz--Woronowicz]
\label{main 1}\mbox{
}\newline
Let $(\mathcal{X},\tau )$ be a $C^{\ast }$--dynamical system. $\varrho $ is
a thermal equilibrium state of $(\mathcal{X},\tau )$ iff it is a $(\tau
,\beta )$--KMS state of $\left( \mathcal{X},\tau \right) $ for some $\beta
\in \left[ 0,\infty \right] $.
\end{satz}

\noindent The parameter $\beta \in \left[ 0,\infty \right] $ is named \emph{%
inverse temperature} of the system and is a \emph{consequence} of the 2nd
law of thermodynamics. It is a universal parameter of the (possibly
infinite) system. In fact, $\beta $ tunes the value of the internal energy
density of the system. Equivalently,\emph{\ it fixes a time scale} since $%
\varrho $ is a $(\tau _{t},\beta )$--KMS state iff $\varrho $ is a $(\tau
_{\beta t},1)$--KMS state ($\beta <\infty $). The boundary case $\beta =0$
corresponds to the $\tau $--invariant traces, also called chaotic states,
whereas $(\tau ,\infty )$--KMS states are by definition ground states. [$%
(\tau ,-\beta )$--KMS states correspond to $(\tau ,\beta )$--KMS states with
a reversal of time.]

The notion of local (relative) entropy seems to be more natural than the
concept of local temperature. Indeed, the 2nd law of thermodynamics as
expressed in Definitions \ref{def second law1}--\ref{def second law2} is a
formal expression of the unavoidable lost while one interacts with an
object, which is at equilibrium before the interaction. Entropy is only a
quantitative counterpart of this lost. It corresponds to heat production in
thermodynamics which we study in the context of electricity theory. The
positivity of the heat production, which is the content of the 2nd law of
thermodynamics, implies the existence of the AC--conductivity measure. See
Section \ref{Sect Conductivity Measure From Joule's Law}.

\begin{bemerkung}[Dynamics versus thermal equilibrium\ states]
\label{remark conjecture}\mbox{
}\newline
Let a state $\varrho \in \mathcal{X}^{\ast }$ with GNS representation $(%
\mathcal{H},\pi ,\Psi )$. Its normal extension $\hat{\varrho}$ on $\pi (\mathcal{X})^{\prime
\prime }$ is a KMS state for a $\sigma $--weakly continuous one--parameter
group $\tau \equiv \{\tau _{t}\}_{t\in {\mathbb{R}}}$ of $\ast $%
--automorphisms of $\pi (\mathcal{X})^{\prime \prime }$ iff $\hat{\varrho}$ is faithful. 
See, e.g., \cite[p. 85]{BratteliRobinson}. In this
case, the group $\tau $ is unique. The faithfulness of states is a
physically natural property: By definition, an observable exists iff the
corresponding physical property can be observed. Therefore, one could fix a
state $\varrho \in \mathcal{X}^{\ast }$ of the system that must be, by
definition, a thermal equilibrium state, i.e., a KMS state. This assumption
implicitly imposes the existence of some (unique) dynamics given by a group $%
\tau ^{(\varrho )}$ and is justified a posteriori via the 2nd law.
Constructing KMS states $\varrho ^{(\tau )}$ from a given dynamics $\tau $
may be technically more involved. It is however the approach we use because
the dynamics is fixed by microscopic interactions between particles.
\end{bemerkung}

\section{$C^{\ast }$--Dynamical Systems for Interacting Fermions\label{sect
2.1 copy(1)}\label{Section main results}}

The mathematical framework used here is exactly the one of \cite{OhmV}. It
is concisely described below. The only additional information is the exact
definition of the probability space modelling disorder.

\subsection{Disordered Media within Electromagnetic Fields\label{sect 2.1}}

Disorder in the crystal is modeled by a random variable with distribution $%
\mathfrak{a}_{\Omega }$ taking values in the measurable space $(\Omega ,%
\mathfrak{A}_{\Omega })$. The probability space $(\Omega ,\mathfrak{A}%
_{\Omega },\mathfrak{a}_{\Omega })$ is defined as follows:

\begin{itemize}
\item[$\Omega :$] Let $\mathfrak{L}:=\mathbb{Z}^{d}$ ($d\in \mathbb{N}$) and
\begin{equation}
\mathfrak{b}:=\left\{ \{x,x^{\prime }\}\subset \mathfrak{L}\text{ }:\text{ }%
|x-x^{\prime }|=1\right\}  \label{n-o bonds}
\end{equation}%
be the set of non--oriented bonds of the cubic lattice $\mathfrak{L}$. Then,
\begin{equation*}
\Omega :=[-1,1]^{\mathfrak{L}}\times \mathbb{D}^{\mathfrak{b}}\text{\quad
with\quad }\mathbb{D}:=\{z\in \mathbb{C}:\left\vert z\right\vert \leq 1\}\ .
\end{equation*}%
I.e., any element of $\Omega $ is a pair $\omega =\left( \omega _{1},\omega
_{2}\right) \in \Omega $, where $\omega _{1}$ is a function on lattice sites
with values in $[-1,1]$ and $\omega _{2}$ is a function on bonds with values
in the complex closed unit disc $\mathbb{D}$.

\item[$\mathfrak{A}_{\Omega }:$] Let $\Omega _{x}^{(1)}$, $x\in \mathfrak{L}$%
, be an arbitrary element of the Borel $\sigma $--algebra $\mathfrak{A}%
_{x}^{(1)}$ of the interval $[-1,1]$ w.r.t. the usual metric topology.
Define
\begin{equation*}
\mathfrak{A}_{[-1,1]^{\mathfrak{L}}}:=\bigotimes\limits_{x\in \mathfrak{L}}%
\mathfrak{A}_{x}^{(1)}\ ,
\end{equation*}
i.e., $\mathfrak{A}_{[-1,1]^{\mathfrak{L}}}$ is the $\sigma $--algebra
generated by the cylinder sets $\prod\nolimits_{x\in \mathfrak{L}}\Omega
_{x}^{(1)}$, where $\Omega _{x}^{(1)}=[-1,1]$ for all but finitely many $%
x\in \mathfrak{L}$. In the same way, let
\begin{equation*}
\mathfrak{A}_{\mathbb{D}^{\mathfrak{b}}}:=\bigotimes\limits_{\mathbf{x}\in
\mathfrak{b}}\mathfrak{A}_{\mathbf{x}}^{(2)}\ ,
\end{equation*}%
where $\mathfrak{A}_{\mathbf{x}}^{(2)}$, $\mathbf{x}\in \mathfrak{b}$, is
the Borel $\sigma $--algebra of the complex closed unit disc $\mathbb{D}$
w.r.t. the usual metric topology. Then
\begin{equation*}
\mathfrak{A}_{\Omega }:=\mathfrak{A}_{[-1,1]^{\mathfrak{L}}}\otimes
\mathfrak{A}_{\mathbb{D}^{\mathfrak{b}}}\ .
\end{equation*}

\item[$\mathfrak{a}_{\Omega }:$] The measure $\mathfrak{a}_{\Omega }$ is an
arbitrary \emph{ergodic} probability measure on the measurable space $%
(\Omega ,\mathfrak{A}_{\Omega })$: It is invariant under the action%
\begin{equation}
\left( \omega _{1},\omega _{2}\right) \longmapsto \chi _{x}^{(\Omega
)}\left( \omega _{1},\omega _{2}\right) :=\left( \chi _{x}^{(\mathfrak{L}%
)}\left( \omega _{1}\right) ,\chi _{x}^{(\mathfrak{b})}\left( \omega
_{2}\right) \right) \ ,\qquad x\in \mathbb{Z}^{d}\ ,
\label{translation omega}
\end{equation}%
of the group $(\mathbb{Z}^{d},+)$ of translations on $\Omega $ and, for any $%
\mathcal{X}\in \mathfrak{A}_{\Omega }$ such that $\chi _{x}^{(\Omega
)}\left( \mathcal{X}\right) =\mathcal{X}$ for all $x\in \mathbb{Z}^{d}$, one
has $\mathfrak{a}_{\Omega }(\mathcal{X})\in \{0,1\}$. Here, for any $\omega
=\left( \omega _{1},\omega _{2}\right) \in \Omega $, $x\in \mathbb{Z}^{d}$
and $y,y^{\prime }\in \mathfrak{L}$ with $|y-y^{\prime }|=1$,%
\begin{equation}
\chi _{x}^{(\mathfrak{L})}\left( \omega _{1}\right) \left( y\right) :=\omega
_{1}\left( y+x\right) \ ,\ \chi _{x}^{(\mathfrak{b})}\left( \omega
_{2}\right) \left( \{y,y^{\prime }\}\right) :=\omega _{2}\left(
\{y+x,y^{\prime }+x\}\right) \ .  \label{translation omegabis}
\end{equation}%
We denote by $\mathbb{E}[\ \cdot \ ]$ the expectation value associated with $%
\mathfrak{a}_{\Omega }$.
\end{itemize}

For any $\omega =\left( \omega _{1},\omega _{2}\right) \in \Omega $, $%
V_{\omega }\in \mathcal{B}(\ell ^{2}(\mathfrak{L}))$ is by definition the
self--adjoint multiplication operator with the function $\omega _{1}:%
\mathfrak{L}\rightarrow \lbrack -1,1]$. It represents a bounded static
potential. To all $\omega \in \Omega $ and strength $\vartheta \in \mathbb{R}%
_{0}^{+}$ of hopping disorder, we also associate another self--adjoint
operator $\Delta _{\omega ,\vartheta }\in \mathcal{B}(\ell ^{2}(\mathfrak{L}%
))$ describing the hoppings of a single particle in the lattice:
\begin{eqnarray}
\lbrack \Delta _{\omega ,\vartheta }(\psi )](x) &:=&2d\psi (x)-\sum_{j=1}^{d}\Big((1+\vartheta \overline{\omega _{2}(\{x,x-e_{j}\})})\ \psi (x-e_{j})
\notag \\
&&+\psi (x+e_{j})(1+\vartheta \omega _{2}(\{x,x+e_{j}\}))\Big)
\label{discrete laplacian}
\end{eqnarray}%
for any $x\in \mathfrak{L}$ and $\psi \in \ell ^{2}(\mathfrak{L})$, with $%
\{e_{k}\}_{k=1}^{d}$ being the canonical orthonormal basis of the Euclidian
space $\mathbb{R}^{d}$. In the case of vanishing hopping disorder $\vartheta
=0$ (up to a minus sign) $\Delta _{\omega ,0}$ is the usual $d$--dimensional
discrete Laplacian. Since the hopping amplitudes are complex--valued ($%
\omega _{2}$ takes values in $\mathbb{D}$), note additionally that random
electromagnetic potentials can be implemented in our model.

Then, for any realization $\omega \in \Omega $ of disorder and parameters $%
\vartheta ,\lambda \in \mathbb{R}_{0}^{+}$, the Hamiltonian of a single
quantum particle within a bounded static potential is the discrete Schr\"{o}%
dinger operator $(\Delta _{\omega ,\vartheta }+\lambda V_{\omega })$ acting
on the Hilbert space $\ell ^{2}(\mathfrak{L})$. The coupling constants $%
\vartheta ,\lambda \in \mathbb{R}_{0}^{+}$ represent the strength of
disorder of respectively the external static potential and hopping
amplitudes.

The time--dependent electromagnetic potential is defined by a compactly
supported time--depen%
\-%
dent vector potential%
\begin{equation*}
\mathbf{A}\in \mathbf{C}_{0}^{\infty }:=\underset{l\in \mathbb{R}^{+}}{%
\mathop{\displaystyle \bigcup}}C_{0}^{\infty }(\mathbb{R}\times \left[ -l,l%
\right] ^{d};({\mathbb{R}}^{d})^{\ast })\ ,
\end{equation*}%
where $({\mathbb{R}}^{d})^{\ast }$ is the set of one--forms\footnote{%
In a strict sense, one should take the dual space of the tangent spaces $T({%
\mathbb{R}}^{d})_{x}$, $x\in {\mathbb{R}}^{d}$.} on ${\mathbb{R}}^{d}$ that
take values in $\mathbb{R}$. The smoothness of $\mathbf{A}$ is not essential
in the proofs and is only assumed for simplicity.

\begin{bemerkung}
To simplify notation, we identify in the sequel $(\mathbb{R}^{d})^{\ast }$
with $\mathbb{R}^{d}$ via the canonical scalar product of $\mathbb{R}^{d}$.
\end{bemerkung}

We use the Weyl gauge (also named temporal gauge) for the electromagnetic
field and, as a consequence,%
\begin{equation}
E_{\mathbf{A}}(t,x):=-\partial _{t}\mathbf{A}(t,x)\ ,\quad t\in \mathbb{R},\
x\in \mathbb{R}^{d}\ ,  \label{V bar 0}
\end{equation}%
is the electric field associated with $\mathbf{A}$. We also define the
integrated electric field (or electric tension) along the oriented bond $%
\mathbf{x}:=(x^{(1)},x^{(2)})\in \mathfrak{L}^{2}$ at time $t\in \mathbb{R}$
by%
\begin{equation}
\mathbf{E}_{t}^{\mathbf{A}}\left( \mathbf{x}\right) :=\int\nolimits_{0}^{1}%
\left[ E_{\mathbf{A}}(t,\alpha x^{(2)}+(1-\alpha )x^{(1)})\right]
(x^{(2)}-x^{(1)})\mathrm{d}\alpha \ .  \label{V bar 0bis}
\end{equation}%
Since $\mathbf{A}$ is by assumption compactly supported, the corresponding
electric field satisfies the AC--condition%
\begin{equation}
\int\nolimits_{t_{0}}^{t}E_{\mathbf{A}}(s,x)\mathrm{d}s=0\ ,\quad x\in
\mathbb{R}^{d}\ ,  \label{zero mean field}
\end{equation}%
for sufficiently large times $t\geq t_{1}\geq t_{0}$. From (\ref{zero mean
field}),%
\begin{equation}
t_{1}:=\min \left\{ t\geq t_{0}:\quad \int\nolimits_{t_{0}}^{t^{\prime }}E_{%
\mathbf{A}}(s,x)\mathrm{d}s=0\quad \text{for all }x\in \mathbb{R}^{d}\text{
and }t^{\prime }\geq t\right\}  \label{zero mean field assumption}
\end{equation}%
is the time at which the electric field is turned off. In other words, we
consider \emph{cyclic} electromagnetic processes.

To simplify notation and without loss of generality (w.l.o.g.), fermions are
spinless and have negative charge. The cases of particles with spin and
positively charged particles can be treated by exactly the same methods.
Thus, using the (minimal) coupling of $\mathbf{A}\in \mathbf{C}_{0}^{\infty
} $ to the discrete Laplacian, the discrete magnetic Laplacian is (up to a
minus sign) the self--adjoint operator
\begin{equation*}
\Delta _{\omega ,\vartheta }^{(\mathbf{A})}\equiv \Delta _{\omega ,\vartheta
}^{(\mathbf{A}(t,\cdot ))}\in \mathcal{B}(\ell ^{2}(\mathfrak{L}))\ ,\qquad
t\in \mathbb{R}\ ,
\end{equation*}%
defined\footnote{%
Observe that the sign of the coupling between the electromagnetic potential $%
\mathbf{A}\in \mathbf{C}_{0}^{\infty }$ and the laplacian is wrong in \cite[%
Eq. (2.8)]{OhmI}.} by%
\begin{equation}
\langle \mathfrak{e}_{x},\Delta _{\omega ,\vartheta }^{(\mathbf{A})}%
\mathfrak{e}_{y}\rangle =\exp \left( i\int\nolimits_{0}^{1}\left[ \mathbf{A}%
(t,\alpha y+(1-\alpha )x)\right] (y-x)\mathrm{d}\alpha \right) \langle
\mathfrak{e}_{x},\Delta _{\omega ,\vartheta }\mathfrak{e}_{y}\rangle
\label{eq discrete lapla A}
\end{equation}%
for all $t\in \mathbb{R}$, $\omega \in \Omega $, $\vartheta \in \mathbb{R}%
_{0}^{+}$ and $x,y\in \mathfrak{L}$. Here, $\langle \cdot ,\cdot \rangle $
is the scalar product in $\ell ^{2}(\mathfrak{L})$ and $\left\{ \mathfrak{e}%
_{x}\right\} _{x\in \mathfrak{L}}$ is the canonical orthonormal basis $%
\mathfrak{e}_{x}(y)\equiv \delta _{x,y}$ of $\ell ^{2}(\mathfrak{L})$. In (%
\ref{eq discrete lapla A}), similar to (\ref{V bar 0bis}), $\alpha
y+(1-\alpha )x$ and $y-x$ are seen as vectors in ${\mathbb{R}}^{d}$. In
presence of an electromagnetic field associated to an arbitrary vector
potential $\mathbf{A}\in \mathbf{C}_{0}^{\infty }$, the one--particle
Hamiltonian $(\Delta _{\omega ,\vartheta }+\lambda V_{\omega })$ at fixed $%
\omega \in \Omega $ and $\vartheta ,\lambda \in \mathbb{R}_{0}^{+}$ is
replaced with the time--dependent one%
\begin{equation}
\Delta _{\omega ,\vartheta }^{(\mathbf{A})}+\lambda V_{\omega }\equiv \Delta
_{\omega ,\vartheta }^{(\mathbf{A}(t,\cdot ))}+\lambda V_{\omega }\ ,\qquad
t\in \mathbb{R}\ .  \label{hamiltonian libre total}
\end{equation}

\subsection{Banach Space of Short--Range Interactions\label{Section Banach
space interaction}}

Let $\mathcal{P}_{f}(\mathfrak{L})\subset 2^{\mathfrak{L}}$ be the set of
all finite subsets of $\mathfrak{L}$. For all $\Lambda \in \mathcal{P}_{f}(%
\mathfrak{L})$, $\mathcal{U}_{\Lambda }$ is the finite dimensional $C^{\ast
} $--algebra generated by $\mathbf{1}$ and generators $\{a_{x,\mathrm{s}%
}\}_{x\in \Lambda ,\mathrm{s}\in \mathrm{S}}$ satisfying the canonical
anti--commutation relations, $\mathrm{S}$ being some finite set of spins. As
just explained above, the spin dependence of $a_{x,\mathrm{s}}\equiv a_{x}$
is irrelevant in our proofs (up to trivial modifications) and, w.l.o.g., we
only consider spinless fermions, i.e., the case $\mathrm{S}=\{0\}$.

We denote by $\mathcal{U}$ the CAR $C^{\ast }$--algebra $\mathcal{U}$ of the
infinite system, that is, the inductive limit of the finite dimensional $%
C^{\ast }$--algebras $\{\mathcal{U}_{\Lambda }\}_{\Lambda \in \mathcal{P}%
_{f}(\mathfrak{L})}$. The $C^{\ast }$--algebra of all even elements of $%
\mathcal{U}$ is denoted by $\mathcal{U}^{+}$ and
\begin{equation*}
\mathcal{U}_{0}:=\underset{\Lambda \in \mathcal{P}_{f}(\mathfrak{L})}{%
\bigcup }\mathcal{U}_{\Lambda }\subset \mathcal{U}
\end{equation*}%
is the subset of local elements. See \cite[Section 2.2]{OhmV} for more
details. Finally, let $\{\chi _{x}\}_{x\in \mathfrak{L}}$ be the family of $%
\ast $--automor%
\-%
phisms of $\mathcal{U}$ uniquely defined by the conditions
\begin{equation}
\chi _{x}(a_{y})=a_{y+x}\ ,\qquad x,y\in \mathfrak{L}=\mathbb{Z}^{d}\ .
\label{definition trans U}
\end{equation}

An \emph{interaction} is a family $\Phi =\{\Phi _{\Lambda }\}_{\Lambda \in
\mathcal{P}_{f}(\mathfrak{L})}$ of even and self--adjoint local elements $%
\Phi _{\Lambda }=\Phi _{\Lambda }^{\ast }\in \mathcal{U}^{+}\cap \mathcal{U}%
_{\Lambda }$ with $\Phi _{\emptyset }=0$. We define Banach spaces $\mathcal{W%
}$ of short--range interactions by introducing norms that take into account
space decay of interactions. To this end, we use positive--valued and
non--increasing decay functions $\mathbf{F}:\mathbb{R}_{0}^{+}\rightarrow
\mathbb{R}^{+}$. Like in \cite{OhmV}, we impose the following conditions on $%
\mathbf{F}$:

\begin{itemize}
\item \emph{Summability on }$\mathfrak{L}$\emph{.}
\begin{equation}
\left\Vert \mathbf{F}\right\Vert _{1,\mathfrak{L}}:=\underset{y\in \mathfrak{%
L}}{\sup }\sum_{x\in \mathfrak{L}}\mathbf{F}\left( \left\vert x-y\right\vert
\right) =\sum_{x\in \mathfrak{L}}\mathbf{F}\left( \left\vert x\right\vert
\right) <\infty \ .  \label{(3.1) NS}
\end{equation}

\item \emph{Bounded convolution constant.}
\begin{equation}
\mathbf{D}:=\underset{x,y\in \mathfrak{L}}{\sup }\sum_{z\in \mathfrak{L}}%
\frac{\mathbf{F}\left( \left\vert x-z\right\vert \right) \mathbf{F}\left(
\left\vert z-y\right\vert \right) }{\mathbf{F}\left( \left\vert
x-y\right\vert \right) }<\infty \ .  \label{(3.2) NS}
\end{equation}
\end{itemize}

\noindent Examples of functions $\mathbf{F}:\mathbb{R}_{0}^{+}\rightarrow
\mathbb{R}^{+}$ satisfying (\ref{(3.1) NS})--(\ref{(3.2) NS}) are given by
\begin{equation}
\mathbf{F}\left( r\right) =\left( 1+r\right) ^{-(d+\epsilon )}\quad \text{and%
}\quad \mathbf{F}\left( r\right) =\mathrm{e}^{-\varsigma
r}(1+r)^{-(d+\epsilon )}  \label{example polynomial}
\end{equation}%
for any $\varsigma ,\epsilon \in \mathbb{R}^{+}$. In all the paper, (\ref%
{(3.1) NS})--(\ref{(3.2) NS}) are assumed to be satisfied.

Then, the norm of any interaction $\Phi $ is defined by
\begin{equation}
\left\Vert \Phi \right\Vert _{\mathcal{W}}:=\underset{x,y\in \mathfrak{L}}{%
\sup }\sum\limits_{\Lambda \in \mathcal{P}_{f}(\mathfrak{L}),\;\Lambda
\supset \{x,y\}}\frac{\Vert \Phi _{\Lambda }\Vert _{\mathcal{U}}}{\mathbf{F}%
\left( \left\vert x-y\right\vert \right) }\ .  \label{iteration0}
\end{equation}%
The real separable Banach space $(\mathcal{W},\left\Vert \cdot \right\Vert _{%
\mathcal{W}})$ is the space of interactions $\Phi $ with $\left\Vert \Phi
\right\Vert _{\mathcal{W}}<\infty $. Elements $\Phi \in \mathcal{W}$ are
named \emph{short--range} interactions.

\subsection{Interacting Fermion Systems in Disordered Media\label{Section
dynamics}}

To any $\omega \in \Omega $ and strength $\vartheta \in \mathbb{R}_{0}^{+}$
of hopping disorder, we associate a short--range interaction $\Psi ^{(\omega
,\vartheta )}\in \mathcal{W}$ defined as follows: Fix an interparticle (IP)
interaction $\Psi ^{\mathrm{IP}}\in \mathcal{W}$. Then,
\begin{equation*}
\Psi _{\Lambda }^{(\omega ,\vartheta )}:=\langle \mathfrak{e}_{x},\Delta
_{\omega ,\vartheta }\mathfrak{e}_{y}\rangle a_{x}^{\ast }a_{y}+\langle
\mathfrak{e}_{y},\Delta _{\omega ,\vartheta }\mathfrak{e}_{x}\rangle
a_{y}^{\ast }a_{x}+\Psi _{\{x,y\}}^{\mathrm{IP}}\in \mathcal{U}^{+}\cap
\mathcal{U}_{\Lambda }
\end{equation*}%
whenever $\Lambda =\left\{ x,y\right\} $ for $x,y\in \mathfrak{L}$, and $%
\Psi _{\Lambda }^{(\omega ,\vartheta )}:=\Psi _{\Lambda }^{\mathrm{IP}}$
otherwise.

Let
\begin{equation}
\Lambda _{l}:=\{(x_{1},\ldots ,x_{d})\in \mathfrak{L}\,:\,|x_{1}|,\ldots
,|x_{d}|\leq l\}\ ,\qquad l\in \mathbb{R}_{0}^{+}\ .  \label{eq:def lambda n}
\end{equation}%
We then assume two additional properties of $\Psi ^{\mathrm{IP}}$:

\begin{itemize}
\item \emph{Translation invariance}. For all $x\in \mathfrak{L}$,%
\begin{equation}
\Psi _{\Lambda +x}^{\mathrm{IP}}=\chi _{x}\left( \Psi _{\Lambda }^{\mathrm{IP%
}}\right) \ ,\qquad \Lambda \in \mathcal{P}_{f}(\mathfrak{L})\ .
\label{static potential0}
\end{equation}

\item \emph{Polynomial decay.} There is a constant $\varsigma >2d$ and, for
all $m\in \mathbb{N}_{0}$, an absolutely summable sequence $\{\mathbf{u}%
_{n,m}\}_{n\in \mathbb{N}}\in \ell ^{1}(\mathbb{N})$ such that, for all $%
n\in \mathbb{N}$ with $n>m$,%
\begin{equation}
|\Lambda _{n}\backslash \Lambda _{n-1}|\sum_{z\in \Lambda _{m}}\max_{y\in
\Lambda _{n}\backslash \Lambda _{n-1}}\mathbf{F}\left( \left\vert
z-y\right\vert \right) \leq \frac{\mathbf{u}_{n,m}}{\left( 1+n\right)
^{\varsigma }}\text{ }.  \label{(3.3) NS}
\end{equation}
\end{itemize}

\noindent Examples of functions $\mathbf{F}:\mathbb{R}_{0}^{+}\rightarrow
\mathbb{R}^{+}$ satisfying (\ref{(3.1) NS})--(\ref{(3.2) NS}) and (\ref%
{(3.3) NS}) are obviously given by (\ref{example polynomial}), for
sufficiently large $\epsilon \in \mathbb{R}^{+}$ in the polynomial case.

Conditions (\ref{(3.1) NS})--(\ref{(3.2) NS}) and \cite[Theorem 2.2]{OhmV}
ensure the existence of a (non--autonomous) infinite volume dynamics $\{\tau
_{t,s}^{(\omega ,\vartheta ,\lambda ,\mathbf{A})}\}_{s,t\in {\mathbb{R}}}$,
in presence of electromagnetic fields and static potentials (cf. (\ref%
{hamiltonian libre total})). Indeed, any realization $\omega \in \Omega $,
disorder strengths $\vartheta ,\lambda \in \mathbb{R}_{0}^{+}$,
interparticle interaction $\Psi ^{\mathrm{IP}}\in \mathcal{W}$ and
electromagnetic potential $\mathbf{A}\in \mathbf{C}_{0}^{\infty }$ naturally
define a family $\{\delta _{t}^{(\omega ,\vartheta ,\lambda ,\mathbf{A}%
)}\}_{t\in {\mathbb{R}}}$ of derivations on the subset $\mathcal{U}_{0}$ of
local elements of $\mathcal{U}$. Then, $\{\tau _{t,s}^{(\omega ,\vartheta
,\lambda ,\mathbf{A})}\}_{s,t\in {\mathbb{R}}}$ is the unique strongly
continuous two--para%
\-%
meter family of $\ast $--automor%
\-%
phisms of $\mathcal{U}$ satisfying, in the strong sense on the dense domain $%
\mathcal{U}_{0}\subset \mathcal{U}$,%
\begin{equation*}
\forall s,t\in {\mathbb{R}}:\qquad \partial _{t}\tau _{t,s}^{(\omega
,\vartheta ,\lambda ,\mathbf{A})}=\tau _{t,s}^{(\omega ,\vartheta ,\lambda ,%
\mathbf{A})}\circ \delta _{t}^{(\omega ,\vartheta ,\lambda ,\mathbf{A})}\
,\qquad \tau _{s,s}^{(\omega ,\vartheta ,\lambda ,\mathbf{A})}=\mathbf{1}_{%
\mathcal{U}}\ .
\end{equation*}%
See \cite[Section 2.5]{OhmV} for more details. At $\mathbf{A}=0$, the
(unperturbed) dynamics is autonomous and we denote the corresponding group
of $\ast $--automor%
\-%
phisms by
\begin{equation}
\tau ^{(\omega ,\vartheta ,\lambda )}:=\{\tau _{t}^{(\omega ,\vartheta
,\lambda )}\}_{t\in {\mathbb{R}}}\ .  \label{tho}
\end{equation}

Then, as explained in Section \ref{Section - 2 law}, thermal equilibrium
states are defined to be completely passive states, see Definition \ref{def
second law2}. This definition is based on the 2nd law of thermodynamics. By
Theorem \ref{main 1}, they are $(\tau ^{(\omega ,\vartheta ,\lambda )},\beta
)$--KMS states for some inverse temperature, or time scale, $\beta \in \left[
0,\infty \right] $. For simplicity, we exclude the boundary cases $\beta
=0,+\infty $. As discussed in \cite[Section 2.6]{OhmV}, the set of $(\tau
^{(\omega ,\vartheta ,\lambda )},\beta )$--KMS states is non--empty for any $%
\beta \in \mathbb{R}^{+}$, $\omega \in \Omega $ and $\vartheta ,\lambda \in
\mathbb{R}_{0}^{+}$. Here, $\varrho ^{(\beta ,\omega ,\vartheta ,\lambda )}$
denotes one element of this set.

For any $\beta \in \mathbb{R}^{+}$ and $\vartheta ,\lambda \in \mathbb{R}%
_{0}^{+}$, we impose two natural conditions on the map%
\begin{equation}
\omega \mapsto \varrho ^{(\beta ,\omega ,\vartheta ,\lambda )}  \label{map}
\end{equation}%
from the set $\Omega $ to the dual space $\mathcal{U}^{\ast }$:

\begin{itemize}
\item \emph{Translation invariance.} Recall that $\{\chi _{x}\}_{x\in
\mathfrak{L}}$ is the family of $\ast $--automor%
\-%
phisms of $\mathcal{U}$ uniquely defined by (\ref{definition trans U}). It
implements the action of the group $(\mathbb{Z}^{d},+)$ of lattice
translations on the CAR $C^{\ast }$--algebra $\mathcal{U}$. On the set $%
\Omega $ this action is represented by the family $\{\chi _{x}^{(\Omega
)}\}_{x\in \mathfrak{L}}$, see (\ref{translation omega})--(\ref{translation
omegabis}). Then, we assume that
\begin{equation}
\varrho ^{(\beta ,\chi _{x}^{(\Omega )}(\omega ),\vartheta ,\lambda
)}=\varrho ^{(\beta ,\omega ,\vartheta ,\lambda )}\circ \chi _{x}\ ,\qquad
x\in \mathfrak{L}=\mathbb{Z}^{d}\ .  \label{translation invariant}
\end{equation}

\item \emph{Measurability.} Thermal equilibrium states are supposed to be
random variables. Hence, for any $\beta \in \mathbb{R}^{+}$ and $\vartheta
,\lambda \in \mathbb{R}_{0}^{+}$, we assume that the map (\ref{map}) is
measurable w.r.t. to the $\sigma $--algebra $\mathfrak{A}_{\Omega }$ on $%
\Omega $ and the Borel $\sigma $--algebra $\mathfrak{A}_{\mathcal{U}^{\ast
}} $ of $\mathcal{U}^{\ast }$ generated by the weak$^{\ast }$--topology.
Observe that a similar assumption is also necessary for classical disordered
systems at equilibrium, see, e.g., \cite{bovier}.
\end{itemize}

\noindent These conditions yield the following definition:

\begin{definition}[Random invariant states]
\label{def second law2 copy(1)}\mbox{ }\newline
Let $\omega \mapsto \varrho ^{(\omega )}$ be a map from $\Omega $ to the set
of states on $\mathcal{U}$. We say that this map is a random invariant state
when it is measurable w.r.t. to $\mathfrak{A}_{\Omega }$ and $\mathfrak{A}_{%
\mathcal{U}^{\ast }}$ and translation invariant in the above sense.
\end{definition}

\noindent The map (\ref{map}) is thus a random invariant state. This implies
in particular that, for any $\beta \in \mathbb{R}^{+}$ and $\vartheta
,\lambda \in \mathbb{R}_{0}^{+}$, the averaged state $\bar{\varrho}^{(\beta
,\lambda )}\in \mathcal{U}^{\ast }$ defined by%
\begin{equation}
\bar{\varrho}^{(\beta ,\vartheta ,\lambda )}\left( B\right) :=\mathbb{E}%
\left[ \varrho ^{(\beta ,\omega ,\vartheta ,\lambda )}\left( B\right) \right]
\ ,\qquad B\in \mathcal{U}\ ,  \label{translation invariantbis0}
\end{equation}%
is \emph{translation invariant}, i.e.,
\begin{equation}
\bar{\varrho}^{(\beta ,\vartheta ,\lambda )}=\bar{\varrho}^{(\beta
,\vartheta ,\lambda )}\circ \chi _{x}\ ,\qquad x\in \mathfrak{L}=\mathbb{Z}%
^{d}\ .  \label{translation invariantbis}
\end{equation}%
Recall indeed that $\mathfrak{a}_{\Omega }$ is also a translation invariant
probability measure. [It is even ergodic.]

The existence of such random invariant equilibrium states is not completely
clear in general, similar to the classical case. If the $(\tau ^{(\omega
,\vartheta ,\lambda )},\beta )$--KMS state is unique and (\ref{static
potential0}) is satisfied, then it turns out that the (unique) map (\ref{map}%
) is a random invariant state. Indeed, in this case, the map (\ref{map}) is
even continuous w.r.t. the pointwise convergence in $\Omega $ and the weak$%
^{\ast }$--topology of $\mathcal{U}^{\ast }$. This can be proven by using
\cite[Proposition 5.3.23.]{BratteliRobinson}. Uniqueness of KMS states
appears for instance when either $\Psi ^{\mathrm{IP}},\vartheta =0$, or at
small $\beta \in \mathbb{R}^{+}$, or in dimension $d=1$. Moreover, by using
methods of constructive quantum field theory, one can also verify the
existence of such random invariant thermal equilibrium states at arbitrary $%
\beta \in \mathbb{R}^{+}$ and dimension $d\in \mathbb{N}$ if the
interparticle interaction $\left\Vert \Psi ^{\mathrm{IP}}\right\Vert _{%
\mathcal{W}}$ is small enough and (\ref{static potential0}) holds.

Now, in presence of electromagnetic fields, the time evolution of the state
of the system equals%
\begin{equation}
\rho _{t}^{(\beta ,\omega ,\vartheta ,\lambda ,\mathbf{A})}:=\left\{
\begin{array}{lll}
\varrho ^{(\beta ,\omega ,\vartheta ,\lambda )} & , & \qquad t\leq t_{0}\ ,
\\
\varrho ^{(\beta ,\omega ,\vartheta ,\lambda )}\circ \tau
_{t,t_{0}}^{(\omega ,\vartheta ,\lambda ,\mathbf{A})} & , & \qquad t\geq
t_{0}\ ,%
\end{array}%
\right.  \label{time dependent state}
\end{equation}%
for any $\beta \in \mathbb{R}^{+}$, $\omega \in \Omega $, $\vartheta
,\lambda \in \mathbb{R}_{0}^{+}$ and $\mathbf{A}\in \mathbf{C}_{0}^{\infty }$%
. Recall here that $\mathbf{A}(t,x)=0$ for all $t\leq t_{0}$.

\begin{bemerkung}[Time--dependent states as stochastic processes]
\label{remark stochastic}\mbox{
}\newline
Under the above assumptions, by using Lieb--Robinson bounds as in \cite[%
Lemma 4.3]{brupedraLR}, it is possible to show that the family $\{\omega
\mapsto \rho _{t}^{(\beta ,\omega ,\vartheta ,\lambda ,\mathbf{A})}\}_{t\in
\mathbb{R}}$ defines a stochastic process with values in $\mathcal{U}^{\ast
} $. More precisely, for any $t\in \mathbb{R}$, the map $\omega \mapsto \rho
_{t}^{(\beta ,\omega ,\vartheta ,\lambda ,\mathbf{A})}$ is measurable w.r.t.
to $\mathfrak{A}_{\Omega }$ and $\mathfrak{A}_{\mathcal{U}^{\ast }}$. This
fact is not essential in the sequel.
\end{bemerkung}

\section{Macroscopic Ohm's Law and Green--Kubo Relations\label{Sect
Classical Ohm's Law copy(1)}}

\subsection{Macroscopic Charge Transport Coefficients}

Fix $\omega \in \Omega $, $\vartheta \in \mathbb{R}_{0}^{+}$, $\mathbf{A}\in
\mathbf{C}_{0}^{\infty }$ and time $t\in \mathbb{R}$. For any oriented bond $%
\mathbf{x}:=(x^{(1)},x^{(2)})\in \mathfrak{L}^{2}$, we define the
paramagnetic\ and diamagnetic current observables $I_{\mathbf{x}}^{(\omega
,\vartheta )}$ and $\mathrm{I}_{\mathbf{x}}^{(\omega ,\vartheta ,\mathbf{A}%
)} $ respectively by%
\begin{equation}
I_{\mathbf{x}}^{(\omega ,\vartheta )}:=-2\mathrm{Im}\left( \langle \mathfrak{%
e}_{x^{(1)}},\Delta _{\omega ,\vartheta }\mathfrak{e}_{x^{(2)}}\rangle
a_{x^{(1)}}^{\ast }a_{x^{(2)}}\right) \ ,\qquad \mathbf{x}%
:=(x^{(1)},x^{(2)})\in \mathfrak{L}^{2}\ .  \label{current observable}
\end{equation}%
and%
\begin{eqnarray}
\mathrm{I}_{\mathbf{x}}^{(\omega ,\vartheta ,\mathbf{A})}\equiv \mathrm{I}_{%
\mathbf{x}}^{(\omega ,\vartheta ,\mathbf{A}(t,\cdot ))} &:=&-2\mathrm{Im}%
\Big(\Big(\mathrm{e}^{-i\int\nolimits_{0}^{1}[\mathbf{A}(t,\alpha
x^{(2)}+(1-\alpha )x^{(1)})](x^{(2)}-x^{(1)})\mathrm{d}\alpha }-1\Big)
\notag \\
&&\qquad \qquad \qquad \times \langle \mathfrak{e}_{x^{(1)}},\Delta _{\omega
,\vartheta }\mathfrak{e}_{x^{(2)}}\rangle a_{x^{(1)}}^{\ast }a_{x^{(2)}}\Big)%
\ .  \label{current observable new}
\end{eqnarray}

If the interparticle interaction $\Psi ^{\mathrm{IP}}$ is locally gauge
invariant, that is, for all $x\in \mathfrak{L}$,%
\begin{equation*}
\sum\limits_{\Lambda \in \mathcal{P}_{f}(\mathfrak{L})}\left[ \Psi _{\Lambda
}^{\mathrm{IP}},a_{x}^{\ast }a_{x}\right] =0\ ,
\end{equation*}%
then, in absence of external electromagnetic potentials, $I_{\mathbf{x}%
}^{(\omega ,\vartheta )}$ is the observable related to the flow of particles
from the lattice site $x^{(1)}$ to the lattice site $x^{(2)}$ or the current
from $x^{(2)}$ to $x^{(1)}$. $\mathrm{I}_{\mathbf{x}}^{(\omega ,\vartheta ,%
\mathbf{A})}$ corresponds to a correction, engendered by the presence of an
external electromagnetic potential, to the current $I_{\mathbf{x}}^{(\omega
,\vartheta )}$. See \cite[Section 3.2]{OhmV}. Let%
\begin{equation}
P_{\mathbf{x}}^{(\omega ,\vartheta )}:=2\func{Re}\left( \langle \mathfrak{e}%
_{x^{(1)}},\Delta _{\omega ,\vartheta }\mathfrak{e}_{x^{(2)}}\rangle
a_{x^{(1)}}^{\ast }a_{x^{(2)}}\right) \ ,\qquad \mathbf{x}%
:=(x^{(1)},x^{(2)})\in \mathfrak{L}^{2}\ .  \label{R x}
\end{equation}%
Now, for any $\beta \in \mathbb{R}^{+}$, $\omega \in \Omega $ and $\vartheta
,\lambda \in \mathbb{R}_{0}^{+}$, we define two important functions
associated with these observables:

\begin{itemize}
\item[(p)] The paramagnetic transport coefficient $\sigma _{\mathrm{p}%
}^{(\omega )}\equiv \sigma _{\mathrm{p}}^{(\beta ,\omega ,\vartheta ,\lambda
)}$ is defined, for any $\mathbf{x},\mathbf{y}\in \mathfrak{L}^{2}$ and $%
t\in \mathbb{R}$, by%
\begin{equation}
\sigma _{\mathrm{p}}^{(\omega )}\left( \mathbf{x},\mathbf{y},t\right)
:=\int\nolimits_{0}^{t}\varrho ^{(\beta ,\omega ,\vartheta ,\lambda )}\left(
i[I_{\mathbf{y}}^{(\omega ,\vartheta )},\tau _{s}^{(\omega ,\vartheta
,\lambda )}(I_{\mathbf{x}}^{(\omega ,\vartheta )})]\right) \mathrm{d}s\ .
\label{backwards -1bis}
\end{equation}

\item[(d)] The diamagnetic transport coefficient $\sigma _{\mathrm{d}%
}^{(\omega )}\equiv \sigma _{\mathrm{d}}^{(\beta ,\omega ,\vartheta ,\lambda
)}$ is defined by%
\begin{equation}
\sigma _{\mathrm{d}}^{(\omega )}\left( \mathbf{x}\right) :=\varrho ^{(\beta
,\omega ,\vartheta ,\lambda )}\left( P_{\mathbf{x}}^{(\omega ,\vartheta
)}\right) \ ,\qquad \mathbf{x}\in \mathfrak{L}^{2}\ .
\label{backwards -1bispara}
\end{equation}
\end{itemize}

\noindent For boxes $\Lambda _{l}$ (\ref{eq:def lambda n}), we then define
the space--averaged paramagnetic transport coefficient%
\begin{equation*}
t\mapsto \Xi _{\mathrm{p},l}^{(\omega )}\left( t\right) \equiv \Xi _{\mathrm{%
p},l}^{(\beta ,\omega ,\vartheta ,\lambda )}\left( t\right) \in \mathcal{B}(%
\mathbb{R}^{d})
\end{equation*}%
w.r.t. the canonical orthonormal basis $\{e_{k}\}_{k=1}^{d}$ of the
Euclidian space $\mathbb{R}^{d}$ by%
\begin{equation}
\left\{ \Xi _{\mathrm{p},l}^{(\omega )}\left( t\right) \right\} _{k,q}:=%
\frac{1}{\left\vert \Lambda _{l}\right\vert }\underset{x,y\in \Lambda _{l}}{%
\sum }\sigma _{\mathrm{p}}^{(\omega )}\left( x+e_{q},x,y+e_{k},y,t\right)
\label{average microscopic AC--conductivity}
\end{equation}%
for any $l,\beta \in \mathbb{R}^{+}$, $\omega \in \Omega $, $\vartheta
,\lambda \in \mathbb{R}_{0}^{+}$, $k,q\in \{1,\ldots ,d\}$ and $t\in \mathbb{%
R}$. See \cite[Theorem 3.4, Corollary 3.5]{OhmV} for details on the
properties of $\Xi _{\mathrm{p},l}^{(\omega )}$. The space--averaged
diamagnetic transport coefficient
\begin{equation*}
\Xi _{\mathrm{d},l}^{(\omega )}\equiv \Xi _{\mathrm{d},l}^{(\beta ,\omega
,\vartheta ,\lambda )}\in \mathcal{B}(\mathbb{R}^{d})
\end{equation*}%
corresponds (w.r.t. $\{e_{k}\}_{k=1}^{d}$) to the diagonal matrix defined by%
\begin{equation}
\left\{ \Xi _{\mathrm{d},l}^{(\omega )}\right\} _{k,q}:=\frac{\delta _{k,q}}{%
\left\vert \Lambda _{l}\right\vert }\underset{x\in \Lambda _{l}}{\sum }%
\sigma _{\mathrm{d}}^{(\omega )}\left( x+e_{k},x\right) \in \left[ -2\left(
\vartheta +1\right) ,2\left( \vartheta +1\right) \right] \ .
\label{average microscopic AC--conductivity dia}
\end{equation}%
Both random coefficients turn out to be the paramagnetic and diamagnetic
(in--phase) conductivities.

We define the deterministic paramagnetic transport coefficient
\begin{equation*}
t\mapsto \mathbf{\Xi }_{\mathrm{p}}\left( t\right) \equiv \mathbf{\Xi }_{%
\mathrm{p}}^{(\beta ,\vartheta ,\lambda )}\left( t\right) \in \mathcal{B}(%
\mathbb{R}^{d})
\end{equation*}%
by%
\begin{equation}
\mathbf{\Xi }_{\mathrm{p}}\left( t\right) :=\underset{l\rightarrow \infty }{%
\lim }\mathbb{E}\left[ \Xi _{\mathrm{p},l}^{(\omega )}\left( t\right) \right]
\label{paramagnetic transport coefficient macro}
\end{equation}%
for any $\beta \in \mathbb{R}^{+}$, $\vartheta ,\lambda \in \mathbb{R}%
_{0}^{+}$, $k,q\in \{1,\ldots ,d\}$ and $t\in \mathbb{R}$. It is
well--defined, by Theorem \ref{Theorem AC conductivity measure copy(1)}.
Furthermore, the convergence is uniform for times $t$ in compact sets.
Analogously, we also introduce the deterministic diamagnetic transport
coefficient
\begin{equation*}
\mathbf{\Xi }_{\mathrm{d}}\equiv \mathbf{\Xi }_{\mathrm{d}}^{(\beta
,\vartheta ,\lambda )}\in \mathcal{B}(\mathbb{R}^{d})
\end{equation*}%
defined, for any $\beta \in \mathbb{R}^{+}$ and $\vartheta ,\lambda \in
\mathbb{R}_{0}^{+}$, by
\begin{equation}
\mathbf{\Xi }_{\mathrm{d}}:=\underset{l\rightarrow \infty }{\lim }\mathbb{E}%
\left[ \Xi _{\mathrm{d},l}^{(\omega )}\right] \ .  \label{def sigma_d}
\end{equation}%
Indeed, since the map (\ref{map}) is a random invariant state and $%
\mathfrak{a}_{\Omega }$ is an ergodic measure, we have, for all $l\in
\mathbb{R}^{+}$,
\begin{equation*}
\mathbf{\Xi }_{\mathrm{d}}=\mathbb{E}\left[ \Xi _{\mathrm{d},l}^{(\omega )}%
\right] \,.
\end{equation*}%
Clearly, $\{\mathbf{\Xi }_{\mathrm{d}}\}_{k,k}\in \lbrack -2(\vartheta
+1),2(\vartheta +1)]$ for any $k\in \{1,\ldots ,d\}$.

By using the Akcoglu--Krengel ergodic theorem we show that the limits $%
l\rightarrow \infty $ of $\Xi _{\mathrm{p},l}^{(\omega )}$ and $\Xi _{%
\mathrm{d},l}^{(\omega )}$ converge almost surely to $\mathbf{\Xi }_{\mathrm{%
p}}$ and $\mathbf{\Xi }_{\mathrm{d}}$, respectively.

\begin{satz}[Macroscopic charge transport coefficients]
\label{thm charged transport coefficient}\mbox{
}\newline
Assume (\ref{(3.1) NS})--(\ref{(3.2) NS}), (\ref{static potential0}) and
that the map (\ref{map}) is a random invariant state (see Definition \ref%
{def second law2 copy(1)}). Let $\beta \in \mathbb{R}^{+}$ and $\vartheta
,\lambda \in \mathbb{R}_{0}^{+}$. Then, there is a measurable subset $\tilde{%
\Omega}\equiv \tilde{\Omega}^{(\beta ,\vartheta ,\lambda )}\subset \Omega $
of full measure (that is, $\tilde{\Omega}\in \mathfrak{A}_{\Omega }$ and $%
\mathfrak{a}_{\Omega }(\tilde{\Omega})=1$) such that, for any $\omega \in
\tilde{\Omega}$, one has:\newline
\emph{(p)} Paramagnetic transport coefficient: For all $t\in \mathbb{R}$,%
\begin{equation*}
\mathbf{\Xi }_{\mathrm{p}}\left( t\right) =\underset{l\rightarrow \infty }{%
\lim }\Xi _{\mathrm{p},l}^{(\omega )}\left( t\right) \ .
\end{equation*}%
The above limit is uniform for times $t$ on compact sets.\newline
\emph{(d)} Diamagnetic transport coefficient:
\begin{equation}
\mathbf{\Xi }_{\mathrm{d}}=\underset{l\rightarrow \infty }{\lim }\Xi _{%
\mathrm{d},l}^{(\omega )}\ .  \notag
\end{equation}
\end{satz}

\begin{proof}
Assertion (p) is proven in a similar way as Theorem \ref{main 1 copy(2)}.
See Equation (\ref{conductivity000}) and the arguments thereafter. Note only
that the pointwise convergence of any equicontinuous family of functions on $%
\mathbb{R}$ implies its uniform convergence on compacta. The proof of
Assertion (d) is even simpler because there is no time dependency. We omit
the details.
\end{proof}

\subsection{Macroscopic Ohm's Law\label{section Macroscopic Ohm's Law}}

For any $l\in \mathbb{R}^{+}$ and $\mathbf{A}\in \mathbf{C}_{0}^{\infty }$,
we consider now the space--rescaled vector potential
\begin{equation}
\mathbf{A}_{l}(t,x):=\mathbf{A}(t,l^{-1}x)\ ,\quad t\in \mathbb{R},\ x\in
\mathbb{R}^{d}\ .  \label{rescaled vector potential}
\end{equation}%
Since Ohm's law is a linear\emph{\ }response to electric fields, we also
rescale the strength of the electromagnetic potential $\mathbf{A}_{l}$ by a
real parameter $\eta \in \mathbb{R}$ and study the behavior of current
densities in the limit $\eta \rightarrow 0$.

Exactly like in \cite[Section 3]{OhmII} and \cite[Section 3.3]{OhmV},
w.l.o.g. we consider space--homogeneous (though time--dependent) electric
fields in the box $\Lambda _{l}$ defined by (\ref{eq:def lambda n}) for $%
l\in \mathbb{R}^{+}$. More precisely, let $\vec{w}\in \mathbb{R}^{d}$ be any
(normalized w.r.t. the usual Euclidian norm) vector, $\mathcal{A}\in
C_{0}^{\infty }\left( \mathbb{R};\mathbb{R}\right) $ and set $\mathcal{E}%
_{t}:=-\partial _{t}\mathcal{A}_{t}$ for all $t\in \mathbb{R}$. Then, $%
\mathbf{\bar{A}}\in \mathbf{C}_{0}^{\infty }$ is defined to be the
electromagnetic potential such that the electric field equals $\mathcal{E}%
_{t}\vec{w}$ at time $t\in \mathbb{R}$ for all $x\in \left[ -1,1\right] ^{d}$
and $(0,0,\ldots ,0)$ for $t\in \mathbb{R}$ and $x\notin \left[ -1,1\right]
^{d}$. This choice yields rescaled electromagnetic potentials $\eta \mathbf{%
\bar{A}}_{l}$ as defined by (\ref{rescaled vector potential}) for $l\in
\mathbb{R}^{+}$ and $\eta \in \mathbb{R}$.

For any $l,\beta \in \mathbb{R}^{+}$, $\omega \in \Omega $, $\vartheta
,\lambda \in \mathbb{R}_{0}^{+}$, $\eta \in \mathbb{R}$, $\vec{w}\in \mathbb{%
R}^{d}$, $\mathcal{A}\in C_{0}^{\infty }\left( \mathbb{R};\mathbb{R}\right) $
and $t\geq t_{0}$, the total current density is the sum of three currents
defined from (\ref{current observable}) and (\ref{current observable new}):

\begin{itemize}
\item[(th)] The (thermal) current density
\begin{equation*}
\mathbb{J}_{\mathrm{th}}^{(\omega ,l)}\equiv \mathbb{J}_{\mathrm{th}%
}^{(\beta ,\omega ,\vartheta ,\lambda ,l)}\in \mathbb{R}^{d}
\end{equation*}%
at thermal equilibrium inside the box $\Lambda _{l}$ is defined, for any $%
k\in \{1,\ldots ,d\}$, by%
\begin{equation}
\left\{ \mathbb{J}_{\mathrm{th}}^{(\omega ,l)}\right\} _{k}:=\left\vert
\Lambda _{l}\right\vert ^{-1}\underset{x\in \Lambda _{l}}{\sum }\varrho
^{(\beta ,\omega ,\vartheta ,\lambda )}\left( I_{\left( x+e_{k},x\right)
}^{(\omega ,\vartheta )}\right) \ .  \label{free current}
\end{equation}

\item[(p)] The paramagnetic current density is the map
\begin{equation*}
t\mapsto \mathbb{J}_{\mathrm{p}}^{(\omega ,\eta \mathbf{\bar{A}}_{l})}\left(
t\right) \equiv \mathbb{J}_{\mathrm{p}}^{(\beta ,\omega ,\vartheta ,\lambda
,\eta \mathbf{\bar{A}}_{l})}\left( t\right) \in \mathbb{R}^{d}
\end{equation*}%
defined by the space average of the current increment vector inside the box $%
\Lambda _{l}$ at time $t\geq t_{0}$, that is, for any $k\in \{1,\ldots ,d\}$%
,
\begin{equation}
\left\{ \mathbb{J}_{\mathrm{p}}^{(\omega ,\eta \mathbf{\bar{A}}_{l})}\left(
t\right) \right\} _{k}:=\left\vert \Lambda _{l}\right\vert ^{-1}\underset{%
x\in \Lambda _{l}}{\sum }\rho _{t}^{(\beta ,\omega ,\vartheta ,\lambda ,\eta
\mathbf{\bar{A}}_{l})}\left( I_{\left( x+e_{k},x\right) }^{(\omega
,\vartheta )}\right) -\left\{ \mathbb{J}_{\mathrm{th}}^{(\omega ,l)}\right\}
_{k}\ .  \label{finite volume current density}
\end{equation}

\item[(d)] The diamagnetic (or ballistic) current density%
\begin{equation*}
t\mapsto \mathbb{J}_{\mathrm{d}}^{(\omega ,\eta \mathbf{\bar{A}}_{l})}\left(
t\right) \equiv \mathbb{J}_{\mathrm{d}}^{(\beta ,\omega ,\vartheta ,\lambda
,\eta \mathbf{\bar{A}}_{l})}\left( t\right) \in \mathbb{R}^{d}
\end{equation*}%
is defined analogously, for any $t\geq t_{0}$ and $k\in \{1,\ldots ,d\}$, by%
\begin{equation}
\left\{ \mathbb{J}_{\mathrm{d}}^{(\omega ,\eta \mathbf{\bar{A}}_{l})}\left(
t\right) \right\} _{k}:=\left\vert \Lambda _{l}\right\vert ^{-1}\underset{%
x\in \Lambda _{l}}{\sum }\rho _{t}^{(\beta ,\omega ,\vartheta ,\lambda ,\eta
\mathbf{\bar{A}}_{l})}\left( \mathrm{I}_{(x+e_{k},x)}^{(\omega ,\vartheta
,\eta \mathbf{\bar{A}}_{l})}\right) \ .
\label{finite volume current density2}
\end{equation}
\end{itemize}

\noindent For more details on the physical interpretation of these currents,
see \cite[Section 3.4]{OhmII}.

By \cite[Theorem 3.7]{OhmV} and Conditions (\ref{(3.1) NS})--(\ref{(3.2) NS}%
) and (\ref{(3.3) NS}), the current densities behave, at small $|\eta |$ and
uniformly w.r.t. the size of the box, linearly w.r.t. $\eta $: For any $%
\vartheta _{0}\in \mathbb{R}_{0}^{+}$, $\mathcal{A}\in C_{0}^{\infty }\left(
\mathbb{R};\mathbb{R}\right) $ and $\eta \in \mathbb{R}$,%
\begin{eqnarray*}
\mathbb{J}_{\mathrm{p}}^{(\omega ,\eta \mathbf{\bar{A}}_{l})}\left( t\right)
&=&\eta J_{\mathrm{p},l}^{(\omega ,\mathcal{A})}(t)+\mathcal{O}\left( \eta
^{2}\right) \ ,\quad J_{\mathrm{p},l}^{(\omega ,\mathcal{A}%
)}(t):=\int_{t_{0}}^{t}\left( \Xi _{\mathrm{p},l}^{(\omega )}\left(
t-s\right) \vec{w}\right) \mathcal{E}_{s}\mathrm{d}s\ , \\
\mathbb{J}_{\mathrm{d}}^{(\omega ,\eta \mathbf{\bar{A}}_{l})}\left( t\right)
&=&\eta J_{\mathrm{d},l}^{(\omega ,\mathcal{A})}(t)+\mathcal{O}\left( \eta
^{2}\right) \ ,\quad J_{\mathrm{d},l}^{(\omega ,\mathcal{A})}(t):=\left( \Xi
_{\mathrm{d},l}^{(\omega )}\vec{w}\right) \int_{t_{0}}^{t}\mathcal{E}_{s}%
\mathrm{d}s\ ,
\end{eqnarray*}%
uniformly for $l,\beta \in \mathbb{R}^{+}$, $\omega \in \Omega $, $\vartheta
\in \lbrack 0,\vartheta _{0}]$, $\lambda \in \mathbb{R}_{0}^{+}$, $\vec{w}%
\in \mathbb{R}^{d}$ (normalized) and $t\geq t_{0}$.

The $\mathbb{R}^{d}$--valued linear coefficients
\begin{equation*}
J_{\mathrm{p},l}^{(\omega ,\mathcal{A})}\equiv J_{\mathrm{p},l}^{(\beta
,\omega ,\vartheta ,\lambda ,\vec{w},\mathcal{A})}\qquad \text{and}\qquad J_{%
\mathrm{d},l}^{(\omega ,\mathcal{A})}\equiv J_{\mathrm{d},l}^{(\beta ,\omega
,\vartheta ,\lambda ,\vec{w},\mathcal{A})}
\end{equation*}%
of the paramagnetic and diamagnetic current densities, respectively, become
deterministic for large boxes. They are directly related to $\mathbf{\Xi }_{%
\mathrm{p}}$ and $\mathbf{\Xi }_{\mathrm{d}}$ via Ohm's law:

\begin{satz}[Macroscopic Ohm's law]
\label{main 1 copy(8)}\mbox{
}\newline
Assume (\ref{(3.1) NS})--(\ref{(3.2) NS}), (\ref{static potential0})--(\ref%
{(3.3) NS}) and that the map (\ref{map}) is a random invariant state. Let $%
\beta \in \mathbb{R}^{+}$ and $\vartheta ,\lambda \in \mathbb{R}_{0}^{+}$.
Then, there is a measurable subset $\tilde{\Omega}\equiv \tilde{\Omega}%
^{(\beta ,\vartheta ,\lambda )}\subset \Omega $ of full measure such that,
for any $\omega \in \tilde{\Omega}$, $\vec{w}\in \mathbb{R}^{d}$, $\mathcal{A%
}\in C_{0}^{\infty }\left( \mathbb{R};\mathbb{R}\right) $ and $t\geq t_{0}$,
the following assertions hold true:\newline
\emph{(th)} Thermal current density:
\begin{equation*}
\underset{l\rightarrow \infty }{\lim }\left\{ \mathbb{J}_{\mathrm{th}%
}^{(\omega ,l)}\right\} _{k}=\mathbb{E}\left[ \varrho ^{(\beta ,\omega
,\vartheta ,\lambda )}(I_{\left( e_{k},0\right) }^{(\omega ,\vartheta )})%
\right] \ ,\qquad k\in \{1,\ldots ,d\}\ .
\end{equation*}%
\emph{(p)} Paramagnetic current density:%
\begin{equation}
\underset{l\rightarrow \infty }{\lim }J_{\mathrm{p},l}^{(\omega ,\mathcal{A}%
)}(t)=\underset{l\rightarrow \infty }{\lim }\left( \left. \partial _{\eta }%
\mathbb{J}_{\mathrm{p}}^{(\omega ,\eta \mathbf{\bar{A}}_{l})}\left( t\right)
\right\vert _{\eta =0}\right) =\int_{t_{0}}^{t}\left( \mathbf{\Xi }_{\mathrm{%
p}}\left( t-s\right) \vec{w}\right) \mathcal{E}_{s}\mathrm{d}s\ .  \notag
\end{equation}%
\emph{(d)} Diamagnetic current density:
\begin{equation}
\underset{l\rightarrow \infty }{\lim }J_{\mathrm{d},l}^{(\omega ,\mathcal{A}%
)}(t)=\underset{l\rightarrow \infty }{\lim }\left( \left. \partial _{\eta }%
\mathbb{J}_{\mathrm{d}}^{(\omega ,\eta \mathbf{\bar{A}}_{l})}\left( t\right)
\right\vert _{\eta =0}\right) =\left( \mathbf{\Xi }_{\mathrm{d}}\vec{w}%
\right) \int_{t_{0}}^{t}\mathcal{E}_{s}\mathrm{d}s\ .  \notag
\end{equation}
\end{satz}

\begin{proof}
(th) is similar to \cite[Corollary 5.7 (th)]{OhmIII}. Assertions (p) and (d)
are deduced from Theorem \ref{thm charged transport coefficient} and
Lebesgue's dominated convergence theorem. Note that the intersection of
three measurable sets of full measure has full measure.
\end{proof}

\noindent Like \cite[Theorem 3.7]{OhmV}, Theorem \ref{main 1 copy(8)} can
also be extended to space--inhomo%
\-%
geneous macroscopic electromagnetic fields, that is, for space--rescaled
vector potentials $\mathbf{A}_{l}$ (\ref{rescaled vector potential}) with
arbitrary $\mathbf{A}\in \mathbf{C}_{0}^{\infty }$.

\subsection{Green--Kubo Relations}

Because of Theorem \ref{main 1 copy(8)} (p)--(d), $\mathbf{\Xi }_{\mathrm{p}%
} $ and $\mathbf{\Xi }_{\mathrm{d}}$ are both \emph{charge} transport
coefficients. Thus, they are also named here \emph{paramagnetic }and\emph{\
diamagnetic (in--phase) conductivities}, respectively. From (\ref%
{paramagnetic transport coefficient macro}) we can deduce Green--Kubo
relations for $\mathbf{\Xi }_{\mathrm{p}}$ via current Duhamel fluctuations
as follows.

Fix in all the subsection $\beta \in \mathbb{R}^{+}$ and $\vartheta ,\lambda
\in \mathbb{R}_{0}^{+}$. The Duhamel two--point function $(\cdot ,\cdot
)_{\sim }^{(\omega )}$ is defined by%
\begin{equation*}
(B_{1},B_{2})_{\sim }^{(\omega )}\equiv (B_{1},B_{2})_{\sim }^{(\beta
,\omega ,\vartheta ,\lambda )}:=\int\nolimits_{0}^{\beta }\varrho ^{(\beta
,\omega ,\vartheta ,\lambda )}\left( B_{1}^{\ast }\tau _{i\alpha }^{(\omega
,\vartheta ,\lambda )}(B_{2})\right) \mathrm{d}\alpha
\end{equation*}%
for any $B_{1},B_{2}\in \mathcal{U}$ and $\omega \in \Omega $. See for
instance \cite[Section A]{OhmII} and references therein for more details.
For any $l\in \mathbb{R}^{+}$ and $B\in \mathcal{U}$, set
\begin{equation}
\mathbb{F}^{(l)}\left( B\right) :=\frac{1}{\left\vert \Lambda
_{l}\right\vert ^{1/2}}\underset{x\in \Lambda _{l}}{\sum }\left\{ \chi
_{x}\left( B\right) -\varrho ^{(\beta ,\omega ,\vartheta ,\lambda )}\left(
\chi _{x}\left( B\right) \right) \mathbf{1}_{\mathcal{U}}\right\} \ .
\label{Fluctuation2}
\end{equation}%
We name it the \emph{fluctuation observable }of the element $B\in \mathcal{U}
$ in the box $\Lambda _{l}$. Recall that $\{\chi _{x}\}_{x\in \mathfrak{L}}$
implements the action of the group $(\mathbb{Z}^{d},+)$ of lattice
translations on the CAR $C^{\ast }$--algebra $\mathcal{U}$, see (\ref%
{definition trans U}).

Then, by \cite[Eq. (103)]{OhmII} together with (\ref{paramagnetic transport
coefficient macro}), one obtains \emph{Green--Kubo relations} for the
paramagnetic (in--phase) conductivity: For any $k,q\in \{1,\ldots ,d\}$ and $%
t\in \mathbb{R}$,
\begin{eqnarray}
\left\{ \mathbf{\Xi }_{\mathrm{p}}\left( t\right) \right\} _{k,q} &=&%
\underset{l\rightarrow \infty }{\lim }\mathbb{E}\left[ \left( \mathbb{F}%
^{(l)}(I_{(e_{k},0)}^{(\omega ,\vartheta )}),\mathbb{F}^{(l)}(\tau
_{t}^{(\omega ,\vartheta ,\lambda )}(I_{(e_{q},0)}^{(\omega ,\vartheta
)}))\right) _{\sim }^{(\omega )}\right.  \notag \\
&&\qquad \qquad \qquad \left. -\left( \mathbb{F}^{(l)}(I_{(e_{k},0)}^{(%
\omega ,\vartheta )}),\mathbb{F}^{(l)}(I_{(e_{q},0)}^{(\omega ,\vartheta
)})\right) _{\sim }^{(\omega )}\right]  \label{green kubo current increment}
\end{eqnarray}%
with the current observable $I_{(x,y)}^{(\omega ,\vartheta )}$ defined by (%
\ref{current observable}). The right hand side (r.h.s.) of the above
equation is a current Duhamel fluctuation \emph{increment}. If Conditions (%
\ref{(3.1) NS})--(\ref{(3.2) NS}) and (\ref{static potential0}) hold and the
map (\ref{map}) is a random invariant state, then the above limit always
exists (and is thus finite), by Theorem \ref{Theorem AC conductivity measure
copy(1)}.

Note however that, possibly,
\begin{equation}
\underset{l\rightarrow \infty }{\lim \sup }\ \mathbb{E}\left[ \left( \mathbb{%
F}^{(l)}(I_{(e_{k},0)}^{(\omega ,\vartheta )}),\mathbb{F}%
^{(l)}(I_{(e_{q},0)}^{(\omega ,\vartheta )})\right) _{\sim }^{(\omega )}%
\right] =\infty  \label{current fluct}
\end{equation}%
for some $k,q\in \{1,\ldots ,d\}$. In other words, it is not a priori clear
whether the interacting quantum system has finite current Duhamel
fluctuations or not. When it is finite (and so are\ both terms in the r.h.s.
of (\ref{green kubo current increment})), similar to \cite[Section 3]{OhmIV}%
, we can construct a Hilbert space of fluctuations, which implies the
existence of a finite conductivity measure as a spectral measure. The
finiteness of current Duhamel fluctuations is proven in \cite[Section 3]%
{OhmIV} for the non--interacting case with random static potentials and
space--homogeneous hopping terms. This can also be shown for sufficiently
small $\left\Vert \Psi ^{\mathrm{IP}}\right\Vert _{\mathcal{W}}$ and
disorder strengths $\vartheta ,\lambda $, by using methods of constructive
quantum field theory.

\section{AC--Conductivity Measure From Joule's Law\label{Sect Conductivity
Measure From Joule's Law}}

Similar to \cite[Section 4.3]{OhmIII}, our derivation of a macroscopic
(in--phase) AC--conducti%
\-%
vity measure is based on the 2nd law of thermodynamics. It dovetails with
the celebrated Joule's law of (classical) electricity theory. To this end we
start by introducing energy increment densities, in particular the heat
production density.

\subsection{Energy Increment Densities\label{sect 2.7}}

The internal energy observable $H_{L}^{(\omega ,\vartheta ,\lambda )}\in
\mathcal{U}^{+}\cap \mathcal{U}_{\Lambda }$ of the interacting fermion
system for the box $\Lambda _{L}$ (\ref{eq:def lambda n}) is defined by%
\begin{eqnarray}
H_{L}^{(\omega ,\vartheta ,\lambda )} &:=&\sum\limits_{\Lambda \subset
\Lambda _{L}}\Psi _{\Lambda }^{(\omega ,\vartheta )}+\lambda
\sum\limits_{x\in \Lambda _{L}}\omega _{1}(x)a_{x}^{\ast }a_{x}
\label{def H loc} \\
&=&\sum\limits_{x,y\in \Lambda _{L}}\langle \mathfrak{e}_{x},(\Delta
_{\omega ,\vartheta }+\lambda V_{\omega })\mathfrak{e}_{y}\rangle
a_{x}^{\ast }a_{y}+\sum\limits_{\Lambda \subset \Lambda _{L}}\Psi _{\Lambda
}^{\mathrm{IP}}  \notag
\end{eqnarray}%
for $\omega =(\omega _{1},\omega _{2})\in \Omega $, $\vartheta ,\lambda \in
\mathbb{R}_{0}^{+}$ and $L\in \mathbb{R}^{+}$. When the electromagnetic
field is switched on, i.e., for $t\geq t_{0}$, the total energy observable
for the box $\Lambda _{L}$ that includes the region where the
electromagnetic field does not vanish equals%
\begin{equation*}
H_{L}^{(\omega ,\vartheta ,\lambda )}+W_{t}^{(\omega ,\vartheta ,\mathbf{A}%
)}\ ,
\end{equation*}%
where, for any $\omega \in \Omega $, $\vartheta \in \mathbb{R}_{0}^{+}$, $%
\mathbf{A}\in \mathbf{C}_{0}^{\infty }$ and $t\in \mathbb{R}$,
\begin{equation*}
W_{t}^{(\omega ,\vartheta ,\mathbf{A})}:=\sum\limits_{x,y\in \mathfrak{L}%
}\langle \mathfrak{e}_{x},(\Delta _{\omega ,\vartheta }^{(\mathbf{A}%
)}-\Delta _{\omega ,\vartheta })\mathfrak{e}_{y}\rangle a_{x}^{\ast
}a_{y}\in \mathcal{U}^{+}\cap \mathcal{U}_{0}
\end{equation*}%
is the electromagnetic\emph{\ }potential energy observable.

Like in \cite[Sections 3.1, 3.4]{OhmV}, we now define four sorts of energy
increments associated with the fermion system for any $\beta \in \mathbb{R}%
^{+}$, $\omega \in \Omega $, $\vartheta ,\lambda \in \mathbb{R}_{0}^{+}$ and
$\mathbf{A}\in \mathbf{C}_{0}^{\infty }$:

\begin{itemize}
\item[($\mathbf{Q}$)] The internal energy increment $\mathbf{S}^{(\omega ,%
\mathbf{A})}\equiv \mathbf{S}^{(\beta ,\omega ,\vartheta ,\lambda ,\mathbf{A}%
)}$ is the map from $\mathbb{R}$ to $\mathbb{R}_{0}^{+}$ defined by%
\begin{equation*}
\mathbf{S}^{(\omega ,\mathbf{A})}\left( t\right) :=\lim_{L\rightarrow \infty
}\left\{ \rho _{t}^{(\beta ,\omega ,\vartheta ,\lambda ,\mathbf{A}%
)}(H_{L}^{(\omega ,\vartheta ,\lambda )})-\varrho ^{(\beta ,\omega
,\vartheta ,\lambda )}(H_{L}^{(\omega ,\vartheta ,\lambda )})\right\} \ .
\end{equation*}%
Under Conditions (\ref{(3.1) NS})--(\ref{(3.2) NS}) and (\ref{(3.3) NS}),
this map has non--negative finite value and is the heat production because
of \cite[Theorem 3.2]{OhmV}.

\item[($\mathbf{P}$)] The electromagnetic potential energy increment $%
\mathbf{P}^{(\omega ,\mathbf{A})}\equiv \mathbf{P}^{(\beta ,\omega
,\vartheta ,\lambda ,\mathbf{A})}$ is the map from $\mathbb{R}$ to $\mathbb{R%
}$ defined by%
\begin{equation*}
\mathbf{P}^{(\omega ,\mathbf{A})}\left( t\right) :=\rho _{t}^{(\beta ,\omega
,\vartheta ,\lambda ,\mathbf{A})}(W_{t}^{(\omega ,\vartheta ,\mathbf{A})})\ .
\end{equation*}

\item[(p)] The paramagnetic energy increment $\mathfrak{J}_{\mathrm{p}%
}^{(\omega ,\mathbf{A})}\equiv \mathfrak{I}_{\mathrm{p}}^{(\beta ,\omega
,\vartheta ,\lambda ,\mathbf{A})}$ is the map from $\mathbb{R}$ to $\mathbb{R%
}$ defined by%
\begin{multline*}
\mathfrak{I}_{\mathrm{p}}^{(\omega ,\mathbf{A})}\left( t\right)
:=\lim_{L\rightarrow \infty }\left\{ \rho _{t}^{(\beta ,\omega ,\vartheta
,\lambda ,\mathbf{A})}(H_{L}^{(\omega ,\vartheta ,\lambda )}+W_{t}^{(\omega
,\vartheta ,\mathbf{A})})\right. \\
\left. -\varrho ^{(\beta ,\omega ,\vartheta ,\lambda )}(H_{L}^{(\omega
,\vartheta ,\lambda )}+W_{t}^{(\omega ,\vartheta ,\mathbf{A})})\right\} \ .
\end{multline*}

\item[(d)] The diamagnetic energy increment $\mathfrak{I}_{\mathrm{d}%
}^{(\omega ,\mathbf{A})}\equiv \mathfrak{I}_{\mathrm{d}}^{(\beta ,\omega
,\vartheta ,\lambda ,\mathbf{A})}$ is the map from $\mathbb{R}$ to $\mathbb{R%
}$ defined by
\begin{equation*}
\mathfrak{I}_{\mathrm{d}}^{(\omega ,\mathbf{A})}\left( t\right) :=\varrho
^{(\beta ,\omega ,\vartheta ,\lambda )}(W_{t}^{(\omega ,\vartheta ,\mathbf{A}%
)})\ .
\end{equation*}
\end{itemize}

\noindent See \cite{OhmII} for more discussions on the physical
interpretation of these energies. Note that the limits described in ($%
\mathbf{Q}$) and (p) exist at all times. Indeed, the total energy increment
\begin{equation*}
\rho _{t}^{(\beta ,\omega ,\vartheta ,\lambda ,\mathbf{A})}(H_{L}^{(\omega
,\vartheta ,\lambda )}+W_{t}^{(\omega ,\vartheta ,\mathbf{A})})-\varrho
^{(\beta ,\omega ,\vartheta ,\lambda )}(H_{L}^{(\omega ,\vartheta ,\lambda
)})
\end{equation*}%
is shown in \cite[Theorem 3.2 (ii)]{OhmV} to be the \emph{work} performed by
the electric field and is given in the limit $L\rightarrow \infty $ by an
expression like (\ref{work}), which, on the other hand, equals
\begin{equation*}
\mathbf{S}^{(\omega ,\mathbf{A})}\left( t\right) +\mathbf{P}^{(\omega ,%
\mathbf{A})}\left( t\right) =\mathfrak{I}_{\mathrm{p}}^{(\omega ,\mathbf{A}%
)}\left( t\right) +\mathfrak{I}_{\mathrm{d}}^{(\omega ,\mathbf{A})}\left(
t\right) \ .
\end{equation*}

Under Conditions (\ref{(3.1) NS})--(\ref{(3.2) NS}) and (\ref{(3.3) NS}),
all increment energies defined above are of order $\mathcal{O}\left( \eta
^{2}l^{d}\right) $, as $l\rightarrow \infty $, by \cite[Theorem 3.8]{OhmV}.
Indeed, because of possibly non--vanishing thermal currents, the energy
increments $\mathbf{P}^{(\omega ,\mathbf{A})}$ and $\mathfrak{I}_{\mathrm{d}%
}^{(\omega ,\mathbf{A})}$ are rather $\mathcal{O}\left( \left\vert \eta
\right\vert l^{d}\right) $ at small $l\in \mathbb{R}_{0}^{+}$. As a
consequence, for any $\beta \in \mathbb{R}^{+}$, $\omega \in \Omega $, $%
\vartheta ,\lambda \in \mathbb{R}_{0}^{+}$ and $\mathbf{A}\in \mathbf{C}%
_{0}^{\infty }$, we define four energy densities:

\begin{itemize}
\item[($\mathbf{Q}$)] The \emph{heat production} (or internal energy
increment) density $\mathbf{s}\equiv \mathbf{s}^{(\beta ,\omega ,\vartheta
,\lambda ,\mathbf{A})}$ is the map from $\mathbb{R}$ to $\mathbb{R}_{0}^{+}$
defined by%
\begin{equation}
\mathbf{s}\left( t\right) :=\underset{(\eta ,l^{-1})\rightarrow (0,0)}{\lim }%
\left\{ \left( \eta ^{2}\left\vert \Lambda _{l}\right\vert \right) ^{-1}%
\mathbf{S}^{(\omega ,\eta \mathbf{A}_{l})}\left( t\right) \right\} \ .
\end{equation}

\item[($\mathbf{P}$)] The (electromagnetic) \emph{potential} energy
(increment) density $\mathbf{p}\equiv \mathbf{p}^{(\beta ,\omega ,\vartheta
,\lambda ,\mathbf{A})}$ is the map from $\mathbb{R}$ to $\mathbb{R}$ defined
by%
\begin{equation}
\mathbf{p}\left( t\right) :=\underset{\eta \rightarrow 0}{\lim }\ \underset{%
l\rightarrow \infty }{\lim }\left\{ \left( \eta ^{2}\left\vert \Lambda
_{l}\right\vert \right) ^{-1}\mathbf{P}^{(\omega ,\eta \mathbf{A}%
_{l})}\left( t\right) \right\} \ .
\end{equation}

\item[(p)] The \emph{paramagnetic} energy (increment) density $\mathfrak{i}_{%
\mathrm{p}}\equiv \mathfrak{i}_{\mathrm{p}}^{(\beta ,\omega ,\vartheta
,\lambda ,\mathbf{A})}$ is the map from $\mathbb{R}$ to $\mathbb{R}$ defined
by%
\begin{equation}
\mathfrak{i}_{\mathrm{p}}\left( t\right) :=\underset{(\eta
,l^{-1})\rightarrow (0,0)}{\lim }\left\{ \left( \eta ^{2}\left\vert \Lambda
_{l}\right\vert \right) ^{-1}\mathfrak{I}_{\mathrm{p}}^{(\omega ,\eta
\mathbf{A}_{l})}\left( t\right) \right\} \ .
\label{paramagnetic energy density}
\end{equation}

\item[(d)] The \emph{diamagnetic} energy (increment) density $\mathfrak{i}_{%
\mathrm{d}}\equiv \mathfrak{i}_{\mathrm{d}}^{(\beta ,\omega ,\vartheta
,\lambda ,\mathbf{A})}$ the map from $\mathbb{R}$ to $\mathbb{R}$ defined by
\begin{equation}
\mathfrak{i}_{\mathrm{d}}\left( t\right) :=\underset{\eta \rightarrow 0}{%
\lim }\ \underset{l\rightarrow \infty }{\lim }\left\{ \left( \eta
^{2}\left\vert \Lambda _{l}\right\vert \right) ^{-1}\mathfrak{I}_{\mathrm{d}%
}^{(\omega ,\eta \mathbf{A}_{l})}\left( t\right) \right\} \ .
\end{equation}
\end{itemize}

\noindent On a measurable subset of full measure, all energy (increment)
densities become deterministic functions that are derived in the next
subsection. We explain this in the next subsection.

\subsection{Macroscopic Joule's Law\label{sect Classical Joule's Law}}

Similar to the heuristics presented in \cite[Section 4.2]{OhmIII}, we expect
from Theorem \ref{main 1 copy(8)} that, for $\beta \in \mathbb{R}^{+}$, $%
\vartheta ,\lambda \in \mathbb{R}_{0}^{+}$ and any (possibly space
inhomogeneous) electromagnetic potential $\mathbf{A}\in \mathbf{C}%
_{0}^{\infty }$, the electric field $E_{\mathbf{A}}$ yields \emph{%
space--dependent} paramagnetic and diamagnetic current linear response
coefficients respectively equal to%
\begin{eqnarray}
J_{\mathrm{p}}(t,x) &\equiv &J_{\mathrm{p}}^{(\beta ,\vartheta ,\lambda ,%
\mathbf{A})}(t,x):=\int_{t_{0}}^{t}\mathbf{\Xi }_{\mathrm{p}}\left(
t-s\right) E_{\mathbf{A}}(s,x)\mathrm{d}s\ ,  \label{current para limit} \\
J_{\mathrm{d}}(t,x) &\equiv &J_{\mathrm{d}}^{(\beta ,\vartheta ,\lambda ,%
\mathbf{A})}(t,x):=\mathbf{\Xi }_{\mathrm{d}}\int_{t_{0}}^{t}E_{\mathbf{A}%
}(s,x)\mathrm{d}s\ ,  \label{current dia limit}
\end{eqnarray}%
at any position $x\in \mathbb{R}^{d}$ and time $t\in \mathbb{R}$. These
current linear response coefficients yield two electric work or energy
densities produced by the paramagnetic and diamagnetic currents. This fact
is proven in the following theorem:

\begin{satz}[Macroscopic Joule's law]
\label{main 1 copy(1)}\mbox{
}\newline
Assume (\ref{(3.1) NS})--(\ref{(3.2) NS}), (\ref{static potential0})--(\ref%
{(3.3) NS}) and that the map (\ref{map}) is a random invariant state. Let $%
\beta \in \mathbb{R}^{+}$ and $\vartheta ,\lambda \in \mathbb{R}_{0}^{+}$.
Then, there is a measurable subset $\tilde{\Omega}\equiv \tilde{\Omega}%
^{(\beta ,\vartheta ,\lambda )}\subset \Omega $ of full measure such that,
for any $\omega \in \tilde{\Omega}$, $\mathbf{A}\in \mathbf{C}_{0}^{\infty }$
and $t\geq t_{0}$: \newline
\emph{(p)} Paramagnetic energy density:%
\begin{equation*}
\mathfrak{i}_{\mathrm{p}}\left( t\right) =\int\nolimits_{\mathbb{R}^{d}}%
\mathrm{d}^{d}x\int\nolimits_{t_{0}}^{t}\mathrm{d}s\left\langle E_{\mathbf{A}%
}(s,x),J_{\mathrm{p}}(s,x)\right\rangle _{{\mathbb{R}}^{d}}\ .
\end{equation*}%
\emph{(d)} Diamagnetic energy density:%
\begin{equation*}
\mathfrak{i}_{\mathrm{d}}\left( t\right) =\int\nolimits_{\mathbb{R}^{d}}%
\mathrm{d}^{d}x\int\nolimits_{t_{0}}^{t}\mathrm{d}s\left\langle E_{\mathbf{A}%
}(s,x),J_{\mathrm{d}}(s,x)\right\rangle _{{\mathbb{R}}^{d}}\ .
\end{equation*}%
\emph{(\textbf{Q})} Heat production density:%
\begin{equation*}
\mathbf{s}\left( t\right) =\mathfrak{i}_{\mathrm{p}}\left( t\right)
-\int\nolimits_{\mathbb{R}^{d}}\mathrm{d}^{d}x\int\nolimits_{t_{0}}^{t}%
\mathrm{d}s\left\langle E_{\mathbf{A}}(s,x),J_{\mathrm{p}}(t,x)\right\rangle
_{{\mathbb{R}}^{d}}\ .
\end{equation*}%
\emph{(\textbf{P})} Electromagnetic potential energy density:%
\begin{equation*}
\mathbf{p}\left( t\right) =\mathfrak{i}_{\mathrm{d}}\left( t\right)
+\int\nolimits_{\mathbb{R}^{d}}\mathrm{d}^{d}x\int\nolimits_{t_{0}}^{t}%
\mathrm{d}s\left\langle E_{\mathbf{A}}(s,x),J_{\mathrm{p}}(t,x)\right\rangle
_{{\mathbb{R}}^{d}}\ .
\end{equation*}
\end{satz}

\begin{proof}
The proof is very similar to the proof of \cite[Theorem 4.1]{OhmIII}. It is
a consequence of the Akcoglu--Krengel ergodic theorem, Lieb--Robinson bounds
\cite[Theorem 3.6 (iv)]{brupedraLR} and \cite[Theorem 3.8]{OhmV}. For the
detailed proof of Assertion (p), see Theorem \ref{main 1 copy(2)}. We omit
the details for Assertions (d), (Q) and (P).
\end{proof}

For more discussions on this subject, see \cite[Section 4.2]{OhmIII}. In
fact, the above result is an extension of \cite[Theorem 4.1]{OhmIII} to
fermion systems \emph{with interactions}.

\subsection{AC--Conductivity Measure\label{sect ac cond}}

At $\beta \in \mathbb{R}^{+}$ and $\vartheta ,\lambda \in \mathbb{R}_{0}^{+}$%
, the paramagnetic transport coefficient $\mathbf{\Xi }_{\mathrm{p}}\equiv
\mathbf{\Xi }_{\mathrm{p}}^{(\beta ,\vartheta ,\lambda )}$ is a
well--defined $\mathcal{B}(\mathbb{R}^{d})$--valued function of time. See (%
\ref{paramagnetic transport coefficient macro}). It is also named here
paramagnetic (in--phase) conductivity, because of Theorem \ref{main 1
copy(8)}.

The positivity of the heat production (Theorem \ref{main 1 copy(1)}), i.e.,
the 2nd law of thermodynamics, implies that the symmetric part of $\mathbf{%
\Xi }_{\mathrm{p}}$ is \emph{conditionally} positive definite or,
equivalently \cite[Proposition 4.4]{SSV}, \emph{negative definite} \emph{in
the sense of Schoenberg}. Observe that the symmetric part of $\mathbf{\Xi }_{%
\mathrm{p}}$ is only conditionally positive definite, and not positive
definite, because of the AC--condition (\ref{zero mean field}) on external
electric fields. Therefore, similar to \cite[Theorem 4.12]{SSV} for
complex--valued negative definite functions (in the sense of Schoenberg),
there is a L\'{e}vy--Khintchine representation of the symmetric part of the
(continuous) paramagnetic (in--phase) conductivity $\mathbf{\Xi }_{\mathrm{p}%
}$. The corresponding L\'{e}vy measure $\mu _{\mathrm{AC}}$ is the
AC--conductivity measure we are looking for. Note that the measure $\nu
^{2}\mu _{\mathrm{AC}}\left( \mathrm{d}\nu \right) $ on $\mathbb{R}%
\backslash \{0\}$ is a priori not a finite measure. However, if Conditions (%
\ref{(3.1) NS})--(\ref{(3.2) NS}) and (\ref{static potential0})--(\ref{(3.3)
NS}) hold and the map (\ref{map}) is a random invariant state, then such a
property holds true because $\mathbf{\Xi }_{\mathrm{p}}\in C^{2}\left(
\mathbb{R},\mathcal{B}(\mathbb{R}^{d})\right) $, by Theorem \ref{Theorem AC
conductivity measure copy(1)}.

Indeed, for any $\Upsilon \in \mathcal{B}(\mathbb{R}^{d})$, define its
symmetric and antisymmetric parts, w.r.t. to the canonical scalar product of
$\mathbb{R}^{d}$, respectively by%
\begin{equation}
\lbrack \Upsilon ]_{\mathrm{+}}:=\frac{1}{2}\left( \Upsilon +\Upsilon ^{%
\mathrm{t}}\right) \text{\qquad and\qquad }[\Upsilon ]_{\mathrm{-}}:=\frac{1%
}{2}\left( \Upsilon -\Upsilon ^{\mathrm{t}}\right) \text{ }.
\label{symm-antisymm Theta}
\end{equation}%
Here, $\Upsilon ^{\mathrm{t}}\in \mathcal{B}(\mathbb{R}^{d})$ stands for the
transpose of the operator $\Upsilon \in \mathcal{B}(\mathbb{R}^{d})$ (w.r.t.
the canonical scalar product of $\mathbb{R}^{d}$). Then we have:

\begin{satz}[{L\'{e}vy--Khintchine representation of $[\mathbf{\Xi }_{\mathrm{%
p}}]_{+}$}]
\label{Theorem AC conductivity measure}\mbox{
}\newline
Assume (\ref{(3.1) NS})--(\ref{(3.2) NS}), (\ref{static potential0})--(\ref%
{(3.3) NS}) and that the map (\ref{map}) is a random invariant state. For
any $\beta \in \mathbb{R}^{+}$ and $\vartheta ,\lambda \in \mathbb{R}%
_{0}^{+} $, there is a unique finite and symmetric $\mathcal{B}_{+}(\mathbb{R%
}^{d})$--valued measure $\mathbf{\mu }\equiv \mathbf{\mu }^{(\beta
,\vartheta ,\lambda )}$ on $\mathbb{R}$ such that, for any $t\in \mathbb{R}$%
,
\begin{equation*}
\lbrack \mathbf{\Xi }_{\mathrm{p}}\left( t\right) ]_{+}=-\frac{t^{2}}{2}%
\mathbf{\mu }\left( \left\{ 0\right\} \right) +\int\nolimits_{\mathbb{R}%
\backslash \left\{ 0\right\} }\left( \cos \left( t\nu \right) -1\right) \nu
^{-2}\mathbf{\mu }\left( \mathrm{d}\nu \right) \ .
\end{equation*}%
$\mathcal{B}_{+}(\mathbb{R}^{d})\subset \mathcal{B}(\mathbb{R}^{d})$ stands
for the set of positive linear operators on $\mathbb{R}^{d}$, i.e.,
symmetric operators w.r.t. to the canonical scalar product of $\mathbb{R}%
^{d} $ with positive eigenvalues.
\end{satz}

\begin{proof}
For all $\varphi \in C_{0}^{\infty }(\mathbb{R};\mathbb{R}^{d})$, observe
that its derivative $\varphi ^{\prime }\in C_{0}^{\infty }(\mathbb{R};%
\mathbb{R}^{d})$ satisfies
\begin{equation}
\int\nolimits_{\mathbb{R}}\varphi ^{\prime }\left( s\right) \mathrm{d}s=0\in
\mathbb{R}^{d}\ .  \label{distrib sympa 2}
\end{equation}%
As a consequence, we infer from Theorem \ref{main 1 copy(1)} (p) and the
equality%
\begin{equation}
\mathbf{\Xi }_{\mathrm{p}}\left( -t\right) =\mathbf{\Xi }_{\mathrm{p}}\left(
t\right) ^{\mathrm{t}},\qquad t\in \mathbb{R}\ ,  \label{t moins t Sigma p}
\end{equation}%
that, for any $\varphi \in C_{0}^{\infty }(\mathbb{R};\mathbb{R}^{d})$,%
\begin{eqnarray}
\frac{1}{2}\int\nolimits_{\mathbb{R}}\mathrm{d}s\int\nolimits_{\mathbb{R}}%
\mathrm{d}t\ \left\langle \varphi ^{\prime }\left( s\right) ,[\mathbf{\Xi }_{%
\mathrm{p}}\left( t-s\right) ]_{+}\varphi ^{\prime }\left( t\right)
\right\rangle _{\mathbb{R}^{d}} &=&  \label{positivity bochner} \\
\int_{t_{0}}^{t_{1}}\mathrm{d}s\int_{t_{0}}^{s}\mathrm{d}t\ \left\langle
\varphi ^{\prime }\left( s\right) ,\mathbf{\Xi }_{\mathrm{p}}\left(
t-s\right) \varphi ^{\prime }\left( t\right) \right\rangle _{\mathbb{R}^{d}}
&\geq &0\ .  \notag
\end{eqnarray}%
Note that (\ref{t moins t Sigma p}) is a simple consequence of the
stationarity of KMS states. By Theorem \ref{Theorem AC conductivity measure
copy(1)}, if (\ref{(3.1) NS})--(\ref{(3.2) NS}) and (\ref{static potential0}%
)--(\ref{(3.3) NS}) hold and the map (\ref{map}) is a random invariant
state, then $[\mathbf{\Xi }_{\mathrm{p}}]_{+}\in C^{2}\left( \mathbb{R},%
\mathcal{B}(\mathbb{R}^{d})\right) $. By integration by parts, it follows
from (\ref{positivity bochner}) that
\begin{equation*}
\int\nolimits_{\mathbb{R}}\mathrm{d}s\int\nolimits_{\mathbb{R}}\mathrm{d}%
t\left\langle \varphi \left( s\right) ,\partial _{t}^{2}[\mathbf{\Xi }_{%
\mathrm{p}}\left( t-s\right) ]_{+}\varphi \left( t\right) \right\rangle _{%
\mathbb{R}^{d}}\leq 0\ .
\end{equation*}%
By (\ref{symm-antisymm Theta}) and (\ref{t moins t Sigma p}), $[\mathbf{\Xi }%
_{\mathrm{p}}(-t)]_{+}=[\mathbf{\Xi }_{\mathrm{p}}(t)]_{+}$. Hence,
\begin{equation*}
\partial _{t}^{2}[\mathbf{\Xi }_{\mathrm{p}}(-t)]_{+}=\partial _{t}^{2}[%
\mathbf{\Xi }_{\mathrm{p}}(t)]_{+}
\end{equation*}%
for any $t\in \mathbb{R}$. Therefore, $-\partial _{t}^{2}[\mathbf{\Xi }_{%
\mathrm{p}}]_{+}:\mathbb{R}\rightarrow \mathcal{B}(\mathbb{R}^{d})$ is a
weakly positive definite continuous map that is symmetric w.r.t. time
reversal. Moreover, for any $t\in \mathbb{R}$, $\partial _{t}^{2}[\mathbf{%
\Xi }_{\mathrm{p}}(t)]_{+}$ is (by definition) a symmetric operator w.r.t.
the canonical scalar product of $\mathbb{R}^{d}$. Then, we can apply
Corollary \ref{Bochner thm2} with $\Upsilon =-\partial _{t}^{2}[\mathbf{\Xi }%
_{\mathrm{p}}]_{+}$ to deduce the existence of a unique finite and symmetric
$\mathcal{B}_{+}(\mathbb{R}^{d})$--valued measure $\mathbf{\mu }$ on $%
\mathbb{R}$ such that%
\begin{equation}
\partial _{t}^{2}[\mathbf{\Xi }_{\mathrm{p}}(t)]_{+}=-\int\nolimits_{\mathbb{%
R}}\cos \left( t\nu \right) \mathbf{\mu }\left( \mathrm{d}\nu \right) \ .
\label{derive seconde}
\end{equation}%
Observe that $[\mathbf{\Xi }_{\mathrm{p}}(0)]_{+}=\partial _{t}[\mathbf{\Xi }%
_{\mathrm{p}}(0)]_{+}=0$. Therefore, by integrating this last expression
twice, we then obtain that%
\begin{equation}
\lbrack \mathbf{\Xi }_{\mathrm{p}}(t)]_{+}=-\frac{t^{2}}{2}\mathbf{\mu }%
\left( \left\{ 0\right\} \right) -\int\nolimits_{0}^{t}\mathrm{d}%
s\int\nolimits_{0}^{s}\mathrm{d}\alpha \int\nolimits_{\mathbb{R}\backslash
\left\{ 0\right\} }\mathbf{\mu }\left( \mathrm{d}\nu \right) \cos \left(
\alpha \nu \right) \ .  \label{derive seconde2}
\end{equation}%
Since $\mathbf{\mu }$ is a a finite measure on $\mathbb{R}$, we can apply
twice the Fubini (--Tonelli) theorem to deduce that
\begin{eqnarray*}
\lbrack \mathbf{\Xi }_{\mathrm{p}}(t)]_{+} &=&-\frac{t^{2}}{2}\mathbf{\mu }%
\left( \left\{ 0\right\} \right) -\int\nolimits_{0}^{t}\mathrm{d}%
s\int\nolimits_{\mathbb{R}\backslash \left\{ 0\right\} }\mathbf{\mu }\left(
\mathrm{d}\nu \right) \left( \nu ^{-1}\sin \left( s\nu \right) \right) \\
&=&-\frac{t^{2}}{2}\mathbf{\mu }\left( \left\{ 0\right\} \right)
+\int\nolimits_{\mathbb{R}\backslash \left\{ 0\right\} }\left( \cos \left(
t\nu \right) -1\right) \nu ^{-2}\mathbf{\mu }\left( \mathrm{d}\nu \right) \ .
\end{eqnarray*}%
All integrals are of course well--defined because $\sin \left( \nu \right) =%
\mathcal{O}(\nu )$ and $1-\cos \left( \nu \right) =\mathcal{O}(\nu ^{2})$,
as $\nu \rightarrow 0$.
\end{proof}

From Theorem \ref{main 1 copy(1)} it is easy to see that the restriction of $%
\nu ^{-2}\mathbf{\mu }\left( \mathrm{d}\nu \right) $ on $\mathbb{R}%
\backslash \{0\}$ quantifies the heat production per unit volume due to the
component of frequency $\nu \in \mathbb{R}\backslash \{0\}$ of the electric
field in accordance with Joule's law in the AC--regime. By (\ref{positivity
bochner}), note at this point that the antisymmetric component $[\mathbf{\Xi
}_{\mathrm{p}}]_{-}$ of the paramagnetic conductivity does not contribute to
heat production. Therefore, we define this measure to be the (in--phase)\
\emph{AC--cond%
\-%
uctivity measure}:

\begin{definition}[AC--conductivity measure]
\label{def second law2 copy(2)}\mbox{ }\newline
We name $\mu _{\mathrm{AC}}\equiv \mu _{\mathrm{AC}}^{(\beta ,\vartheta
,\lambda )}$, the restriction of $\nu ^{-2}\mathbf{\mu }\left( \mathrm{d}\nu
\right) $ to $\mathbb{R}\backslash \{0\}$, the (in--phase)\emph{\ }AC--cond%
\-%
uctivity measure.
\end{definition}

\begin{bemerkung}[AC--conductivity measure from the 2nd law]
\label{AC--Conductivity Measure from the 2nd Law}\mbox{
}\newline
AC--Conductivity measures are obtained here for thermal equilibrium states
at strictly positive temperatures, that are, $(\tau ^{(\omega ,\vartheta
,\lambda )},\beta )$--KMS states with $\beta \in (0,\infty )$. See Theorem %
\ref{main 1} and Section \ref{Section dynamics}. The use of KMS states is
however not strictly necessary to get such measures: Theorem \ref{Theorem AC
conductivity measure} also holds for passive states $\varrho ^{(\omega )}$,
provided the map $\omega \mapsto \varrho ^{(\omega )}$ is a random invariant
state (Definition \ref{def second law2 copy(1)}). In other words,
AC--conductivity measures result from the 2nd law, only. This will be
discussed in more detail in a review article in preparation. In fact, in the
present paper, we have only considered KMS states to stick to \cite{OhmV}
where heat productions $\mathbf{Q}^{(\omega ,\mathbf{A})}$ are considered
and known to be well--defined for KMS states, see \cite[Definition 3.1,
Theorem 3.2]{OhmV}.
\end{bemerkung}

The AC--conductivity measure does not vanish in general, see, e.g., \cite[%
Theorem 4.7]{OhmIV}. Moreover, in the non--interacting case, we show in \cite%
[Theorem 4.1]{OhmIV} that $\mathbf{\mu }\left( \left\{ 0\right\} \right) =0$
and $\mu _{\mathrm{AC}}$ is a \emph{finite} measure on $\mathbb{R}\backslash
\{0\}$:%
\begin{equation*}
\Vert \mu _{\mathrm{AC}}\Vert _{\mathcal{B}(\mathbb{R}^{d})}\left( \mathbb{R}%
\backslash \{0\}\right) =\int\nolimits_{\mathbb{R}\backslash \{0\}}\nu
^{-2}\Vert \mathbf{\mu }\Vert _{\mathcal{B}(\mathbb{R}^{d})}\left( \mathrm{d}%
\nu \right) <\infty \ .
\end{equation*}%
In particular, the measure $\mathbf{\mu }([-\nu ,\nu ])$ is $\mathcal{O}(\nu
^{2})$ in the limit $\nu \rightarrow 0^{+}$. These properties are directly
related with the finiteness of current Duhamel fluctuations in the limit of
large space scales, which is not clear in presence of interactions, see (\ref%
{current fluct}) and discussion thereafter.

At high frequencies, by finiteness of the positive measure $\mathbf{\mu }$,
the AC--conductivity measure satisfies
\begin{equation}
\mu _{\mathrm{AC}}\left( \left[ \nu ,\infty \right) \right) \leq \nu ^{-2}%
\mathbf{\mu }\left( \left[ \nu ,\infty \right) \right) \leq \nu ^{-2}\mathbf{%
\mu }\left( \mathbb{R}\right) \ ,\qquad \nu \in \mathbb{R}^{+}\ .
\label{AC behavior}
\end{equation}%
The same property of course holds for negative frequencies, by symmetry of $%
\mathbf{\mu }$ (w.r.t. $\nu $). We can compare this property with the
corresponding one of the celebrated Drude model.

Indeed, the (in--phase) AC--conductivity measure obtained from the Drude
model is absolutely continuous w.r.t. the Lebesgue measure with the function
\begin{equation}
\nu \mapsto \vartheta _{\mathrm{T}}\left( \nu \right) \sim \frac{\mathrm{T}}{%
1+\mathrm{T}^{2}\nu ^{2}}  \label{drude function}
\end{equation}%
being the corresponding Radon--Nikodym derivative. Here, the \emph{%
relaxation time} $\mathrm{T}>0$ is related to the mean time interval between
two collisions of a charged carrier with defects in the crystal. This
function is the Fourier transform of the in--phase conductivity
\begin{equation*}
t\mapsto D\exp \left( -\mathrm{T}^{-1}\left\vert t\right\vert \right) \ ,
\end{equation*}%
where $D\in \mathbb{R}^{+}$ is some strictly positive constant. See for
instance \cite[Section 1]{OhmIV} for more discussions.

At high frequencies, Drude's approach \emph{heavily overestimates} the
AC--conduc%
\-%
tivity measure $\mu _{\mathrm{AC}}$ obtained from the more realistic\ model
studied here. Indeed, we can infer from (\ref{AC behavior}) that, in the
limit $\nu \rightarrow \infty $ of high frequencies,
\begin{equation}
\lim_{\nu \rightarrow \infty }\left\{ \nu ^{2}\mu _{\mathrm{AC}}\left( \left[
\nu ,\infty \right) \right) \right\} =0\ ,  \label{saymptotics}
\end{equation}%
whereas, by (\ref{drude function}), the corresponding quantity for the Drude
model diverges in the same limit:%
\begin{equation*}
\nu ^{2}\int_{\nu }^{\infty }\vartheta _{\mathrm{T}}\left( z\right) \mathrm{d%
}z\sim \nu ^{2}\arctan \left( \frac{1}{\mathrm{T}\nu }\right) =\mathcal{O}%
\left( \mathrm{T}^{-1}\nu \right) \ ,\qquad \text{as }\nu \rightarrow \infty
\ .
\end{equation*}%
The same behavior as for the Drude model holds for the AC--conductivity
measure obtained from the Lorentz--Drude model.

Hence, the asymptotics (\ref{saymptotics}) motivates the use of the
relaxation time as an effective $\nu $--dependent parameter of the Drude
model, i.e., one replaces $\mathrm{T}$ with $\mathrm{T}(\nu )$ in (\ref%
{drude function}), as observed for instance in \cite{T}. Indeed, with this
Ansatz and the asymptotics (\ref{saymptotics}), either $\mathrm{T}(\nu )$
vanishes faster than $\nu ^{-3}$ or it diverges faster than $\nu $, as $\nu
\rightarrow \infty $. Note that experimental measurements seem to indicate
that
\begin{equation*}
\mathrm{T}(\nu )=\frac{\mathrm{T}(0)}{1+D\mathrm{T}(0)\nu ^{2}}
\end{equation*}%
in some metals. See for instance \cite{T} for one experimental evidence of
this fact and \cite{NS1,NS2,SE,Y} for theoretical studies.

The concept of relaxation time or mean free path \cite{meanfreepath} (of
electrons) in the Drude model and its extensions is very intuitive. However,
the microscopic interpretation of this classical notion is difficult, in
particular if one has to take $\mathrm{T}$ as a $\nu $--dependent parameter.
Quoting meanwhile \cite[p. 24]{Anderson-physics}:\bigskip

\noindent \textit{Physicists had to wait for the discovery of quantum
mechanics to understand why electrons apparently do not scatter from ions
that occupy regular lattice sites: The wave character of an electron causes
the electron to diffract from an ideal crystal. Resistance appears only when
electrons scatter from imperfections in the crystal. With that quantum
mechanical revision, the Drude model can still be used, but in the new
picture an electron is envisaged as zigzagging between impurities.} \bigskip

\noindent Indeed, the average length an electron travels before it seems to
collide with an ion or defects in the crystal is experimentally measured in
metals to be about \emph{two order of magnitude} larger than the lattice
constant. [Note however that defects in our model are allowed to appear on
all lattice sites via the probability measure $\mathfrak{a}_{\Omega }$, see
Section \ref{sect 2.1}.]

The high frequency asymptotics of the (in--phase) AC--conductivity discussed
above makes explicit further problems with this classical picture. Observe
moreover that if the interparticle interaction has stronger polynomial decay
than in the assumptions of Theorem \ref{Theorem AC conductivity measure},
then the asymptotics (\ref{saymptotics}) can be improved by replacing $\nu
^{2}$ with $\nu ^{k}$ for an integer $k>2$. To show this, one uses
Lieb--Robinson bounds for multi--commutators \cite[Theorems 3.8--3.9]%
{brupedraLR} of order $k+1>3$ to get $\mathbf{\Xi }_{\mathrm{p}}\in
C^{k}\left( \mathbb{R},\mathcal{B}(\mathbb{R}^{d})\right) $. See also Remark %
\ref{remark cond decay}. However, we expect the model to physically break
down for frequencies $\nu $ corresponding to wavelengths (of light) of the
order of the lattice spacing. For usual materials, it would dovetail with
the frequency range of hard X--rays.

Similar to \cite[Corollary 3.5]{OhmV}, we deduce now general properties of
the paramagnetic conductivity from Theorem \ref{Theorem AC conductivity
measure}:

\begin{koro}[{Properties of $[\mathbf{\Xi }_{\mathrm{p}}]_{+}$}]
\label{Theorem AC conductivity measure2}\mbox{
}\newline
Assume all conditions of Theorem \ref{Theorem AC conductivity measure} and
let $\beta \in \mathbb{R}^{+}$ and $\vartheta ,\lambda \in \mathbb{R}%
_{0}^{+} $. Then, $[\mathbf{\Xi }_{\mathrm{p}}]_{+}\in C^{2}\left( \mathbb{R}%
,\mathcal{B}(\mathbb{R}^{d})\right) $ and the following holds:\newline
\emph{(i)} Time--reversal symmetry of $[\mathbf{\Xi }_{\mathrm{p}}]_{+}$: $[%
\mathbf{\Xi }_{\mathrm{p}}(0)]_{+}=0$ and%
\begin{equation*}
\lbrack \mathbf{\Xi }_{\mathrm{p}}(-t)]_{+}=[\mathbf{\Xi }_{\mathrm{p}%
}(t)]_{+}\ ,\qquad t\in \mathbb{R}\ .
\end{equation*}%
\emph{(ii)} Negativity of $[\mathbf{\Xi }_{\mathrm{p}}]_{+}$:
\begin{equation*}
-[\mathbf{\Xi }_{\mathrm{p}}(t)]_{+}\in \mathcal{B}_{+}(\mathbb{R}^{d})\
,\qquad t\in \mathbb{R}\ .
\end{equation*}%
\emph{(iii)} Ces\`{a}ro mean of $[\mathbf{\Xi }_{\mathrm{p}}]_{+}$: If $%
\mathbf{\mu }\left( \left\{ 0\right\} \right) =0$ and $\Vert \mu _{\mathrm{AC%
}}\Vert _{\mathcal{B}(\mathbb{R}^{d})}\left( \mathbb{R}\backslash
\{0\}\right) <\infty $ then
\begin{equation*}
\underset{t\rightarrow \infty }{\lim }\ \frac{1}{t}\int_{0}^{t}[\mathbf{\Xi }%
_{\mathrm{p}}(s)]_{+}\mathrm{d}s=-\mu _{\mathrm{AC}}\left( \mathbb{R}%
\backslash \left\{ 0\right\} \right) \ .
\end{equation*}
\end{koro}

\begin{proof}
(i)--(iii) are direct consequences of Theorem \ref{Theorem AC conductivity
measure}, the Fubini (--Tonelli) theorem and Lebesgue's dominated
convergence theorem.
\end{proof}

Assuming (\ref{(3.1) NS})--(\ref{(3.2) NS}), note that, for any $l,\beta \in
\mathbb{R}^{+}$, $\omega \in \Omega $ and $\vartheta ,\lambda \in \mathbb{R}%
_{0}^{+}$, there exists\footnote{$\mu _{\mathrm{p},l}^{(\omega )}$ is a
finite measure because we take KMS states. For passive states, we only have
the existence of finite volume AC--conductivity measures, similar to Theorem %
\ref{Theorem AC conductivity measure} and Definition \ref{def second law2
copy(2)} for $l\in \mathbb{R}^{+}$.} a (generally non--zero) symmetric and
finite $\mathcal{B}_{+}(\mathbb{R}^{d})$--valued measure $\mu _{\mathrm{p}%
,l}^{(\omega )}\equiv \mu _{\mathrm{p},l}^{(\beta ,\omega ,\vartheta
,\lambda )}$ on $\mathbb{R}$ such that%
\begin{equation}
\lbrack \Xi _{\mathrm{p},l}^{(\omega )}(t)]_{+}=\int_{\mathbb{R}}\left( \cos
\left( t\nu \right) -1\right) \mu _{\mathrm{p},l}^{(\omega )}(\mathrm{d}\nu
)\ ,\qquad t\in \mathbb{R}\ .  \label{micro conduc}
\end{equation}%
Away from $\nu =0$ and as $l\rightarrow \infty $ the finite microscopic
conductivity measure $\mu _{\mathrm{p},l}^{(\omega )}$ converges in the weak$%
^{\ast }$--topology to the macroscopic AC--conductivity measure $\mu _{%
\mathrm{AC}}$:

\begin{satz}[From microscopic to macroscopic AC--conductivity measures]
\label{Theorem AC conductivity measure copy(2)}Assume Conditions (\ref{(3.1)
NS})--(\ref{(3.2) NS}), (\ref{static potential0}), (\ref{(3.3) NS}) with $%
\varsigma >3d$, and that the map (\ref{map}) is a random invariant state.
Let $\beta \in \mathbb{R}^{+}$ and $\vartheta ,\lambda \in \mathbb{R}%
_{0}^{+} $. Then, $\mathbf{\Xi }_{\mathrm{p}}\in C^{3}\left( \mathbb{R},%
\mathcal{B}(\mathbb{R}^{d})\right) $ and there is a measurable subset $%
\tilde{\Omega}\equiv \tilde{\Omega}^{(\beta ,\vartheta ,\lambda )}\subset
\Omega $ of full measure such that, for all $\omega \in \tilde{\Omega}$:%
\newline
\emph{(i)} Tightness: The sequence $\{\mu _{\mathrm{p},l}^{(\omega
)}\}_{l\in \mathbb{R}^{2}}$ of finite measures is tight.\newline
\emph{(ii)} Weak$^{\ast }$--convergence away from $\nu =0$: For any $\vec{w}%
_{1},\vec{w}_{2}\in \mathbb{R}^{d}$ and any bounded continuous function $f$
on $\mathbb{R}\backslash \{0\}$ with $0\notin \overline{\mathrm{supp}f}$,%
\begin{equation*}
\lim_{l\rightarrow \infty }\int_{\mathbb{R}}f\left( \nu \right) \nu
^{2}\langle \vec{w}_{1},\mu _{\mathrm{p},l}^{(\omega )}\left( \mathrm{d}\nu
\right) \vec{w}_{2}\rangle _{\mathbb{R}^{d}}=\int_{\mathbb{R}\backslash
\{0\}}f\left( \nu \right) \nu ^{2}\langle \vec{w}_{1},\mu _{\mathrm{AC}%
}\left( \mathrm{d}\nu \right) \vec{w}_{2}\rangle _{\mathbb{R}^{d}}\ .
\end{equation*}
\end{satz}

\begin{proof}
Fix $\beta \in \mathbb{R}^{+}$ and $\vartheta ,\lambda \in \mathbb{R}%
_{0}^{+} $. Under assumptions of the theorem, $\mathbf{\Xi }_{\mathrm{p}}\in
C^{3}\left( \mathbb{R},\mathcal{B}(\mathbb{R}^{d})\right) $ and there is a
measurable subset $\tilde{\Omega}\equiv \tilde{\Omega}^{(\beta ,\vartheta
,\lambda )}\subset \Omega $ of full measure such that%
\begin{equation}
\partial _{t}^{2}[\Xi _{\mathrm{p},l}^{(\omega )}(t)]_{+}=\underset{%
l\rightarrow \infty }{\lim }\partial _{t}^{2}[\Xi _{\mathrm{p},l}^{(\omega
)}(t)]_{+}\ ,\qquad t\in \mathbb{R}\ .  \label{pointwise limit derive 2}
\end{equation}%
The proof is omitted as the arguments are very similar to those proving
Theorem \ref{thm charged transport coefficient} (p). Note only that
Condition (\ref{(3.3) NS}) with $\varsigma >3d$ is imposed to obtain
Lieb--Robinson bounds for multi--commutators \cite[Theorems 3.8--3.9]%
{brupedraLR} of order \emph{four}. This is needed to obtain the
equicontinuity of the family
\begin{equation*}
\left\{ t\mapsto \partial _{t}^{2}[\Xi _{\mathrm{p},l}^{(\omega
)}(t)]_{+}\right\} _{l\in \mathbb{R}^{+},\omega \in \Omega }
\end{equation*}%
of functions of time. See for instance Remark \ref{remark cond decay}, the
proofs of Theorems \ref{Theorem AC conductivity measure copy(1)} and \ref%
{main 1 copy(2)}.

Meanwhile, for $l\in \mathbb{R}^{+}$, we apply twice the Fubini (--Tonelli)
theorem to deduce that
\begin{equation}
\partial _{t}^{2}[\Xi _{\mathrm{p},l}^{(\omega )}(t)]_{+}=-\int\nolimits_{%
\mathbb{R}}\cos \left( t\nu \right) \mathbf{\mu }_{l}^{(\omega )}(\mathrm{d}%
\nu )\ ,\qquad t\in \mathbb{R}\ ,  \label{pointwise limit derive 3}
\end{equation}%
with $\mathbf{\mu }_{l}^{(\omega )}:=\nu ^{2}\mu _{\mathrm{p},l}^{(\omega )}$%
. Observe from (\ref{derive seconde}) and (\ref{pointwise limit derive 2})--(%
\ref{pointwise limit derive 3}) that $\mathbf{\mu }_{l}^{(\omega )}$ is a
finite measure and
\begin{equation}
\lim_{l\rightarrow \infty }\mathbf{\mu }_{l}^{(\omega )}(\mathbb{R}%
)=\lim_{l\rightarrow \infty }\partial _{t}^{2}[\Xi _{\mathrm{p},l}^{(\omega
)}(0)]_{+}=\partial _{t}^{2}[\mathbf{\Xi }_{\mathrm{p}}\left( 0\right) ]_{+}=%
\mathbf{\mu }(\mathbb{R})\in \mathcal{B}_{+}(\mathbb{R}^{d})\ .
\label{pointwise limit derive 4}
\end{equation}

Now, take any vector $\vec{w}\in \mathbb{R}^{d}$. Let $\mathbf{\mu }_{l,\vec{%
w}}^{(\omega )}$ and $\mathbf{\mu }_{\vec{w}}$ be the measures on $\mathbb{R}
$ respectively defined, for any Borel set $\mathcal{X}\subset \mathbb{R}$,
by
\begin{equation*}
\mathbf{\mu }_{l,\vec{w}}^{(\omega )}\left( \mathcal{X}\right) :=\langle
\vec{w},\mathbf{\mu }_{l}^{(\omega )}\left( \mathcal{X}\right) \vec{w}%
\rangle _{\mathbb{R}^{d}}\quad \text{and}\quad \mathbf{\mu }_{\vec{w}}\left(
\mathcal{X}\right) :=\langle \vec{w},\mathbf{\mu }\left( \mathcal{X}\right)
\vec{w}\rangle _{\mathbb{R}^{d}}\ .
\end{equation*}%
Assume w.l.o.g. that $\mathbf{\mu }_{\vec{w}}(\mathbb{R})>0$. Then, by
combining (\ref{derive seconde}) and (\ref{micro conduc})--(\ref{pointwise
limit derive 4}) with $\partial _{t}^{2}[\mathbf{\Xi }_{\mathrm{p}}]_{+}\in
C(\mathbb{R},\mathcal{B}(\mathbb{R}^{d}))$ (Theorem \ref{Theorem AC
conductivity measure copy(1)}) and \cite[Theorems 3.2.3 and 3.3.6]{Durrett},
we deduce that, on the subset $\tilde{\Omega}$ of full measure, the sequence
$\{\mathbf{\mu }_{l,\vec{w}}^{(\omega )}\}_{l\in \mathbb{R}^{2}}$ is tight
and converges in the weak$^{\ast }$--topology to $\mathbf{\mu }_{\vec{w}}$,
as $l\rightarrow \infty $. By Definition \ref{def second law2 copy(2)}, this
implies Assertion (ii) for $\vec{w}_{1}=\vec{w}_{2}$. Its extension to
arbitrary vectors $\vec{w}_{1},\vec{w}_{2}\in \mathbb{R}^{d}$ is a
consequence of the polarization identity, see, e.g., (\ref{polarization
identity}). Assertion (i) easily follows from the tightness of $\{\mathbf{%
\mu }_{l,\vec{w}}^{(\omega )}\}_{l\in \mathbb{R}^{2}}$ for $\vec{w}\in
\mathbb{R}^{d}$ and the polarization identity.
\end{proof}

\subsection{Time--Reversal Invariance of Random Equilibrium States\label%
{sect time ref}}

In this subsection we define time--reversal invariance of \emph{random}
fermion systems and derive its consequences on conductivity. We do not
define this symmetry of random systems as the almost surely time--reversal
invariance. But instead, we give a weaker, and hence more general, notion of
\textquotedblleft time--reversal invariance in average\textquotedblright .
This is done in the same spirit of what we do above to introduce translation
invariance for disordered systems at thermal equilibrium. See, for instance,
Definition \ref{def second law2 copy(1)}. In fact, by doing this, we allow
for a large class of random magnetic potentials.

Let $\mathcal{X}$ be a $C^{\ast }$--algebra with unity $\mathbf{1}$ and
assume the existence of a map $\Theta :\mathcal{X}\rightarrow \mathcal{X}$
with the following properties:

\begin{itemize}
\item $\Theta $ is antilinear and continuous.

\item $\Theta \left( \mathbf{1}\right) =\mathbf{1}$ and $\Theta \circ \Theta
=\mathrm{Id}_{\mathcal{X}}$.

\item $\Theta \left( B_{1}B_{2}\right) =\Theta \left( B_{1}\right) \Theta
\left( B_{2}\right) $ for all $B_{1},B_{2}\in \mathcal{X}$.

\item $\Theta \left( B^{\ast }\right) =\Theta \left( B\right) ^{\ast }$ for
all $B\in \mathcal{X}$.
\end{itemize}

\noindent Such a map is called a \emph{time--reversal} operation of the $%
C^{\ast }$--algebra $\mathcal{X}$. For $\mathcal{X}=\mathcal{U}$ (CAR $%
C^{\ast }$--algebra of the lattice $\mathfrak{L}$), there is a natural
time--reversal operation $\mathfrak{T}$, which is uniquely defined by the
condition%
\begin{equation}
\mathfrak{T}(a_{x})=a_{x}\qquad x\in \mathfrak{L}\ .  \label{tr 1}
\end{equation}%
See also \cite[Section 2.1.4]{OhmII}.

For any strongly continuous one--parameter group $\tau :=\{\tau _{t}\}_{t\in
{\mathbb{R}}}$ of $\ast $--auto%
\-%
morphisms of $\mathcal{X}$, the family $\tau ^{\Theta }:=\{\tau _{t}^{\Theta
}\}_{t\in {\mathbb{R}}}$ defined by
\begin{equation*}
\tau _{t}^{\Theta }:=\Theta \circ \tau _{t}\circ \Theta \ ,\qquad t\in {%
\mathbb{R}}\ ,
\end{equation*}%
is again a strongly continuous one--parameter group of automorphisms.
Similarly, for any state $\rho \in \mathcal{X}^{\ast }$, the linear
functional $\rho ^{\Theta }$ defined by%
\begin{equation*}
\rho ^{\Theta }\left( B\right) =\overline{\rho \circ \Theta \left( B\right) }%
\ ,\qquad B\in \mathcal{X}\ ,
\end{equation*}%
is again a state. We say that $\tau $\ and $\rho $ are \emph{time--reversal
invariant} w.r.t. $\Theta $ if they satisfy $\tau _{t}^{\Theta }=\tau _{-t}$
for all $t\in \mathbb{R}$ and $\rho ^{\Theta }=\rho $. If $\tau $\ is
time--reversal invariant then, for all $\beta \in \mathbb{R}^{+}$, there is
at least one time--reversal invariant $(\tau ,\beta )$--KMS state $\varrho
\in \mathcal{X}^{\ast }$, provided the set of $(\tau ,\beta )$--KMS states
is not empty. This follows from the convexity of the set of KMS states, see
\cite[Lemma A.12]{OhmII}.

Now, we introduce a notion of time--reversal invariance for the random
system considered here. If $\Psi $ is an interaction, we call it
time--reversal invariant whenever
\begin{equation*}
\mathfrak{T}(\Psi _{\Lambda })=\Psi _{\Lambda },\qquad \Lambda \in \mathcal{P%
}_{f}(\mathfrak{L})\ .
\end{equation*}%
For any $\omega =(\omega _{1},\omega _{2})\in \Omega $, we define $\overline{%
\omega }:=(\omega _{1},\overline{\omega _{2}})\in \Omega $, where%
\begin{equation*}
\overline{\omega _{2}}(b):=\overline{\omega _{2}(b)},\qquad b\in \mathfrak{b}%
\ .
\end{equation*}%
We say that the random state (\ref{map}) is time--reversal symmetric if, for
all $\omega \in \Omega $,%
\begin{equation*}
\rho ^{(\beta ,\overline{\omega },\vartheta ,\lambda )}=[\rho ^{(\beta
,\omega ,\vartheta ,\lambda )}]^{\mathfrak{T}}\ .
\end{equation*}%
Similarly, we call the random dynamic (\ref{tho}) on $\mathcal{U}$
time--reversal symmetric if, for all $\omega \in \Omega $,%
\begin{equation*}
\tau _{-t}^{(\overline{\omega },\vartheta ,\lambda )}=\mathfrak{T}\circ \tau
_{t}^{(\omega ,\vartheta ,\lambda )}\circ \mathfrak{T}\ ,\qquad t\in {%
\mathbb{R}}\ .
\end{equation*}%
It is not difficult to see that, if the interparticle interaction $\Psi ^{%
\mathrm{IP}}$ is time--reversal invariant then the (unperturbed)\ random
dynamics $\tau ^{(\omega ,\vartheta ,\lambda )}$ is time--reversal symmetric
in the above sense for any $\vartheta ,\lambda \in \mathbb{R}_{0}^{+}$.
Further, we say that the $\Omega $--valued random variable $\omega $, the
distribution of which is given by the probability space $(\Omega ,\mathfrak{A%
}_{\Omega },\mathfrak{a}_{\Omega })$, is time--reversal invariant if the map
$\omega \mapsto \overline{\omega }$ is measurable w.r.t. $\mathfrak{A}%
_{\Omega }$ and preserves the measure $\mathfrak{a}_{\Omega }$.

Like in the case of translation invariance, the existence of random
invariant thermal equilibrium states which are time--reversal symmetric in
the above sense is not clear in general. If the $(\tau ^{(\omega ,\vartheta
,\lambda )},\beta )$--KMS state is unique and $\Psi ^{\mathrm{IP}}$ is
time--reversal invariant, then the (unique) map (\ref{map}) is a random
state which is time--reversal symmetric. The arguments to prove this are
similar to the ones used in the proof of \cite[Lemma A.12]{OhmII}. As
already discussed, if (\ref{static potential0}) holds then (\ref{map}) is,
moreover, a random invariant state. See Section \ref{Section dynamics}.

Time--reversal invariance implies the following important properties of
charge transport coefficients related to the models considered here:

\begin{satz}[Consequences of time--reversal symmetry]
\mbox{
}\newline
Assume (\ref{(3.1) NS})--(\ref{(3.2) NS}), (\ref{static potential0})--(\ref%
{(3.3) NS}), time--reversal invariance of the interparticle interaction $%
\Psi ^{\mathrm{IP}}$ and the ($\Omega $--valued) random variable $\omega $,
as well as that the map (\ref{map}) is a random invariant state which is
time--reversal symmetric. Let $\beta \in \mathbb{R}^{+}$ and $\vartheta
,\lambda \in \mathbb{R}_{0}^{+}$. Then, the following assertions hold true:%
\newline
\emph{(th)} Vanishing thermal current density:%
\begin{equation*}
\underset{l\rightarrow \infty }{\lim }\left\{ \mathbb{J}_{\mathrm{th}%
}^{(\omega ,l)}\right\} _{k}=\mathbb{E}\left[ \varrho ^{(\beta ,\omega
,\vartheta ,\lambda )}(I_{\left( e_{k},0\right) }^{(\omega ,\vartheta )})%
\right] =0\ ,\qquad k\in \{1,\ldots ,d\}\ .
\end{equation*}%
\emph{(p)} Vanishing antisymmetric part of the paramagnetic conductivity:\
\begin{equation*}
\lbrack \mathbf{\Xi }_{\mathrm{p}}\left( t\right) ]_{-}=0\ ,\qquad t\in
\mathbb{R}\ .
\end{equation*}
\end{satz}

\begin{proof}
(th)\emph{\ }directly follows form Theorem \ref{main 1 copy(8)} (th), the
equality $\mathfrak{T}(I_{\left( e_{k},0\right) }^{(\omega ,\vartheta
)})=-I_{\left( e_{k},0\right) }^{(\overline{\omega },\vartheta )}$, which is
a consequence of (\ref{tr 1}), $\varrho ^{(\beta ,\omega ,\vartheta ,\lambda
)}(I_{\left( e_{k},0\right) }^{(\omega ,\vartheta )})\in \mathbb{R}$, the
time--reversal invariance of the random variable $\omega $ and the
time--reversal symmetry of the random state $\varrho ^{(\beta ,\omega
,\vartheta ,\lambda )}$. These facts combined with the time--reversal
symmetry of $\tau ^{(\omega ,\vartheta ,\lambda )}$, which follows from the
assumptions on $\Psi ^{\mathrm{IP}}$, and the stationarity of KMS states
imply (p).
\end{proof}

\section{Epilogue: AC--Conductivity and L\'{e}vy Processes\label{epilogue}}

By Theorem \ref{Theorem AC conductivity measure}, charge transport
properties of interacting fermions in disordered media are governed by a L%
\'{e}vy measure. This suggests an alternative effective description of the
phenomenon of linear conductivity by using L\'{e}vy Processes in Fourier
space. It is a very interesting mathematical result since L\'{e}vy
statistics turn out to efficiently describe quantum phenomena. Indeed,
quantum Monte-Carlo methods have already permitted to observe that certain
quantum processes obey L\'{e}vy statistics. Moreover, a relation between
quantum systems and (classical) stochastic processes has also been
experimentally observed: For instance, in quantum optics, the (subrecoil)
cooling process of atom in presence of laser radiation can be modeled by a L%
\'{e}vy process \cite{9780511755668} with so--called quantum jumps in
momentum space (w.r.t. \emph{space} variables). This gives very good
agreements with experimental measurements, see \cite[Chap. 8]{9780511755668}%
. However, as far as we know, there is no rigorous derivation of this fact
from quantum mechanics. Thus, this section is written to propose an approach
to that issue and suggest a L\'{e}vy processes that could be behind the
phenomenon of linear conductivity.

For simplicity, we assume that the paramagnetic conductivity $\mathbf{\Xi }_{%
\mathrm{p}}$ is of the form $\mathbf{\sigma }_{\mathrm{p}}\mathbf{1}_{%
\mathbb{R}^{d}}$ with $\mathbf{\sigma }_{\mathrm{p}}$ being a real--valued
function of time. In particular, $\mathbf{[\Xi }_{\mathrm{p}}]_{-}=0$ and,
by (\ref{t moins t Sigma p}), $\mathbf{\sigma }_{\mathrm{p}}(t)=\mathbf{%
\sigma }_{\mathrm{p}}(-t)$ for any $t\in \mathbb{R}$ with $\mathbf{\sigma }_{%
\mathrm{p}}(0)=0$. This property of $\mathbf{\Xi }_{\mathrm{p}}$ holds true,
for instance, if the random variables $\left\{ \left( \omega _{1}\left(
x\right) ,\omega _{2}\left( b\right) \right) \right\} _{x\in \mathfrak{L}%
,b\in \mathfrak{b}}$ are independently and identically distributed and the
interparticle interaction $\Psi ^{\mathrm{IP}}\in \mathcal{W}$ has the form%
\begin{equation*}
\Psi _{\Lambda }^{\mathrm{IP}}=v\left( \left\vert x-y\right\vert \right)
a_{x}^{\ast }a_{x}a_{y}^{\ast }a_{y}
\end{equation*}%
whenever $\Lambda =\left\{ x,y\right\} $ for $x,y\in \mathfrak{L}$, and $%
\Psi _{\Lambda }^{\mathrm{IP}}=0$ when $\left\vert \Lambda \right\vert >2$.
Here, $v\left( r\right) :\mathbb{R}_{0}^{+}\rightarrow \mathbb{R}^{+}$ is a
real--valued function such that%
\begin{equation*}
\underset{r\in \mathbb{R}_{0}^{+}}{\sup }\left\{ \frac{v\left( r\right) }{%
\mathbf{F}\left( r\right) }\right\} <\infty \ .
\end{equation*}%
See \cite[Lemma 5.23]{OhmIII} for more details.

In this case, by Theorem \ref{Theorem AC conductivity measure}, for any $%
\beta \in \mathbb{R}^{+}$ and $\vartheta ,\lambda \in \mathbb{R}_{0}^{+}$,
there is a unique finite and symmetric $\mathbb{R}$--valued measure $%
\mathfrak{m}_{\mathrm{AC}}$ on $\mathbb{R}\backslash \left\{ 0\right\} $
such that, for any $\alpha \in \mathbb{R}$,
\begin{equation}
\mathbf{\sigma }_{\mathrm{p}}\left( \alpha \right) =-\frac{\alpha ^{2}}{2}%
D_{\left\{ 0\right\} }+\int\nolimits_{\mathbb{R}\backslash \left\{ 0\right\}
}\left( \mathrm{e}^{i\alpha \nu }-1-i\alpha \nu \mathbf{1}\left[ \left\vert
\nu \right\vert <1\right] \right) \mathfrak{m}_{\mathrm{AC}}\left( \mathrm{d}%
\nu \right) \ ,  \label{Levy measure0}
\end{equation}%
where $\mathbf{\mu }\left( \left\{ 0\right\} \right) =D_{\left\{ 0\right\} }%
\mathbf{1}_{\mathbb{R}^{d}}$ with $D_{\left\{ 0\right\} }\in \mathbb{R}$ and
\begin{equation}
\int\nolimits_{\mathbb{R}}\left( 1\wedge \nu ^{2}\right) \mathfrak{m}_{%
\mathrm{AC}}\left( \mathrm{d}\nu \right) <\infty \ .  \label{Levy measure1}
\end{equation}%
Equations (\ref{Levy measure0})--(\ref{Levy measure1}) correspond to the L%
\'{e}vy--Khintchine representation of the function $\mathbf{\sigma }_{%
\mathrm{p}}$. Observe that $\mu _{\mathrm{AC}}=\mathfrak{m}_{\mathrm{AC}}%
\mathbf{1}_{\mathbb{R}^{d}}$ and by a slight abuse of terminology, we also
name $\mathfrak{m}_{\mathrm{AC}}$ the AC--conductivity measure.

Therefore, by \cite[Theorem 2.1.]{Kyp}, there is a probability space $%
(\Omega _{L},\mathfrak{F},\mathbb{P})$ on which a $\mathbb{R}$--valued L\'{e}%
vy process $\digamma =\left\{ \digamma _{t}:t\in \mathbb{R}_{0}^{+}\right\} $
with characteristic exponent $\mathbf{\sigma }_{\mathrm{p}}$ (up to a minus
sign) exists. More explicitly,
\begin{equation*}
\mathbb{E}_{\mathbb{P}}\left[ \exp \left( i\alpha \digamma _{t}\right) %
\right] =\exp \left( t\mathbf{\sigma }_{\mathrm{p}}\left( \alpha \right)
\right) \ ,\qquad \alpha \in \mathbb{R},\ t\in \mathbb{R}_{0}^{+}\ ,
\end{equation*}%
with $\mathbb{E}_{\mathbb{P}}[\ \cdot \ ]$ being the expectation value
associated with the probability measure $\mathbb{P}$. In this context, $%
\mathfrak{m}_{\mathrm{AC}}$ is called the L\'{e}vy measure of $\digamma $.
It describes the jumps of $\digamma $. In other words, similar to laser
cooling \cite{9780511755668}, such a L\'{e}vy Process describes quantum
jumps in Fourier space (but w.r.t. \emph{time} coordinates instead of
position coordinates as in sub-recoil laser cooling). For a comprehensive
account on L\'{e}vy processes, see for instance \cite{bertoin,Kyp} and
references therein.

By (\ref{Levy measure0}), $\digamma $ has no drift but a diffusion component
when $D_{\left\{ 0\right\} }>0$. There is also a Poisson random measure $N$
(see, e.g., \cite[Definition 2.3.]{Kyp}) distributed on
\begin{equation*}
\left( \mathbb{R}_{0}^{+}\times \mathbb{R}\backslash \left\{ 0\right\} ,%
\mathfrak{A}_{\mathbb{R}_{0}^{+}\times \mathbb{R}\backslash \left\{
0\right\} }\right) \ ,
\end{equation*}%
$\mathfrak{A}_{\mathbb{R}_{0}^{+}\times \mathbb{R}\backslash \left\{
0\right\} }$ being the Borel $\sigma $--algebra of $\mathbb{R}_{0}^{+}\times
\mathbb{R}\backslash \left\{ 0\right\} $, with characteristic measure (or
intensity)$\ \mathfrak{m}_{\mathrm{AC}}$ such that
\begin{equation}
\digamma _{t}=\sqrt{D_{\left\{ 0\right\} }}B_{t}+\int\nolimits_{0}^{t}\int%
\nolimits_{\left\vert \nu \right\vert \geq 1}\nu N\left( \mathrm{d}s\mathrm{d%
}\nu \right) +\int\nolimits_{0}^{t}\int\nolimits_{\left\vert \nu \right\vert
<1}\nu M\left( \mathrm{d}s\mathrm{d}\nu \right) \ ,\qquad t\in \mathbb{R}%
_{0}^{+}\ .  \label{levy processes}
\end{equation}%
Here, $M$ is the associated martingale measure%
\begin{equation*}
M\left( \mathrm{d}s\mathrm{d}\nu \right) :=N\left( \mathrm{d}s\mathrm{d}\nu
\right) -\mathrm{d}s\mathfrak{m}_{\mathrm{AC}}\left( \mathrm{d}\nu \right)
\end{equation*}%
and $B$ is a Brownian motion. The second term in the r.h.s. of (\ref{levy
processes}) is a compound Poisson process with rate $\mathfrak{m}_{\mathrm{AC%
}}(\mathbb{R}\backslash (-1,1))$ and jump distribution
\begin{equation*}
\left( \mathfrak{m}_{\mathrm{AC}}(\mathbb{R}\backslash (-1,1))\right) ^{-1}%
\mathfrak{m}_{\mathrm{AC}}\ ,
\end{equation*}%
provided $\mathfrak{m}_{\mathrm{AC}}(\mathbb{R}\backslash (-1,1))>0$. The
third term in the r.h.s. of (\ref{levy processes}) is another L\'{e}vy
process, which is a square integrable martingale on the same probability
space. It is the uniform limit $\varepsilon \rightarrow 0^{+}$ (along an
appropriate deterministic subsequence) on compacta of the compound Poisson
process with drift
\begin{equation*}
\int\nolimits_{0}^{t}\int\nolimits_{\varepsilon \leq \left\vert \nu
\right\vert <1}\nu N\left( \mathrm{d}s\mathrm{d}\nu \right)
-t\int\nolimits_{\varepsilon \leq \left\vert \nu \right\vert <1}\nu
\mathfrak{m}_{\mathrm{AC}}\left( \mathrm{d}\nu \right) \ ,\qquad t\in
\mathbb{R}_{0}^{+},\ \varepsilon \in \left( 0,1\right) \ .
\end{equation*}%
The limit L\'{e}vy process can also be seen as a superposition of an
infinite number of compound Poisson processes with drift, see for instance
\cite[Section 2.5]{Kyp}.

When
\begin{equation}
0<\mathfrak{m}_{\mathrm{AC}}\left( \mathbb{R}\backslash \left\{ 0\right\}
\right) <\infty \qquad \text{and}\qquad D_{\left\{ 0\right\} }=0\ ,
\label{finiteness}
\end{equation}%
$\digamma _{t}$ is a compound Poisson process with rate $\mathfrak{m}_{%
\mathrm{AC}}(\mathbb{R}\backslash \left\{ 0\right\} )$ and jump distribution
\begin{equation}
\left( \mathfrak{m}_{\mathrm{AC}}(\mathbb{R}\backslash \left\{ 0\right\}
)\right) ^{-1}\mathfrak{m}_{\mathrm{AC}}\ .  \label{probability measure}
\end{equation}%
See \cite[Lemma 2.13]{Kyp}. In particular, the AC--conductivity measure $%
\mathfrak{m}_{\mathrm{AC}}$ describes the jump structure of the symmetric L%
\'{e}vy process $\digamma $ in the frequency domain $\mathbb{R}$.

As an example, we can take the AC--conductivity measure obtained from the
Drude model. This measure is absolutely continuous w.r.t. the Lebesgue
measure with Radon--Nikodym density $\vartheta _{\mathrm{T}}$ defined by (%
\ref{drude function}). Recall that the relaxation time $\mathrm{T}>0$ is the
mean time interval between two collisions of a charged carrier with defects
in the crystal. For all $\mathrm{T}>0$, the measure of the full set $\mathbb{%
R}\backslash \left\{ 0\right\} $ equals $\Vert \vartheta _{\mathrm{T}}\Vert
_{1}=D$. In particular, the mean time between frequency jumps does not
depend on $\mathrm{T}>0$ in this new classical process. In the limit $%
\mathrm{T}\rightarrow 0^{+}$ of perfect isolator $\vartheta _{\mathrm{T}%
}\rightarrow 0$ uniformly on $\mathbb{R}$ while in the limit $\mathrm{T}%
\rightarrow \infty $ of perfect conductor $\vartheta _{\mathrm{T}%
}\rightarrow 0$ uniformly on $\mathbb{R}\backslash \lbrack -\varepsilon
,\varepsilon ]$ for any $\varepsilon >0$. Hence, by a similar expression to (%
\ref{probability measure}) for the Drude model and because of (\ref{drude
function}), the probability of large (frequency) jumps increases in the
limit $\mathrm{T}\rightarrow 0^{+}$ (isolator limit), but decreases when $%
\mathrm{T}\rightarrow \infty $ (conductor limit). The stochastic process $%
\digamma $ gives an alternative classical picture to electrical conduction.

\section{Technical Proofs\label{Section technical proof Ohm-VI}}

\subsection{Study of the Paramagnetic Conductivity}

Lieb--Robinson bounds and their extensions \cite{brupedraLR} to
multi--commutators are here pivotal mathematical tools.

For any $\vartheta _{0},\lambda \in \mathbb{R}_{0}^{+}$, $\vartheta \in
\lbrack 0,\vartheta _{0}]$, $\omega \in \Omega $, $t\in \mathbb{R}$, $%
B_{1}\in \mathcal{U}^{+}\cap \mathcal{U}_{\Lambda ^{(1)}}$ and $B_{2}\in
\mathcal{U}_{\Lambda ^{(2)}}$ with disjoint sets $\Lambda ^{(1)},\Lambda
^{(2)}\in \mathcal{P}_{f}(\mathfrak{L})$,
\begin{eqnarray}
\left\Vert \left[ \tau _{t}^{(\omega ,\vartheta ,\lambda )}\left(
B_{1}\right) ,B_{2}\right] \right\Vert _{\mathcal{U}} &\leq &2\mathbf{D}%
^{-1}\left\Vert B_{1}\right\Vert _{\mathcal{U}}\left\Vert B_{2}\right\Vert _{%
\mathcal{U}}\left( \mathrm{e}^{2\mathbf{D}\left\vert t\right\vert
D_{\vartheta _{0}}}-1\right)  \label{Lieb--Robinson bounds simplified} \\
&&\times \sum_{x\in \Lambda ^{(1)}}\sum_{y\in \Lambda ^{(2)}}\mathbf{F}%
\left( \left\vert x-y\right\vert \right) \ .  \notag
\end{eqnarray}%
This is the usual Lieb--Robinson bound. See, e.g., \cite[Theorem 2.1 (iii)]%
{OhmV}. Here, the real constant $D_{\vartheta _{0}}$ is defined, for any $%
\vartheta _{0}\in \mathbb{R}_{0}^{+}$, by%
\begin{equation}
D_{\vartheta _{0}}:=\sup \left\{ \left\Vert \Psi ^{(\omega ,\vartheta
)}\right\Vert _{\mathcal{W}}:\omega \in \Omega ,\ \vartheta \in \lbrack
0,\vartheta _{0}]\right\} <\infty \ .  \label{norm sup interaction}
\end{equation}%
See Sections \ref{sect 2.1} and \ref{Section dynamics}. As a consequence,
the paramagnetic transport coefficient $\sigma _{\mathrm{p}}^{(\omega )}$
defined by (\ref{backwards -1bis}) satisfies%
\begin{eqnarray}
\left\vert \sigma _{\mathrm{p}}^{(\omega )}\left( \mathbf{x},\mathbf{y}%
,t\right) \right\vert &\leq &8\mathbf{D}^{-1}\left( 1+\vartheta _{0}\right)
^{2}\left\vert t\right\vert \left( \mathrm{e}^{2\mathbf{D}\left\vert
t\right\vert D_{\vartheta _{0}}}-1\right)
\label{Lieb--Robinson bounds simplifiedbis} \\
&&\times \sum_{x\in \{x^{(1)},x^{(2)}\}}\sum_{y\in \{y^{(1)},y^{(2)}\}}%
\mathbf{F}\left( \left\vert x-y\right\vert \right)  \notag
\end{eqnarray}%
for $t\in \mathbb{R}$ and%
\begin{equation*}
\mathbf{x:=}(x^{(1)},x^{(2)})\in \mathfrak{L}^{2}\ ,\qquad \mathbf{y:=}%
(y^{(1)},y^{(2)})\in \mathfrak{L}^{2}
\end{equation*}%
with $\{x^{(1)},x^{(2)}\}\cap \{y^{(1)},y^{(2)}\}=\emptyset $. This
inequality implies the existence of the macroscopic paramagnetic
conductivity defined by (\ref{paramagnetic transport coefficient macro})
with its first derivative. The existence and continuity of its second
derivative follow from Lieb--Robinson bounds for multi--commutators \cite[%
Theorems 3.8--3.9]{brupedraLR} of order three:

\begin{satz}[Paramagnetic conductivity]
\label{Theorem AC conductivity measure copy(1)}\mbox{
}\newline
Assume (\ref{(3.1) NS})--(\ref{(3.2) NS}), (\ref{static potential0}) and
that the map (\ref{map}) is a random invariant state. Let $\beta \in \mathbb{%
R}^{+}$ and $\vartheta ,\lambda \in \mathbb{R}_{0}^{+}$. Then, there is $%
\mathbf{\Xi }_{\mathrm{p}}\in C^{1}(\mathbb{R},\mathcal{B}(\mathbb{R}^{d}))$
such that, uniformly for times $t$ on compacta,
\begin{equation*}
\mathbf{\Xi }_{\mathrm{p}}\left( t\right) =\underset{l\rightarrow \infty }{%
\lim }\mathbb{E}\left[ \Xi _{\mathrm{p},l}^{(\omega )}\left( t\right) \right]
\quad \text{and}\quad \partial _{t}\mathbf{\Xi }_{\mathrm{p}}\left( t\right)
=\underset{l\rightarrow \infty }{\lim }\ \partial _{t}\mathbb{E}\left[ \Xi _{%
\mathrm{p},l}^{(\omega )}\left( t\right) \right] \ .
\end{equation*}%
Moreover, if (\ref{(3.3) NS}) also holds, then $\mathbf{\Xi }_{\mathrm{p}%
}\in C^{2}(\mathbb{R},\mathcal{B}(\mathbb{R}^{d}))$ and, uniformly for times
$t$ on compacta,
\begin{equation*}
\partial _{t}^{2}\mathbf{\Xi }_{\mathrm{p}}\left( t\right) =\underset{%
l\rightarrow \infty }{\lim }\ \partial _{t}^{2}\mathbb{E}\left[ \Xi _{%
\mathrm{p},l}^{(\omega )}\left( t\right) \right] \ .
\end{equation*}
\end{satz}

\begin{proof}
The three limits are proven in the same way. The first two only need (\ref%
{Lieb--Robinson bounds simplifiedbis}), which follows from usual
Lieb--Robinson bounds. By contrast, the last limit requires Lieb--Robinson
bounds for multi--commutators \cite[Theorems 3.8--3.9]{brupedraLR} of order
three and is thus technically more difficult than the other ones. As a
consequence, we focus on the limit of $\partial _{t}^{2}\mathbb{E}[\Xi _{%
\mathrm{p},l}^{(\omega )}\left( t\right) ]$, as $l\rightarrow \infty $, and
we omit the details for the first two.

Fix $\beta \in \mathbb{R}^{+}$, $\vartheta _{0},\lambda \in \mathbb{R}%
_{0}^{+}$, $\vartheta \in \lbrack 0,\vartheta _{0}]$, $k,q\in \{1,\ldots
,d\} $ and $t\in \mathbb{R}$. By \cite[Theorem 2.1 (i)]{OhmV}, $\{\tau
_{t}^{(\omega ,\vartheta ,\lambda )}\}_{t\in {\mathbb{R}}}$ is a $C_{0}$%
--group of $\ast $--auto%
\-%
morphisms with generator $\delta ^{(\omega ,\vartheta ,\lambda )}$. We thus
compute from Equations (\ref{backwards -1bis}) and (\ref{average microscopic
AC--conductivity}) that
\begin{equation*}
\partial _{t}^{2}\left\{ \mathbb{E}\left[ \Xi _{\mathrm{p},l}^{(\omega
)}\left( t\right) \right] \right\} _{k,q}=\frac{1}{\left\vert \Lambda
_{l}\right\vert }\underset{x,y\in \Lambda _{l}}{\sum }\mathbb{E}\left[
\varrho ^{(\beta ,\omega ,\vartheta ,\lambda )}\left( i[I_{\left(
y+e_{k},y\right) }^{(\omega ,\vartheta )},\tau _{t}^{(\omega ,\vartheta
,\lambda )}\circ \delta ^{(\omega ,\vartheta ,\lambda )}(I_{\mathbf{(}%
x+e_{q},x\mathbf{)}}^{(\omega ,\vartheta )})]\right) \right] \ .
\end{equation*}%
Then, since the map (\ref{map}) is a random invariant state and $\mathfrak{a}%
_{\Omega }$ is an ergodic measure while (\ref{static potential0}) holds, one
computes that%
\begin{eqnarray}
&&\partial _{t}^{2}\left\{ \mathbb{E}\left[ \Xi _{\mathrm{p},l}^{(\omega
)}\left( t\right) \right] \right\} _{k,q}  \notag \\
&=&\frac{1}{\left\vert \Lambda _{l}\right\vert }\underset{x,y\in \Lambda _{l}%
}{\sum }\mathbb{E}\left[ \varrho ^{(\beta ,\chi _{y}^{(\Omega )}\left(
\omega \right) ,\vartheta ,\lambda )}\left( i[I_{\left( e_{k},0\right)
}^{(\chi _{y}^{(\Omega )}\left( \omega \right) ,\vartheta )},\tau
_{t}^{(\chi _{y}^{(\Omega )}\left( \omega \right) ,\vartheta ,\lambda
)}\circ \delta ^{(\chi _{y}^{(\Omega )}\left( \omega \right) ,\vartheta
,\lambda )}(I_{\mathbf{(}x-y+e_{q},x-y\mathbf{)}}^{(\chi _{y}^{(\Omega
)}\left( \omega \right) ,\vartheta )})]\right) \right]   \notag \\
&=&\underset{x\in \mathfrak{L}}{\sum }\xi _{l}\left( x\right) \mathbb{E}%
\left[ \varrho ^{(\beta ,\omega ,\vartheta ,\lambda )}\left( i[I_{\left(
e_{k},0\right) }^{(\omega ,\vartheta )},\tau _{t}^{(\omega ,\vartheta
,\lambda )}\circ \delta ^{(\omega ,\vartheta ,\lambda )}(I_{\mathbf{(}%
x+e_{q},x\mathbf{)}}^{(\omega ,\vartheta )})]\right) \right]
\label{derivee seconde1}
\end{eqnarray}
with%
\begin{equation*}
\xi _{l}\left( x\right) :=\frac{1}{\left\vert \Lambda _{l}\right\vert }%
\underset{y\in \Lambda _{l}}{\sum }\mathbf{1}_{\left\{ x\in \Lambda
_{l}-y\right\} }\in \left[ 0,1\right] \ ,\qquad x\in \mathfrak{L}\ ,\ l\in
\mathbb{R}^{+}\ .
\end{equation*}%
For any $l\in \mathbb{R}^{+}$, the map $x\mapsto \xi _{l}\left( x\right) $
on $\mathfrak{L}$ has finite support and, for any $x\in \mathfrak{L}$,
\begin{equation}
\underset{l\rightarrow \infty }{\lim }\xi _{l}\left( x\right) =1\ .
\label{limit xi}
\end{equation}%
Paramagnetic current observables (\ref{current observable}) are obviously
local elements, i.e.,$\ I_{\mathbf{x}}^{(\omega ,\vartheta )}\in \mathcal{U}%
_{0}$ for any $\mathbf{x}\in \mathfrak{L}^{2}$, while from \cite[Theorem 2.1
(ii)]{OhmV}%
\begin{equation*}
\delta ^{(\omega ,\vartheta ,\lambda )}(B)=i\sum\limits_{z,u\in \mathfrak{L}%
}\langle \mathfrak{e}_{z},(\Delta _{\omega ,\vartheta }+\lambda V_{\omega })%
\mathfrak{e}_{u}\rangle \left[ a_{z}^{\ast }a_{u},B\right]
+i\sum\limits_{\Lambda \in \mathcal{P}_{f}(\mathfrak{L})}\left[ \Psi
_{\Lambda }^{\mathrm{IP}},B\right]
\end{equation*}%
for any $B\in \mathcal{U}_{0}$. Therefore, we get that%
\begin{eqnarray}
&&\underset{x\in \mathfrak{L}}{\sum }\xi _{l}\left( x\right) \left(
i[I_{\left( e_{k},0\right) }^{(\omega ,\vartheta )},\tau _{t}^{(\omega
,\vartheta ,\lambda )}\circ \delta ^{(\omega ,\vartheta ,\lambda )}(I_{%
\mathbf{(}x+e_{q},x\mathbf{)}}^{(\omega ,\vartheta )})]\right)  \notag \\
&=&i\underset{x\in \mathfrak{L}}{\sum }\xi _{l}\left( x\right)
\sum\limits_{z,u\in \mathfrak{L}}\langle \mathfrak{e}_{z},(\vartheta \Delta
_{\omega }+\lambda V_{\omega })\mathfrak{e}_{u}\rangle \lbrack I_{\left(
e_{k},0\right) }^{(\omega ,\vartheta )},\tau _{t}^{(\omega ,\vartheta
,\lambda )}([a_{z}^{\ast }a_{u},I_{\mathbf{(}x+e_{q},x\mathbf{)}}^{(\omega
,\vartheta )}])]  \notag \\
&&+i\underset{x\in \mathfrak{L}}{\sum }\xi _{l}\left( x\right)
\sum\limits_{\Lambda \in \mathcal{P}_{f}(\mathfrak{L})}[I_{\left(
e_{k},0\right) }^{(\omega ,\vartheta )},\tau _{t}^{(\omega ,\vartheta
,\lambda )}([\Psi _{\Lambda }^{\mathrm{IP}},I_{\mathbf{(}x+e_{q},x\mathbf{)}%
}^{(\omega ,\vartheta )}])]\ .  \label{ref *}
\end{eqnarray}%
The most delicate term in this equation is the last one. In fact, for all $%
x\in \mathfrak{L}$ and $m\in \mathbb{N}$, define the set
\begin{equation*}
\mathcal{D}\left( x,m\right) :=\left\{ \Lambda \in \mathcal{P}_{f}(\mathfrak{%
L}):x\in \Lambda ,\text{ }\Lambda \subseteq \Lambda _{m}+x,\ \Lambda
\nsubseteq \Lambda _{m-1}+x\right\} \subset 2^{\mathfrak{L}}\ ,
\end{equation*}%
while $\mathcal{D}\left( x,0\right) :=\{\{x\}\}$. By using (\ref{(3.3) NS}),
$\varsigma >2d$, and
\begin{equation*}
\mathcal{P}_{f}\left( \mathfrak{L}\right) =\underset{x\in \mathfrak{L},\
m\in \mathbb{N}_{0}}{\bigcup }\mathcal{D}\left( x,m\right) \ ,
\end{equation*}%
together with Lieb--Robinson bounds for multi--commutators of order three
\cite[Corollary 3.10]{brupedraLR} (tree--decay bounds), one gets that, for $%
\omega \in \Omega $, $\vartheta _{0},\lambda \in \mathbb{R}_{0}^{+}$, $%
\vartheta \in \lbrack 0,\vartheta _{0}]$, $k,q\in \{1,\ldots ,d\}$ and $T\in
\mathbb{R}^{+}$,%
\begin{eqnarray*}
&&\underset{x\in \mathfrak{L}}{\sum }\sum\limits_{\Lambda \in \mathcal{P}%
_{f}(\mathfrak{L})}\sup_{t\in \left[ -T,T\right] }\left\Vert [I_{\left(
e_{k},0\right) }^{(\omega ,\vartheta )},\tau _{t}^{(\omega ,\vartheta
,\lambda )}([\Psi _{\Lambda }^{\mathrm{IP}},I_{\mathbf{(}x+e_{q},x\mathbf{)}%
}^{(\omega ,\vartheta )}])]\right\Vert _{\mathcal{U}} \\
&=&\underset{x\in \mathfrak{L}}{\sum }\sum\limits_{\Lambda \in \mathcal{P}%
_{f}(\mathfrak{L})}\sup_{t\in \left[ -T,T\right] }\left\Vert [\tau
_{-t}^{(\omega ,\vartheta ,\lambda )}(I_{\left( e_{k},0\right) }^{(\omega
,\vartheta )}),[I_{\mathbf{(}x+e_{q},x\mathbf{)}}^{(\omega ,\vartheta
)},\Psi _{\Lambda }^{\mathrm{IP}}]]\right\Vert _{\mathcal{U}} \\
&\leq &D\left( 1+\vartheta _{0}\right) ^{2}d^{\varsigma }\left(
2D_{\vartheta _{0}}\left\Vert \mathbf{u}_{\cdot ,1}\right\Vert _{\ell ^{1}(%
\mathbb{N})}\left\vert T\right\vert \mathrm{e}^{4\mathbf{D}\left\vert
T\right\vert D_{\vartheta _{0}}}+2^{\varsigma }\right) ^{2} \\
&&\times \underset{x\in \mathfrak{L}}{\sup }\sum\limits_{m\in \mathbb{N}%
_{0}}\left( m+1\right) ^{\varsigma }\sum\limits_{\Lambda \in \mathcal{D}%
\left( x,m\right) }\left\Vert \Psi _{\Lambda }^{\mathrm{IP}}\right\Vert _{%
\mathcal{U}} \\
&<&\infty \ .
\end{eqnarray*}%
Here, the positive constant $D\in \mathbb{R}^{+}$ does not depend on $\omega
\in \Omega $, $\vartheta _{0},\lambda \in \mathbb{R}_{0}^{+}$, $\vartheta
\in \lbrack 0,\vartheta _{0}]$ and $k,q\in \{1,\ldots ,d\}$. Note that
\begin{equation*}
\underset{x\in \mathfrak{L}}{\sup }\sum\limits_{m\in \mathbb{N}_{0}}\left(
m+1\right) ^{\varsigma }\sum\limits_{\Lambda \in \mathcal{D}\left(
x,m\right) }\left\Vert \Psi _{\Lambda }^{\mathrm{IP}}\right\Vert _{\mathcal{U%
}}<\infty
\end{equation*}%
is a consequence of (\ref{(3.3) NS}) and $\Psi ^{\mathrm{IP}}\in \mathcal{W}$%
. The same kind of inequality holds for the 1st term in the r.h.s. of (\ref%
{ref *}). Then, using Lebesgue's dominated convergence theorem, one gets
from (\ref{derivee seconde1})--(\ref{limit xi}) that the map
\begin{equation*}
t\mapsto \partial _{t}^{2}\mathbb{E}\left[ \Xi _{\mathrm{p},l}^{(\omega
)}\left( t\right) \right] =\mathbb{E}\left[ \partial _{t}^{2}\Xi _{\mathrm{p}%
,l}^{(\omega )}\left( t\right) \right]
\end{equation*}%
converges uniformly on compacta, as $l\rightarrow \infty $, to a continuous
function $\partial _{t}^{2}\mathbf{\Xi }_{\mathrm{p}}\in C(\mathbb{R},%
\mathcal{B}(\mathbb{R}^{d}))$.
\end{proof}

\begin{bemerkung}[Conductivity and space decays of interactions]
\label{remark cond decay}\mbox{ }\newline
Under stronger assumptions like in the case of exponential decays of
interactions, much stronger results can be deduced from Lieb--Robinson
bounds for multi--commutators. In particular, under assumptions of \cite[%
Theorem 4.6]{brupedraLR} in the autonomous case, one verifies that $\mathbf{%
\Xi }_{\mathrm{p}}\in C^{\infty }(\mathbb{R},\mathcal{B}(\mathbb{R}^{d}))$
is a Gevrey map of order $d$. In particular, for $d=1$, $\mathbf{\Xi }_{%
\mathrm{p}}$ is in this case a real analytic map. Recall that $d\in \mathbb{N%
}$ is the space dimension of the lattice $\mathfrak{L}=\mathbb{Z}^{d}$.
\end{bemerkung}

\subsection{Study of the Paramagnetic Energy Increment\label{section
energetic study}}

The aim of this subsection is to derive the paramagnetic energy density $%
\mathfrak{i}_{\mathrm{p}}$ defined by (\ref{paramagnetic energy density}).
This is achieved in various lemmata which then yield two theorems and one
corollary. The derivation ends with Theorem \ref{main 1 copy(2)}, which
serves as springboard to obtain Theorem \ref{main 1 copy(1)}.

First, by assuming (\ref{(3.1) NS})--(\ref{(3.2) NS}) and (\ref{(3.3) NS}),
\cite[Theorem 3.8 (p)]{OhmV} says that, for any $l,\beta \in \mathbb{R}^{+}$%
, $\omega \in \Omega $, $\vartheta _{0},\lambda \in \mathbb{R}_{0}^{+}$, $%
\vartheta \in \lbrack 0,\vartheta _{0}]$, $\eta \in \mathbb{R}$, $\mathbf{A}%
\in \mathbf{C}_{0}^{\infty }$ and $t\geq t_{0}$,
\begin{equation}
\mathfrak{I}_{\mathrm{p}}^{(\omega ,\eta \mathbf{A}_{l})}\left( t\right)
=\eta ^{2}l^{d}\int\nolimits_{t_{0}}^{t}\int\nolimits_{t_{0}}^{s_{1}}\mathbf{%
X}_{l}^{(\omega )}(s_{1},s_{2})\ \mathrm{d}s_{2}\mathrm{d}s_{1}+\mathcal{O}%
(\eta ^{3}l^{d})\ ,  \label{Lemma non existing}
\end{equation}%
where, for any $s_{1},s_{2}\in \mathbb{R}$,
\begin{equation}
\mathbf{X}_{l}^{(\omega )}(s_{1},s_{2})\equiv \mathbf{X}_{l}^{(\beta ,\omega
,\vartheta ,\lambda ,\mathbf{A})}:=\frac{1}{\left\vert \Lambda
_{l}\right\vert }\underset{\mathbf{x},\mathbf{y}\in \mathfrak{K}}{\sum }%
\sigma _{\mathrm{p}}^{(\omega )}\left( \mathbf{x},\mathbf{y,}%
s_{1}-s_{2}\right) \mathbf{E}_{s_{1}}^{\mathbf{A}_{l}}(\mathbf{x})\mathbf{E}%
_{s_{2}}^{\mathbf{A}_{l}}(\mathbf{y})\ .  \label{Lemma non existing2}
\end{equation}%
The subleading term in the r.h.s. of (\ref{Lemma non existing}) is order $%
\mathcal{O}(\eta ^{3}l^{d})$, uniformly for $\beta \in \mathbb{R}^{+}$, $%
\omega \in \Omega $, $\vartheta \in \lbrack 0,\vartheta _{0}]$, $\lambda \in
\mathbb{R}_{0}^{+}$ and $t\geq t_{0}$. Here,
\begin{equation}
\mathfrak{K}:=\left\{ \mathbf{x}=(x^{(1)},x^{(2)})\in \mathfrak{L}^{2}\ :\
|x^{(1)}-x^{(2)}|=1\right\}  \label{proche voisins}
\end{equation}%
is the set of oriented bonds of nearest neighbors. Note also that the
integral in (\ref{Lemma non existing}) can be exchanged with the (finite)
sum (\ref{Lemma non existing2}) because $\mathbf{A}\in \mathbf{C}%
_{0}^{\infty }$.

The first important result of the present subsection is a proof that the
random variable $\mathbf{X}_{l}^{(\omega )}$ almost surely converges to a
constant function, as $l\rightarrow \infty $. See Corollary \ref{lemma
conductivty4 copy(6)}. To prove this, Condition (\ref{(3.3) NS}) is not
anymore necessary. Then, Lebesgue's dominated convergence theorem yields the
paramagnetic energy increment $\mathfrak{I}_{\mathrm{p}}^{(\omega ,\eta
\mathbf{A}_{l})}\left( t\right) $ in the limit $(\eta ,l^{-1})\rightarrow
(0,0)$, see Theorem \ref{main 1 copy(2)}.

We use the same strategy of proof as the one of \cite[Section 5.4]{OhmIII}
for the non--interacting case with homogeneous hopping terms. However, in
spite of interactions, we strongly simplify the corresponding technical
arguments by using Lieb--Robinson bounds. In particular, we do not anymore
need complex times. But like in \cite[Section 5.4]{OhmIII}, the (compact)
support $\mathrm{supp}(\mathbf{A}(t,\cdot ))\subset \mathbb{R}^{d}$ of the
vector potential $\mathbf{A}(t,\cdot )$ at $t\in \mathbb{R}$ is divided in
small regions to use the piecewise--constant approximation of the smooth
electric field $E_{\mathbf{A}}$. To do this, we assume w.l.o.g. that, for
all $t\in \mathbb{R}$,
\begin{equation}
\mathrm{supp}(\mathbf{A}(t,\cdot ))\subset \lbrack -1/2,1/2]^{d}\ .
\label{assumption A}
\end{equation}%
From now on we fix the parameters $\beta \in \mathbb{R}^{+}$, $\vartheta
_{0},\lambda \in \mathbb{R}_{0}^{+}$, $\vartheta \in \lbrack 0,\vartheta
_{0}]$ and $\mathbf{A}\in \mathbf{C}_{0}^{\infty }$ with (\ref{assumption A}%
).

Then, for every $n\in \mathbb{N}$, we divide the elementary box $%
[-1/2,1/2]^{d}$ in $n^{d}$ boxes $\{b_{j}\}_{j\in \mathcal{D}_{n}}$ of
side--length $1/n$, where
\begin{equation}
\mathcal{D}_{n}:=\{-\left( n-1\right) /2,-\left( n-3\right) /2,\cdots
,\left( n-3\right) /2,\left( n-1\right) /2\}^{d}\ .  \label{boxes b1}
\end{equation}%
Explicitly, for any $j\in \mathcal{D}_{n}$,%
\begin{equation}
b_{j}:=jn^{-1}+n^{-1}[-1/2,1/2]^{d}\text{\quad and\quad }[-1/2,1/2]^{d}=%
\underset{j\in \mathcal{D}_{n}}{\bigcup }b_{j}\ .  \label{boxes b2}
\end{equation}%
For any $l\in \mathbb{R}^{+}$, $\omega \in \Omega $, $n\in \mathbb{N}$ and $%
s_{1},s_{2}\in \mathbb{R}$, let
\begin{equation}
\mathbf{Y}_{l,n}^{(\omega )}(s_{1},s_{2}):=\frac{1}{\left\vert \Lambda
_{l}\right\vert }\ \underset{j\in \mathcal{D}_{n}}{\sum }\ \underset{\mathbf{%
x},\mathbf{y}\in \mathfrak{K}\cap (lb_{j})^{2}}{\sum }\sigma _{\mathrm{p}%
}^{(\omega )}\left( \mathbf{x},\mathbf{y,}s_{1}-s_{2}\right) \mathbf{E}%
_{s_{1}}^{\mathbf{A}_{l}}(\mathbf{x})\mathbf{E}_{s_{2}}^{\mathbf{A}_{l}}(%
\mathbf{y})\ .  \label{B fractbisbis}
\end{equation}%
We show now that the accumulation points of $\mathbf{Y}_{l,n}^{(\omega )}$,
as $l\rightarrow \infty $, do not depend on $n\in \mathbb{N}$ and coincide
with those of $\mathbf{X}_{l}^{(\omega )}$:

\begin{lemma}[Approximation I]
\label{lemma conductivty1}\mbox{
}\newline
Assume (\ref{(3.1) NS})--(\ref{(3.2) NS}). Let $n\in \mathbb{N}$. Then,
\begin{equation*}
\underset{l\rightarrow \infty }{\lim }\left\vert \mathbf{X}_{l}^{(\omega
)}(s_{1},s_{2})-\mathbf{Y}_{l,n}^{(\omega )}(s_{1},s_{2})\right\vert =0\ ,
\end{equation*}%
uniformly for $\omega \in \Omega $ and $s_{1},s_{2}\in \mathbb{R}$.
\end{lemma}

\begin{proof}
We observe from (\ref{Lemma non existing2}), (\ref{boxes b2}) and (\ref{B
fractbisbis}) that%
\begin{eqnarray}
&&\left\vert \mathbf{X}_{l}^{(\omega )}(s_{1},s_{2})-\mathbf{Y}%
_{l,n}^{(\omega )}(s_{1},s_{2})\right\vert  \label{B fractbisbisbis} \\
&\leq &\frac{1}{\left\vert \Lambda _{l}\right\vert }\underset{j,k\in
\mathcal{D}_{n},j\neq k}{\sum }\ \underset{\mathbf{x}\in \mathfrak{K}\cap
(lb_{j})^{2}}{\sum }\ \underset{\mathbf{y}\in \mathfrak{K}\cap (lb_{k})^{2}}{%
\sum }\left\vert \sigma _{\mathrm{p}}^{(\omega )}\left( \mathbf{x},\mathbf{y,%
}s_{1}-s_{2}\right) \mathbf{E}_{s_{1}}^{\mathbf{A}_{l}}(\mathbf{x})\mathbf{E}%
_{s_{2}}^{\mathbf{A}_{l}}(\mathbf{y})\right\vert  \notag \\
&&+\frac{1}{\left\vert \Lambda _{l}\right\vert }\underset{j\in \mathcal{D}%
_{n}}{\sum }\ \underset{\mathbf{x}\in \partial (lb_{j})}{\sum }\ \underset{%
\mathbf{y}\in \mathfrak{K}}{\sum }\left\vert \sigma _{\mathrm{p}}^{(\omega
)}\left( \mathbf{x},\mathbf{y,}s_{1}-s_{2}\right) \mathbf{E}_{s_{1}}^{%
\mathbf{A}_{l}}(\mathbf{x})\mathbf{E}_{s_{2}}^{\mathbf{A}_{l}}(\mathbf{y}%
)\right.  \notag \\
&&\qquad \qquad \qquad \qquad \qquad \left. +\sigma _{\mathrm{p}}^{(\omega
)}\left( \mathbf{y},\mathbf{x,}s_{1}-s_{2}\right) \mathbf{E}_{s_{1}}^{%
\mathbf{A}_{l}}(\mathbf{y})\mathbf{E}_{s_{2}}^{\mathbf{A}_{l}}(\mathbf{x}%
)\right\vert \ ,  \notag
\end{eqnarray}%
where, for any $\Lambda \in \mathcal{P}_{f}(\mathfrak{L})$ with complement $%
\Lambda ^{c}\subset \mathfrak{L}$,%
\begin{equation*}
\partial \Lambda :=\left\{ \mathbf{x=}(x^{(1)},x^{(2)})\in \mathfrak{K}%
:\{x^{(1)},x^{(2)}\}\cap \Lambda \neq 0,\text{ }\{x^{(1)},x^{(2)}\}\cap
\Lambda ^{c}\neq 0\right\} \ .
\end{equation*}%
Because $\mathbf{A}\in \mathbf{C}_{0}^{\infty }$, note that
\begin{equation}
\left\Vert \mathbf{E}^{\mathbf{A}}\right\Vert _{\infty }:=\sup \left\{
\left\vert E_{\mathbf{A}}(t,x)\right\vert \ :\ (t,x)\in \mathrm{supp}%
(A)\right\} <\infty \ .  \label{bound easy}
\end{equation}%
Therefore, using (\ref{(3.1) NS}), (\ref{Lieb--Robinson bounds simplifiedbis}%
), (\ref{bound easy}) and the fact that $\mathbf{A}(t,\cdot )=0$ for any $%
t\notin \lbrack t_{0},t_{1}]$ (cf. (\ref{zero mean field assumption})), we
deduce from Inequality (\ref{B fractbisbisbis}) that%
\begin{eqnarray}
\left\vert \mathbf{X}_{l}^{(\omega )}(s_{1},s_{2})-\mathbf{Y}_{l,n}^{(\omega
)}(s_{1},s_{2})\right\vert &\leq &8\left( 1+\vartheta _{0}\right) ^{2}\left(
t_{1}-t_{0}\right) \Vert \mathbf{E}^{\mathbf{A}}\Vert _{\infty }^{2}
\label{bound easy1} \\
&&\times \left( \mathrm{e}^{2\mathbf{D}\left( t_{1}-t_{0}\right)
D_{\vartheta _{0}}}\mathbf{D}^{-1}\mathbf{K}_{l}+\mathbf{\tilde{K}}%
_{l}\right)  \notag
\end{eqnarray}%
for all $l\in \mathbb{R}^{+}$, $\omega \in \Omega $, $n\in \mathbb{N}$ and $%
s_{1},s_{2}\in \mathbb{R}$, where
\begin{equation}
\mathbf{K}_{l}:=\frac{1}{\left\vert \Lambda _{l}\right\vert }\underset{%
j,k\in \mathcal{D}_{n},j\neq k}{\sum }\ \underset{x\in \mathfrak{L}\cap
(lb_{j})}{\sum }\ \underset{z_{1,2}\in \mathfrak{L},|z_{1,2}|=1}{\sum }\
\underset{y\in \mathfrak{L}\cap (lb_{k})}{\sum }\mathbf{F}\left( \left\vert
x+z_{1}+z_{2}-y\right\vert \right)  \label{def K bold}
\end{equation}

and%
\begin{equation}
\mathbf{\tilde{K}}_{l}:=2dn^{d}\left( 8+\left\Vert \mathbf{F}\right\Vert _{1,%
\mathfrak{L}}\right) \underset{z\in \mathfrak{L},|z|=1}{\sum }\frac{1}{%
\left\vert \Lambda _{l}\right\vert }\underset{x\in \mathfrak{L}}{\sum }%
\mathbf{1}\left[ \left( x,x+z\right) \in \partial (lb_{0})\right] \ .
\label{boud sup0}
\end{equation}%
Clearly, one has%
\begin{equation}
\underset{l\rightarrow \infty }{\lim }\mathbf{\tilde{K}}_{l}=0\ .
\label{boud sup 2}
\end{equation}%
Therefore, it remains to prove that $\mathbf{K}_{l}$ vanishes when $%
l\rightarrow \infty $ in order to prove the lemma.

To this end, for any $l\in \mathbb{R}^{+}$ and $\delta \in \lbrack 0,1]$,
define two constants:
\begin{multline*}
\mathbf{K}_{l,\delta }^{\lesseqgtr }:=\frac{1}{\left\vert \Lambda
_{l}\right\vert }\ \underset{j,k\in \mathcal{D}_{n},j\neq k}{\sum }\
\underset{x\in \mathfrak{L}\cap (lb_{j})}{\sum }\ \underset{z_{1,2}\in
\mathfrak{L},|z_{1,2}|=1}{\sum }\ \underset{y\in \mathfrak{L}\cap (lb_{k})}{%
\sum } \\
\mathbf{1}\left[ |x+z_{1}+z_{2}-y|\lesseqgtr \delta l\right] \mathbf{F}%
\left( \left\vert x+z_{1}+z_{2}-y\right\vert \right) \ .
\end{multline*}%
Obviously, by (\ref{def K bold}), for any $\delta ,l\in \mathbb{R}^{+}$,
\begin{equation}
\mathbf{K}_{l}=\mathbf{K}_{l,\delta }^{\leq }+\mathbf{K}_{l,\delta }^{>}\ .
\label{(3.1) NS cool-1}
\end{equation}%
Recall that $\mathbf{F}:\mathbb{R}_{0}^{+}\rightarrow \mathbb{R}^{+}$, which
encodes the short range property of interactions, is a non--increasing
function, by assumption. As a consequence, explicit estimates using $\mathbf{%
F}\left( r\right) \leq \mathbf{F}\left( 0\right) $ show that, for any $%
\delta ,l\in \mathbb{R}^{+}$,%
\begin{equation}
\left\vert \mathbf{K}_{l,\delta }^{\leq }\right\vert =\mathcal{O}\left(
\delta ^{d+1}l^{d}\right) \ ,  \label{(3.1) NS cool0}
\end{equation}%
while%
\begin{equation}
\left\vert \mathbf{K}_{l,\delta }^{>}\right\vert \leq 4d^{2}n^{d}\underset{%
x\in \mathfrak{L},|x|>\delta l}{\sum }\mathbf{F}\left( \left\vert
x\right\vert \right) \ .  \label{(3.1) NS cool}
\end{equation}%
Take $\delta =l^{-\frac{(d+1/2)}{d+1}}$. Then, by (\ref{(3.1) NS cool0}), $|%
\mathbf{K}_{l,\delta }^{\leq }|=\mathcal{O}(l^{-1/2})$ and $\delta l=l^{%
\frac{1}{2\left( d+1\right) }}$, which combined with (\ref{(3.1) NS}), (\ref%
{(3.1) NS cool-1}) and (\ref{(3.1) NS cool}) yield
\begin{equation*}
\underset{l\rightarrow \infty }{\lim }\mathbf{K}_{l}=0\ .
\end{equation*}%
By (\ref{bound easy1}) and (\ref{boud sup 2}), we thus arrive at the
assertion.
\end{proof}

We now consider piecewise--constant approximations of the (smooth) electric
field $E_{\mathbf{A}}$ (\ref{V bar 0}), that is,%
\begin{equation}
E_{\mathbf{A}}(t,x):=-\partial _{t}\mathbf{A}(t,x)\ ,\qquad t\in \mathbb{R}%
,\ x\in \mathbb{R}^{d}\ .  \label{V bar}
\end{equation}%
For any $j\in \mathcal{D}_{n}$, let $z^{(j)}\in b_{j}$ be any fixed point of
the box $b_{j}$. Then, we define the function%
\begin{eqnarray}
\mathbf{\bar{Y}}_{l,n}^{(\omega )}(s_{1},s_{2}) &:=&\frac{1}{\left\vert
\Lambda _{l}\right\vert }\ \underset{j\in \mathcal{D}_{n}}{\sum }\ \underset{\mathbf{x},\mathbf{y}\in \mathfrak{K}\cap (lb_{j})^{2}}{\sum }\sigma _{\mathrm{p}}^{(\omega )}\left( \mathbf{x},\mathbf{y,}s_{1}-s_{2}\right)
\label{piece wise approx} \\
&&\qquad \times \left[ E_{\mathbf{A}}(s_{1},z^{(j)})\right] (x^{(1)}-x^{(2)})\left[ E_{\mathbf{A}}(s_{2},z^{(j)})\right] (y^{(1)}-y^{(2)})  \notag
\end{eqnarray}%
for any $l\in \mathbb{R}^{+}$, $\omega \in \Omega $, $n\in \mathbb{N}$ and $%
s_{1},s_{2}\in \mathbb{R}$, where $\mathbf{x}:=(x^{(1)},x^{(2)})$ and $%
\mathbf{y}:=(y^{(1)},y^{(2)})$, see (\ref{proche voisins}). This new
function approximates (\ref{B fractbisbis}) arbitrarily well, as $%
l\rightarrow \infty $ and $n\rightarrow \infty $:

\begin{lemma}[Approximation II]
\label{lemma conductivty2}\mbox{
}\newline
Assume (\ref{(3.1) NS})--(\ref{(3.2) NS}). Then%
\begin{equation*}
\underset{n\rightarrow \infty }{\lim }\left\{ \underset{l\rightarrow \infty }%
{\lim \sup }\left\vert \mathbf{Y}_{l,n}^{(\omega )}\left( s_{1},s_{2}\right)
-\mathbf{\bar{Y}}_{l,n}^{(\omega )}\left( s_{1},s_{2}\right) \right\vert
\right\} =0\ ,
\end{equation*}%
uniformly for $\omega \in \Omega $ and $s_{1},s_{2}\in \mathbb{R}$.
\end{lemma}

\begin{proof}
By taking the canonical orthonormal basis $\{e_{k}\}_{k=1}^{d}$ of $\mathbb{R%
}^{d}$, we directly infer from (\ref{V bar 0bis}), (\ref{rescaled vector
potential}) and (\ref{V bar}) that, for any $l\in \mathbb{R}^{+}$, $\mathbf{A%
}\in \mathbf{C}_{0}^{\infty }$, $j\in \mathcal{D}_{n}$, $t\in \mathbb{R}$, $%
k\in \{1,\cdots ,d\}$ and $x\in lb_{j}$,%
\begin{eqnarray*}
&&\left\vert \mathbf{E}_{t}^{\mathbf{A}_{l}}(x,x\pm e_{k})-\left[ E_{\mathbf{%
A}}(t,z^{(j)})\right] (\pm e_{k})\right\vert \\
&\leq &\int\nolimits_{0}^{1}\left\vert \left[ \partial _{t}\mathbf{A}%
(t,z^{(j)})\right] (e_{k})-\left[ \partial _{t}\mathbf{A}_{l}(t,x\pm
(1-\alpha )e_{k})\right] (e_{k})\right\vert \mathrm{d}\alpha \\
&\leq &\underset{y\in \tilde{b}_{j,l}}{\sup }\left\vert \left[ \partial _{t}%
\mathbf{A}(t,z^{(j)})\right] (e_{k})-\left[ \partial _{t}\mathbf{A}(t,y)%
\right] (e_{k})\right\vert <\infty \ ,
\end{eqnarray*}%
where
\begin{equation*}
\tilde{b}_{j,l}:=\left\{ y\in \mathbb{R}^{d}\ :\ \underset{x\in b_{j}}{\min }%
\left\vert y-x\right\vert \leq l^{-1}\right\} \ .
\end{equation*}%
In particular, since $\mathbf{A}\in \mathbf{C}_{0}^{\infty }$, there is a
finite constant $D_{\mathbf{A}}\in \mathbb{R}^{+}$ not depending on $j\in
\mathcal{D}_{n}$, $t\in \mathbb{R}$, $k\in \{1,\cdots ,d\}$ and $x\in b_{j}$
such that%
\begin{equation}
\left\vert \mathbf{E}_{t}^{\mathbf{A}_{l}}(x,x\pm e_{k})-\left[ E_{\mathbf{A}%
}(t,z^{(j)})\right] (\pm e_{k})\right\vert \leq D_{\mathbf{A}%
}(n^{-1}+l^{-1})\ .  \label{inequality G cool1}
\end{equation}%
Therefore, using (\ref{(3.1) NS}), (\ref{Lieb--Robinson bounds simplifiedbis}%
), (\ref{bound easy}), (\ref{inequality G cool1}) and the fact that $\mathbf{%
A}(t,\cdot )=0$ for any $t\notin \lbrack t_{0},t_{1}]$ (cf. (\ref{zero mean
field assumption})), like in (\ref{bound easy1}), we deduce from (\ref{B
fractbisbis}) and (\ref{piece wise approx}) that%
\begin{eqnarray*}
\left\vert \mathbf{Y}_{l,n}^{(\omega )}\left( s_{1},s_{2}\right) -\mathbf{%
\bar{Y}}_{l,n}^{(\omega )}\left( s_{1},s_{2}\right) \right\vert &\leq
&64d^{2}D_{\mathbf{A}}\left\Vert \mathbf{E}^{\mathbf{A}}\right\Vert _{\infty
}\left( 1+\vartheta _{0}\right) ^{2}\left( t_{1}-t_{0}\right) (n^{-1}+l^{-1})
\\
&&\times \left( \left\Vert \mathbf{F}\right\Vert _{1,\mathfrak{L}}\mathbf{D}%
^{-1}\mathrm{e}^{2\mathbf{D}\left( t_{1}-t_{0}\right) D_{\vartheta
_{0}}}+2\right) \ .
\end{eqnarray*}%
This upper bound implies the lemma.
\end{proof}

By taking the canonical orthonormal basis $\{e_{k}\}_{k=1}^{d}$ of $\mathbb{R%
}^{d}$ and setting $e_{-k}:=-e_{k}$ for each $k\in \{1,\cdots ,d\}$, we
rewrite the function (\ref{piece wise approx}) as
\begin{eqnarray}
\mathbf{\bar{Y}}_{l,n}^{(\omega )}(s_{1},s_{2}) &:=&\frac{1}{n^{d}}\underset{j\in \mathcal{D}_{n}}{\sum }\ \underset{k,q\in \{1,-1,\cdots ,d,-d\}}{\sum }\mathbf{Z}_{l,j,k,q}^{(\omega )}(s_{1}-s_{2})  \label{Y bar} \\
&&\qquad \qquad \times \left[ E_{\mathbf{A}}(s_{1},z^{(j)})\right] (e_{k})\left[ E_{\mathbf{A}}(s_{2},z^{(j)})\right] (e_{q})  \notag
\end{eqnarray}%
for any $l\in \mathbb{R}^{+}$, $\omega \in \Omega $, $n\in \mathbb{N}$ and $%
s_{1},s_{2}\in \mathbb{R}$, where, for $l\in \mathbb{R}^{+}$, $\omega \in
\Omega $, $j\in \mathcal{D}_{n}$, $k,q\in \{1,-1,\cdots ,d,-d\}$ and $t\in
\mathbb{R}$,%
\begin{equation}
\mathbf{Z}_{l,j,k,q}^{(\omega )}(t):=\frac{n^{d}}{\left\vert \Lambda
_{l}\right\vert }\sum\limits_{x,y\in \mathfrak{L}\cap (lb_{j})}\sigma _{%
\mathrm{p}}^{(\omega )}\left( y,y-e_{q},x,x-e_{k}\mathbf{,}t\right) \ .
\label{def Z}
\end{equation}%
Notice that, as compared to (\ref{piece wise approx}), we have added in (\ref%
{Y bar}) terms related to $x,y$ on the boundary of $\mathfrak{L}\cap
(lb_{j}) $, but we use the same notation $\mathbf{\bar{Y}}_{l,n}^{(\omega )}$
for simplicity. These terms have indeed vanishing contribution in the limit $%
l\rightarrow \infty $. Here, for any $\omega \in \Omega $, $t\in \mathbb{R}$%
, $\mathbf{x}:=(x^{(1)},x^{(2)})\in \mathfrak{L}^{2}$ and $\mathbf{y}%
:=(y^{(1)},y^{(2)})\in \mathfrak{L}^{2}$,
\begin{equation}
\sigma _{\mathrm{p}}^{(\omega )}(x^{(1)},x^{(2)},y^{(1)},y^{(2)},t)\equiv
\sigma _{\mathrm{p}}^{(\omega )}\left( \mathbf{x},\mathbf{y,}t\right) \ ,
\label{B fract notation}
\end{equation}%
see (\ref{backwards -1bis}). Hence, it remains to analyze the limit of (\ref%
{def Z}), as $l\rightarrow \infty $. But before doing this study, observe
that, for all $x,y\in \mathfrak{L}$, $k,q\in \{1,-1,\cdots ,d,-d\}$ and $%
t\in \mathbb{R}$, the map
\begin{equation*}
\omega \mapsto \sigma _{\mathrm{p}}^{(\omega )}\left( y,y-e_{q},x,x-e_{k}%
\mathbf{,}t\right)
\end{equation*}%
is bounded and measurable w.r.t. the $\sigma $--algebra $\mathfrak{A}%
_{\Omega }$, by assumption. Indeed, the map (\ref{map}) is a random
invariant state (Definition \ref{def second law2 copy(1)}). Recall also that
$\mathbb{E}[\ \cdot \ ]$ is the expectation value associated with the
probability measure $\mathfrak{a}_{\Omega }$, see Section \ref{sect 2.1}.

\begin{lemma}[Infinite volume limit and ergodicity]
\label{lemma conductivty3}\mbox{
}\newline
Assume (\ref{(3.1) NS})--(\ref{(3.2) NS}), (\ref{static potential0}) and
that the map (\ref{map}) is a random invariant state. For any $t\in \mathbb{R%
}$, there is a measurable subset $\tilde{\Omega}\left( t\right) \equiv
\tilde{\Omega}^{(\beta ,\vartheta ,\lambda )}\left( t\right) \subset \Omega $
of full measure such that, for $n\in \mathbb{N}$, $j\in \mathcal{D}_{n}$, $%
k,q\in \{1,-1,\cdots ,d,-d\}$ and $\omega \in \tilde{\Omega}\left( t\right) $%
,%
\begin{equation*}
\underset{l\rightarrow \infty }{\lim }\mathbf{Z}_{l,j,k,q}^{(\omega
)}(t)=\left\{ \mathbf{\Xi }_{\mathrm{p}}\left( t\right) \right\}
_{k,q}=\sum\limits_{x\in \mathfrak{L}}\mathbb{E}\left[ \sigma _{\mathrm{p}%
}^{(\omega )}\left( x,x-e_{q},0,-e_{k}\mathbf{,}t\right) \right] \in \mathbb{%
R}\ .
\end{equation*}
\end{lemma}

\begin{proof}
For any $\omega \in \Omega $, $t\in \mathbb{R}$, $k,q\in \{1,-1,\cdots
,d,-d\}$ and $x\in \mathfrak{L}$, let
\begin{equation}
\mathfrak{F}_{t,k,q}^{(\omega )}\left( \left\{ x\right\} \right)
=\sum\limits_{y\in \mathfrak{L}}\sigma _{\mathrm{p}}^{(\omega )}\left(
y,y-e_{q},x,x-e_{k}\mathbf{,}t\right) \ .  \label{eq ohm70}
\end{equation}%
By the assumptions of the lemma, this sum is uniformly bounded for all $x\in
\mathfrak{L}$ and defines a random variable. Indeed, we infer from (\ref%
{(3.1) NS}) and (\ref{Lieb--Robinson bounds simplifiedbis}) that%
\begin{equation}
\left\vert \mathfrak{F}_{t,k,q}^{(\omega )}\left( \left\{ x\right\} \right)
\right\vert \leq 8\left( 1+\vartheta _{0}\right) ^{2}\left\vert t\right\vert
\left( \left\Vert \mathbf{F}\right\Vert _{1,\mathfrak{L}}\mathbf{D}%
^{-1}\left( \mathrm{e}^{2\mathbf{D}\left\vert t\right\vert D_{\vartheta
_{0}}}-1\right) +2\right) \ .  \label{eq ohm71}
\end{equation}%
We now define an additive process $\{\mathfrak{F}_{t,k,q}^{(\omega )}\left(
\Lambda \right) \}_{\Lambda \in \mathcal{P}_{f}(\mathfrak{L})}$ by
\begin{equation}
\mathfrak{F}_{t,k,q}^{(\omega )}\left( \Lambda \right) =\sum\limits_{x\in
\Lambda }\mathfrak{F}_{t,k,q}^{(\omega )}\left( \left\{ x\right\} \right)
\label{eq ohm72}
\end{equation}%
for any finite subset $\Lambda \in \mathcal{P}_{f}(\mathfrak{L})$ with
cardinality $\left\vert \Lambda \right\vert <\infty $, see \cite[Definition
5.2]{OhmIII}\footnote{%
Replace the product measure of \cite{OhmIII} with ergodic measures $%
\mathfrak{a}_{\Omega }$, as defined in Section \ref{sect 2.1}.}. Indeed, the
map $\omega \mapsto \mathfrak{F}_{t,k,q}^{(\omega )}\left( \Lambda \right) $
is bounded and measurable w.r.t. the $\sigma $--algebra $\mathfrak{A}%
_{\Omega }$ for all $\Lambda \in \mathcal{P}_{f}(\mathfrak{L})$. Moreover,
by Conditions (\ref{static potential0}) and (\ref{translation invariant}),
\begin{equation}
\mathfrak{F}_{t,k,q}^{(\chi _{x}^{(\Omega )}(\omega ))}\left( \Lambda
\right) =\mathfrak{F}_{t,k,q}^{(\omega )}\left( \Lambda +x\right) \ ,\qquad
\Lambda \in \mathcal{P}_{f}(\mathfrak{L}),\ x\in \mathbb{Z}^{d}\ .
\label{eq ohm73}
\end{equation}%
See Section \ref{Section dynamics}, in particular Definition \ref{def second
law2 copy(1)}. For any $\Lambda \in \mathcal{P}_{f}(\mathfrak{L})$,%
\begin{equation*}
\frac{1}{\left\vert \Lambda \right\vert }\mathbb{E}\left[ \mathfrak{F}%
_{t,k,q}^{(\omega )}\left( \Lambda \right) \right] \leq 8\left( 1+\vartheta
_{0}\right) ^{2}\left\vert t\right\vert \left( \left\Vert \mathbf{F}%
\right\Vert _{1,\mathfrak{L}}\mathbf{D}^{-1}\left( \mathrm{e}^{2\mathbf{D}%
\left\vert t\right\vert D_{\vartheta _{0}}}-1\right) +2\right) \ ,
\end{equation*}
because of (\ref{eq ohm71})--(\ref{eq ohm72}). Then, by (\ref{eq ohm73}) and
ergodicity of the measure $\mathfrak{a}_{\Omega }$, for any $t\in \mathbb{R}$
and $k,q\in \{1,-1,\cdots ,d,-d\}$, \cite[Theorem 5.5]{OhmIII}$^{2}$ applied
on the previous additive process holds and one gets the existence of a
measurable subset
\begin{equation*}
\hat{\Omega}_{k,q}\left( t\right) \equiv \hat{\Omega}_{k,q}^{(\beta
,\vartheta ,\lambda )}\left( t\right) \subset \Omega
\end{equation*}%
of full measure such that, for all $\omega \in \hat{\Omega}_{k,q}\left(
t\right) $, $n\in \mathbb{N}$ and $j\in \mathcal{D}_{n}$,%
\begin{equation}
\underset{l\rightarrow \infty }{\lim }\left\{ \frac{n^{d}}{\left\vert
\Lambda _{l}\right\vert }\mathfrak{F}_{t,k,q}^{(\omega )}\left(
lb_{j}\right) \right\} =\mathbb{E}\left[ \mathfrak{F}_{t,k,q}^{(\omega
)}\left( \left\{ 0\right\} \right) \right] \ .  \label{eq ohm7}
\end{equation}%
In the same way one proves Lemma \ref{lemma conductivty1},
\begin{equation*}
\underset{l\rightarrow \infty }{\lim }\left\{ \frac{n^{d}}{\left\vert
\Lambda _{l}\right\vert }\sum\limits_{x\in \mathfrak{L}\cap
(lb_{j})}\sum\limits_{y\in \mathfrak{L}\backslash \left\{ \mathfrak{L}\cap
(lb_{j})\right\} }\sigma _{\mathrm{p}}^{(\omega )}\left( y,y-e_{q},x,x-e_{k}%
\mathbf{,}s_{1}-s_{2}\right) \right\} =0\ .
\end{equation*}%
Using this with (\ref{eq ohm70}), (\ref{eq ohm72}) and (\ref{eq ohm7}), and
observing meanwhile from the proof of Theorem \ref{Theorem AC conductivity
measure copy(1)} that
\begin{equation*}
\mathbb{E}\left[ \mathfrak{F}_{t,k,q}^{(\omega )}\left( \left\{ 0\right\}
\right) \right] =\sum\limits_{x\in \mathfrak{L}}\mathbb{E}\left[ \sigma _{%
\mathrm{p}}^{(\omega )}\left( x,x-e_{q},0,-e_{k}\mathbf{,}t\right) \right]
=\left\{ \mathbf{\Xi }_{\mathrm{p}}\left( t\right) \right\} _{k,q}
\end{equation*}%
for all $k,q\in \{1,-1,\cdots ,d,-d\}$ and any $t\in \mathbb{R}$, we arrive
at the assertion for any realization $\omega \in \tilde{\Omega}\left(
t\right) $ with%
\begin{equation*}
\tilde{\Omega}\left( t\right) :=\underset{k,q\in \{1,-1,\cdots ,d,-d\}}{%
\bigcap }\hat{\Omega}_{k,q}\left( t\right) \ .
\end{equation*}%
[Any countable intersection of measurable sets of full measure has full
measure.]
\end{proof}

Exactly like in the proof of Lemma \ref{lemma conductivty3}, one shows that,
for any $\beta \in \mathbb{R}^{+}$, $\vartheta ,\lambda \in \mathbb{R}%
_{0}^{+}$ and $t\in \mathbb{R}$, there is a measurable subset $\tilde{\Omega}%
\left( t\right) \equiv \tilde{\Omega}^{(\beta ,\vartheta ,\lambda )}\left(
t\right) \subset \Omega $ of full measure such that, for any $\omega \in
\tilde{\Omega}\left( t\right) $,%
\begin{equation}
\mathbf{\Xi }_{\mathrm{p}}\left( t\right) =\ \underset{l\rightarrow \infty }{%
\lim }\Xi _{\mathrm{p},l}^{(\omega )}\left( t\right) \in \mathcal{B}(\mathbb{%
R}^{d})\ .  \label{conductivity000}
\end{equation}%
This holds under Conditions (\ref{(3.1) NS})--(\ref{(3.2) NS}) and (\ref%
{static potential0}), provided that the map (\ref{map}) is a random
invariant state.

Define the deterministic function%
\begin{eqnarray}
\mathbf{X}_{\infty }\left( s_{1},s_{2}\right)  &:=&\underset{k,q\in
\{1,-1,\cdots ,d,-d\}}{\sum }\left\{ \mathbf{\Xi }_{\mathrm{p}}\left(
s_{1}-s_{2}\right) \right\} _{k,q}  \notag \\
&&\times \int\nolimits_{\mathbb{R}^{d}}\left[ E_{\mathbf{A}}(s_{1},x)\right]
(e_{k})\left[ E_{\mathbf{A}}(s_{2},x)\right] (e_{q})\mathrm{d}^{d}x
\label{X infinity}
\end{eqnarray}%
for any $s_{1},s_{2}\in \mathbb{R}$. We show next that the function $\mathbf{%
X}_{l}^{(\omega )}$ defined by (\ref{Lemma non existing2}) almost surely
converges to $\mathbf{X}_{\infty }\equiv \mathbf{X}_{\infty }^{(\beta
,\vartheta ,\lambda )}$, as $l\rightarrow \infty $:

\begin{satz}[Infinite volume limit of $\mathbf{X}$--integrands -- I]
\label{limit ohm1limit ohm1}\mbox{
}\newline
Assume (\ref{(3.1) NS})--(\ref{(3.2) NS}), (\ref{static potential0}) and
that the map (\ref{map}) is a random invariant state. Let $\beta \in \mathbb{%
R}^{+}$, $\vartheta ,\lambda \in \mathbb{R}_{0}^{+}$ and $s_{1},s_{2}\in
\mathbb{R}$. Then, there is a measurable subset $\tilde{\Omega}\left(
s_{1},s_{2}\right) \equiv \tilde{\Omega}^{(\beta ,\vartheta ,\lambda
)}\left( s_{1},s_{2}\right) \subset \Omega $ of full measure such that, for
any $\mathbf{A}\in \mathbf{C}_{0}^{\infty }$ and $\omega \in \tilde{\Omega}%
\left( s_{1},s_{2}\right) $,%
\begin{equation}
\underset{l\rightarrow \infty }{\lim }\mathbf{X}_{l}^{(\omega )}\left(
s_{1},s_{2}\right) =\mathbf{X}_{\infty }\left( s_{1},s_{2}\right) \ .  \notag
\end{equation}
\end{satz}

\begin{proof}
Let $\beta \in \mathbb{R}^{+}$, $\vartheta _{0},\lambda \in \mathbb{R}%
_{0}^{+}$, $\vartheta \in \lbrack 0,\vartheta _{0}]$ and $s_{1},s_{2}\in
\mathbb{R}$. Assume w.l.o.g. that (\ref{assumption A}) holds. Using Lemmata %
\ref{lemma conductivty1}--\ref{lemma conductivty3} and (\ref{Y bar})--(\ref%
{def Z}), we obtain the existence of a measurable subset $\tilde{\Omega}%
\left( s_{1},s_{2}\right) \equiv \tilde{\Omega}^{(\beta ,\vartheta ,\lambda
)}\left( s_{1},s_{2}\right) \subset \Omega $ of full measure such that, for
any $\omega \in \tilde{\Omega}\left( s_{1},s_{2}\right) $,
\begin{align*}
\underset{l\rightarrow \infty }{\lim }\mathbf{X}_{l}^{(\omega )}\left(
s_{1},s_{2}\right) & =\underset{k,q\in \{1,-1,\cdots ,d,-d\}}{\sum }\left\{
\mathbf{\Xi }_{\mathrm{p}}\left( s_{1}-s_{2}\right) \right\} _{k,q} \\
& \times \underset{n\rightarrow \infty }{\lim }\left\{ \frac{1}{n^{d}}%
\underset{j\in \mathcal{D}_{n}}{\sum }\left[ E_{\mathbf{A}}(s_{1},z^{(j)})%
\right] (e_{k})\left[ E_{\mathbf{A}}(s_{2},z^{(j)})\right] (e_{q})\right\} \
.
\end{align*}%
The latter implies the theorem because the term within the limit $%
n\rightarrow \infty $ is a Riemann sum and $E_{\mathbf{A}}\in \mathbf{C}%
_{0}^{\infty }$ for any $\mathbf{A}\in \mathbf{C}_{0}^{\infty }$, see (\ref%
{V bar}).
\end{proof}

To find the energy increment $\mathfrak{I}_{\mathrm{p}}^{(\omega ,\eta
\mathbf{A}_{l})}\left( t\right) $ given by (\ref{Lemma non existing}) in the
limit $(\eta ,l^{-1})\rightarrow (0,0)$, we use below Lebesgue's dominated
convergence theorem and we thus need to remove the dependency of the
measurable subset $\tilde{\Omega}\left( s_{1},s_{2}\right) $ on $%
s_{1},s_{2}\in \mathbb{R}$, see Theorem \ref{limit ohm1limit ohm1}. To
achieve this, we first show uniform boundedness and equicontinuity of the
function $\mathbf{X}_{l}^{(\omega )}$ defined by (\ref{Lemma non existing2}):

\begin{lemma}[Uniform Boundedness and Equicontinuity of $\mathbf{X}$%
--integrands]
\label{lemma conductivty4 copy(2)}\mbox{
}\newline
Assume (\ref{(3.1) NS})--(\ref{(3.2) NS}). The family%
\begin{equation*}
\left\{ \left( s_{1},s_{2}\right) \mapsto \mathbf{X}_{l}^{(\omega )}\left(
s_{1},s_{2}\right) \right\} _{l\in \mathbb{R}^{+},\omega \in \Omega }
\end{equation*}%
of maps from $\mathbb{R}^{2}$ to $\mathbb{C}$ is uniformly bounded and
equicontinuous.
\end{lemma}

\begin{proof}
The uniform boundedness of this collection of maps is an immediate
consequence of (\ref{Lieb--Robinson bounds simplifiedbis}) and (\ref{bound
easy}). The arguments are indeed similar to those proving Inequality (\ref%
{bound easy1}): Assume w.l.o.g. that (\ref{assumption A}) holds. Then, by
combining (\ref{Lemma non existing2}) with (\ref{Lieb--Robinson bounds
simplifiedbis}) and (\ref{bound easy}) one gets
\begin{equation*}
\left\vert \mathbf{X}_{l}^{(\omega )}\left( s_{1},s_{2}\right) \right\vert
\leq 32d^{2}\left\Vert \mathbf{E}^{\mathbf{A}}\right\Vert _{\infty
}^{2}\left( 1+\vartheta _{0}\right) ^{2}\left( t_{1}-t_{0}\right) \left(
\left\Vert \mathbf{F}\right\Vert _{1,\mathfrak{L}}\mathbf{D}^{-1}\mathrm{e}%
^{2\mathbf{D}\left( t_{1}-t_{0}\right) D_{\vartheta _{0}}}+2\right)
\end{equation*}%
for any $\beta \in \mathbb{R}^{+}$, $\vartheta _{0},\lambda \in \mathbb{R}%
_{0}^{+}$, $\vartheta \in \lbrack 0,\vartheta _{0}]$ and $s_{1},s_{2}\in
\mathbb{R}$.

To prove the uniform equicontinuity, we use \cite[Theorem 3.6]{OhmV}, which
is also an immediate consequence of (\ref{Lieb--Robinson bounds
simplifiedbis}). We omit the details.
\end{proof}

Theorem \ref{limit ohm1limit ohm1} and Lemma \ref{lemma conductivty4 copy(2)}
allows us to eliminate the $(s_{1},s_{2})$--dependency of the measurable set
$\tilde{\Omega}\left( s_{1},s_{2}\right) $ of Theorem \ref{limit ohm1limit
ohm1}.

\begin{koro}[Infinite volume limit of $\mathbf{X}$--integrands -- II]
\label{lemma conductivty4 copy(6)}\mbox{
}\newline
Assume (\ref{(3.1) NS})--(\ref{(3.2) NS}), (\ref{static potential0}) and
that the map (\ref{map}) is a random invariant state. Let $\beta \in \mathbb{%
R}^{+}$ and $\vartheta ,\lambda \in \mathbb{R}_{0}^{+}$. Then, there is a
measurable subset $\tilde{\Omega}\equiv \tilde{\Omega}^{(\beta ,\vartheta
,\lambda )}\subset \Omega $ of full measure such that, for any $%
s_{1},s_{2}\in \mathbb{R}$, $\mathbf{A}\in \mathbf{C}_{0}^{\infty }$ and $%
\omega \in \tilde{\Omega}$,%
\begin{equation}
\underset{l\rightarrow \infty }{\lim }\mathbf{X}_{l}^{(\omega )}\left(
s_{1},s_{2}\right) =\mathbf{X}_{\infty }\left( s_{1},s_{2}\right) \ .
\label{limit ohm1bis}
\end{equation}
\end{koro}

\begin{proof}
Fix $\beta \in \mathbb{R}^{+}$ and $\vartheta ,\lambda \in \mathbb{R}%
_{0}^{+} $. By Theorem \ref{limit ohm1limit ohm1}, for any $s_{1},s_{2}\in
\mathbb{Q}$, there is a measurable subset $\hat{\Omega}\left(
s_{1},s_{2}\right) \subset \Omega $ of full measure such that (\ref{limit
ohm1bis}) holds. Let $\tilde{\Omega}$ be the intersection of all such
subsets $\hat{\Omega}\left( s_{1},s_{2}\right) $. Since this intersection is
countable, $\tilde{\Omega}$ is measurable and has full measure. By Lemma \ref%
{lemma conductivty4 copy(2)} and the density of $\mathbb{Q}$ in $\mathbb{R}$%
, it follows that (\ref{limit ohm1bis}) holds true for any $s_{1},s_{2}\in
\mathbb{R}$, $\mathbf{A}\in \mathbf{C}_{0}^{\infty }$ and $\omega \in \tilde{%
\Omega}$.
\end{proof}

Therefore, because of (\ref{Lemma non existing}), Lemma \ref{lemma
conductivty4 copy(2)} and Corollary \ref{lemma conductivty4 copy(6)}, we can
now use Lebesgue's dominated convergence theorem to get the paramagnetic
energy density $\mathfrak{i}_{\mathrm{p}}$ defined by (\ref{paramagnetic
energy density}):

\begin{satz}[Paramagnetic energy density]
\label{main 1 copy(2)}\mbox{
}\newline
Assume (\ref{(3.1) NS})--(\ref{(3.2) NS}), (\ref{static potential0})--(\ref%
{(3.3) NS}) and that the map (\ref{map}) is a random invariant state. Let $%
\beta \in \mathbb{R}^{+}$ and $\vartheta ,\lambda \in \mathbb{R}_{0}^{+}$.
Then, there is a measurable subset $\tilde{\Omega}\equiv \tilde{\Omega}%
^{(\beta ,\vartheta ,\lambda )}\subset \Omega $ of full measure such that,
for any $\omega \in \tilde{\Omega}$, $\mathbf{A}\in \mathbf{C}_{0}^{\infty }$
and $t\geq t_{0}$,
\begin{equation}
\mathfrak{i}_{\mathrm{p}}\left( t\right) :=\underset{(\eta
,l^{-1})\rightarrow (0,0)}{\lim }\left\{ \left( \eta ^{2}l^{d}\right) ^{-1}%
\mathfrak{I}_{\mathrm{p}}^{(\omega ,\eta \mathbf{A}_{l})}\left( t\right)
\right\} =\int\nolimits_{t_{0}}^{t}\int\nolimits_{t_{0}}^{s_{1}}\mathbf{X}%
_{\infty }(s_{1},s_{2})\ \mathrm{d}s_{2}\mathrm{d}s_{1}\ .  \notag
\end{equation}
\end{satz}

\noindent This theorem yields Theorem \ref{main 1 copy(1)} (p).

\subsection{Appendix: the Bochner Theorem\label{section Bochner Theorem}}

For completeness, we give in this appendix a proof of the Bochner theorem
for weakly positive definite maps $\Upsilon $ from $\mathbb{R}$ to $\mathcal{%
B}(\mathbb{R}^{d})$. By weakly positive definite $\mathcal{B}(\mathbb{R}%
^{d}) $--valued map, we mean that, for any $\varphi \in C_{0}^{\infty }(%
\mathbb{R};\mathbb{R}^{d})$,
\begin{equation}
\int\nolimits_{\mathbb{R}}\mathrm{d}s\int\nolimits_{\mathbb{R}}\mathrm{d}%
t\left\langle \varphi \left( s\right) ,\Upsilon \left( t-s\right) \varphi
\left( t\right) \right\rangle _{\mathbb{R}^{d}}\geq 0\ .
\label{weak positivity of map}
\end{equation}%
It is a simple consequence of the usual Bochner theorem for weakly positive
definite complex--valued functions:

\begin{satz}[The Bochner theorem]
\label{Bochner thm copy(1)}\mbox{
}\newline
The following are equivalent:\newline
\emph{(i)} $f:\mathbb{R}\rightarrow \mathbb{C}$ is a weakly positive
definite and continuous function, i.e.,
\begin{equation*}
\int\nolimits_{\mathbb{R}}\mathrm{d}s\int\nolimits_{\mathbb{R}}\mathrm{d}t\
\overline{\varphi \left( s\right) }f\left( t-s\right) \varphi \left(
t\right) \geq 0\ ,\qquad \varphi \in C_{0}^{\infty }(\mathbb{R};\mathbb{C})\
.
\end{equation*}%
\emph{(ii)} There is a unique finite positive measure $\mu $ on $\mathbb{R}$
such that%
\begin{equation*}
f\left( t\right) =\int\nolimits_{\mathbb{R}}\mathrm{e}^{it\nu }\mu \left(
\mathrm{d}\nu \right) \ ,\qquad t\in \mathbb{R}\ .
\end{equation*}
\end{satz}

\begin{proof}
See for instance \cite[Theorem IX.9 and discussion thereafter]{ReedSimonII}.
\end{proof}

\begin{koro}[A Bochner theorem for real matrix--valued maps]
\label{Bochner thm2}\mbox{
}\newline
Let $\Upsilon :\mathbb{R}\rightarrow \mathcal{B}(\mathbb{R}^{d})$ be a
weakly positive definite continuous map. If $\Upsilon \left( t\right)
=\Upsilon \left( -t\right) \in \mathcal{B}(\mathbb{R}^{d})$ is symmetric
w.r.t. the canonical scalar product of $\mathbb{R}^{d}$ for any $t\in
\mathbb{R}$, then there is a unique finite and symmetric $\mathcal{B}_{+}(%
\mathbb{R}^{d})$--valued measure $\mu _{\Upsilon }$ on $\mathbb{R}$ such that%
\begin{equation*}
\Upsilon \left( t\right) =\int\nolimits_{\mathbb{R}}\cos \left( t\nu \right)
\mu _{\Upsilon }\left( \mathrm{d}\nu \right) \ .
\end{equation*}
\end{koro}

\begin{proof}
First, for any $t\in \mathbb{R}$, we define $\Upsilon (t)$ as an operator on
$\mathbb{C}^{d}$ by
\begin{equation*}
\Upsilon \left( t\right) \left( \vec{w}_{R}+i\vec{w}_{I}\right) =\Upsilon
\left( t\right) \vec{w}_{R}+i\Upsilon \left( t\right) \vec{w}_{I}\ ,\qquad
\vec{w}_{R},\vec{w}_{I}\in \mathbb{R}^{d}\ .
\end{equation*}%
For $\vec{w}\in \mathbb{C}^{d}$, let $f_{\vec{w}}$ be the complex--valued
function on $\mathbb{R}$ defined by
\begin{equation*}
f_{\vec{w}}\left( t\right) :=\left\langle \vec{w},\Upsilon \left( t\right)
\vec{w}\right\rangle _{\mathbb{C}^{d}}\ ,\qquad t\in \mathbb{R}\ .
\end{equation*}%
If $\Upsilon (t)=\Upsilon (-t)\in \mathcal{B}(\mathbb{R}^{d})$ is symmetric
w.r.t. the canonical scalar product of $\mathbb{R}^{d}$ for any $t\in
\mathbb{R}$, then $f_{\vec{w}}$ is a weakly positive definite and continuous
(complex--valued) function. By Theorem \ref{Bochner thm copy(1)}, for any $%
\vec{w}\in \mathbb{C}^{d}$, there is a unique finite positive measure $\mu _{%
\vec{w}}$ on $\mathbb{R}$ such that%
\begin{equation}
f_{\vec{w}}\left( t\right) =\int\nolimits_{\mathbb{R}}\mathrm{e}^{it\nu }\mu
_{\vec{w}}\left( \mathrm{d}\nu \right) \ ,\qquad t\in \mathbb{R}\ .
\label{bochner d=1}
\end{equation}%
Now, we define a $\mathcal{B}(\mathbb{R}^{d})$--valued measure $\mu
_{\Upsilon }$ on $\mathbb{R}$ by using the polarization identity: For any
Borel set $\mathcal{X}\subset \mathbb{R}$, in the canonical orthonormal
basis $\{e_{k}\}_{k=1}^{d}$ of $\mathbb{R}^{d}$,
\begin{equation}
\left\langle e_{k},\mu _{\Upsilon }\left( \mathcal{X}\right)
e_{q}\right\rangle _{\mathbb{R}^{d}}:=\frac{1}{4}\left( \mu
_{e_{k}+e_{q}}\left( \mathcal{X}\right) -\mu _{e_{k}-e_{q}}\left( \mathcal{X}%
\right) \right) \ .  \label{polarization identity}
\end{equation}%
By this definition, $\mu _{\Upsilon }\left( \mathcal{X}\right) $ is a
symmetric operator on $\mathbb{R}^{d}$ (w.r.t. the canonical scalar
product). Moreover, one can check that, for all $\vec{w}\in \mathbb{R}^{d}$
and any Borel set $\mathcal{X}\subset \mathbb{R}$,
\begin{equation}
\left\langle \vec{w},\mu _{\Upsilon }\left( \mathcal{X}\right) \vec{w}%
\right\rangle _{\mathbb{R}^{d}}=\mu _{\vec{w}}\left( \mathcal{X}\right) \ .
\label{equation}
\end{equation}%
Indeed, if $\vec{w}:=(w_{1},\ldots ,w_{d})\in \mathbb{R}^{d}$ then, by
symmetry of the operator $\Upsilon (t)\in \mathcal{B}(\mathbb{R}^{d})$,
\begin{equation*}
f_{\vec{w}}\left( t\right) =\frac{1}{4}\sum_{k,q=1}^{d}w_{k}w_{q}\left(
f_{e_{k}+e_{q}}\left( t\right) -f_{e_{k}-e_{q}}\left( t\right) \right) \
,\qquad t\in \mathbb{R}\ .
\end{equation*}%
Hence, from the injectivity of the Fourier transform of finite measures,
\begin{equation*}
\mu _{\vec{w}}=\frac{1}{4}\sum_{k,q=1}^{d}w_{k}w_{q}\left( \mu
_{e_{k}+e_{q}}-\mu _{e_{k}-e_{q}}\right)
\end{equation*}%
and (\ref{equation}) follows. By positivity of $\mu _{\vec{w}}$, $\mu
_{\Upsilon }$ is a $\mathcal{B}_{+}(\mathbb{R}^{d})$--valued measure on $%
\mathbb{R}$. Moreover, we deduce from (\ref{bochner d=1}) that
\begin{equation*}
\left\langle \vec{w},\Upsilon \left( t\right) \vec{w}\right\rangle _{\mathbb{%
R}^{d}}=\int\nolimits_{\mathbb{R}}\mathrm{e}^{it\nu }\left\langle \vec{w}%
,\mu _{\Upsilon }\left( \mathrm{d}\nu \right) \vec{w}\right\rangle _{\mathbb{%
R}^{d}}\ ,\qquad t\in \mathbb{R}\ ,\ \vec{w}\in \mathbb{R}^{d}\ .
\end{equation*}%
If $\Upsilon \left( t\right) =\Upsilon \left( -t\right) $ for any $t\in
\mathbb{R}$, then $\mu _{\Upsilon }\left( \mathcal{X}\right) =\mu _{\Upsilon
}\left( -\mathcal{X}\right) $ for any Borel set $\mathcal{X}\subset \mathbb{R%
}$ and hence,
\begin{equation}
\left\langle \vec{w},\Upsilon \left( t\right) \vec{w}\right\rangle _{\mathbb{%
R}^{d}}=\int\nolimits_{\mathbb{R}}\cos \left( t\nu \right) \left\langle \vec{%
w},\mu _{\Upsilon }\left( \mathrm{d}\nu \right) \vec{w}\right\rangle _{%
\mathbb{R}^{d}}\ ,\qquad t\in \mathbb{R}\ ,\ \vec{w}\in \mathbb{R}^{d}\ .
\label{tototototto}
\end{equation}%
By using the symmetry of the operators $\Upsilon \left( t\right) \in
\mathcal{B}(\mathbb{R}^{d})$ and
\begin{equation*}
\int\nolimits_{\mathbb{R}}\cos \left( t\nu \right) \mu _{\Upsilon }\left(
\mathrm{d}\nu \right) \in \mathcal{B}(\mathbb{R}^{d})
\end{equation*}%
at any fixed $t\in \mathbb{R}$, we arrive at the assertion from (\ref%
{tototototto}).
\end{proof}

\bigskip

\noindent \textit{Acknowledgments:} This research is supported by the agency
FAPESP under Grant 2013/13215-5 as well as by the Basque Government through
the grant IT641-13 and the BERC 2014-2017 program and by the Spanish
Ministry of Economy and Competitiveness MINECO: BCAM Severo Ochoa
accreditation SEV-2013-0323, MTM2014-53850. Finally, we thank very much the
referees for their work and interest in the improvement of this paper.


\begin{thebibliography}{BBAC}
\bibitem[BBAC]{9780511755668} \textsc{F.~Bardou, J.-P. Bouchaud, A.~Aspect,
and C.~Cohen-Tannoudji},
\newblock {\em {L{\'e}vy Statistics and Laser
Cooling}}. \newblock Cambridge University Press, 2001. \newblock Cambridge
Books Online.

\bibitem[B]{bertoin} \textsc{J. Bertoin}, \textit{L\'{e}vy Processes},
Cambridge University Press, 1996.

\bibitem[Bo]{bovier} \textsc{A. Bovier}, \textit{Statistical Mechanics of
Disordered Systems: A Mathematical Perspective}, Cambridge Series in
Statistical and Probabilistic Mathematics, 2006.

\bibitem[BC]{Cornean} \textsc{M. H. Brynildsen, H. D. Cornean}, On the
Verdet constant and Faraday rotation for graphene-like materials, \textit{%
Rev. Math. Phys.} \textbf{25}(4) (2013) 1350007-1--28.

\bibitem[BP1]{OhmV} \textsc{J.-B. Bru and W. de Siqueira Pedra}, Microscopic
Conductivity of Lattice Fermions\ at Equilibrium -- Part II: Interacting
Particles (2014). Preprint mp-arc 14-46.

\bibitem[BP2]{brupedrahistoire} \textsc{J.-B. Bru and W. de Siqueira Pedra},
Microscopic Foundations of Ohm and Joule's Laws -- The Relevance of
Thermodynamics, to appear in the\textit{\ Proceedings of QMATH12} (2014).
Preprint mp-arc 14-26.

\bibitem[BP3]{brupedraLR} \textsc{J.-B. Bru and W. de Siqueira Pedra},
Lieb--Robinson Bounds for Multi--Commutators and Applications to Response
Theory (2014). Preprint\ mp\_arc 14-27.

\bibitem[BPH1]{OhmI} \textsc{J.-B. Bru, W. de Siqueira Pedra and C. Hertling}%
, Heat Production of Non--Interacting Fermions Subjected to Electric Fields,
to appear in \textit{Comm. Pure Appl. Math.} (2014). Preprint mp-arc 13-29.

\bibitem[BPH2]{OhmII} \textsc{J.-B. Bru, W. de Siqueira Pedra and C. Hertling%
}, Microscopic Conductivity of Lattice Fermions at Equilibrium -- Part I:
Non--Interacting Particles (2014). to appear in \textit{J. Math. Phys.}
Preprint mp-arc 14-47.

\bibitem[BPH3]{OhmIII} \textsc{J.-B. Bru, W. de Siqueira Pedra and C.
Hertling}, AC--Conductivity Measure from Heat Production of Free Fermions in
Disordered Media (2013). to appear in \textit{Archive for Rational Mechanics
and Analysis}. Preprint mp\_arc 13-88.

\bibitem[BPH4]{OhmIV} \textsc{J.-B. Bru, W. de Siqueira Pedra and C. Hertling%
}, Macroscopic Conductivity of Free Fermions in Disordered Media,\textit{\
Rev. Math. Phys.} \textbf{26}(5) (2014) 1450008-1--25.

\bibitem[BR2]{BratteliRobinson} \textsc{O. Bratteli and D.W. Robinson},
\textit{Operator Algebras and Quantum Statistical Mechanics, Vol. II, 2nd ed.%
} Springer-Verlag, New York, 1996.

\bibitem[D]{Durrett} \textsc{R. Durrett}, \textit{Probability: Theory and
Examples}. 4th edition, Cambridge University Press, 2010.

\bibitem[FMU]{FroehlichMerkliUeltschi} \textsc{J. Fr\"{o}hlich, M. Merkli
and D. Ueltschi}, Dissipative Transport: Thermal Contacts and Tunnelling
Junctions, \textit{Ann. Henri Poincar\'{e}} \textbf{4} (2003) 897--945.

\bibitem[JP1]{JaksicPillet} \textsc{V. Jaksic and C.-A. Pillet}, A Note on
the Entropy Production Formula, \textit{Contemp. Math.} \textbf{327} (2003)
175--181.

\bibitem[K]{thomson} \textsc{Lord Kelvin (William Thomson)}, On the
Dynamical Theory of Heat, with numerical results deduced from Mr Joule's
equivalent of a Thermal Unit, and M. Regnault's Observations on Steam,
Excerpts. [\S \S 1-14 \& \S \S 99-100] Transactions of the Royal Society of
Edinburgh, March, 1851, and Philosophical Magazine IV. 1852

\bibitem[KLM]{Annale} \textsc{A. Klein, O. Lenoble, and P. M\"{u}ller}, On
Mott's formula for the ac-conductivity in the Anderson model, \textit{Annals
of Mathematics} \textbf{166} (2007) 549--577.

\bibitem[KM1]{JMP-autre} \textsc{A. Klein and P. M\"{u}ller}, The
Conductivity Measure for the Anderson Model, \textit{Journal of Mathematical
Physics, Analysis, Geometry} \textbf{4} (2008) 128--150.

\bibitem[KM2]{JMP-autre2} \textsc{A. Klein and P. M\"{u}ller},
AC-conductivity and Electromagnetic Energy Absorption for the Anderson Model
in Linear Response Theory, to appear in \textit{Markov processes and Related
Fields} (2014).

\bibitem[Ky]{Kyp} \textsc{A. E. Kyprianou}, \textit{Introductory Lectures on
fluctuations of L\'{e}vy process with applications.} Universitext, Springer,
2006.

\bibitem[LR]{liebrobinsonbounds} \textsc{E.H. Lieb and D.W. Robinson}, The
Finite Group Velocity of Quantum Spin Systems, \textit{Commun. Math. Phys.}
\textbf{28} (1972) 251--257.

\bibitem[LTW]{Anderson-physics} \textsc{A. Lagendijk, B. van Tiggelen and D.
S. Wiersma}, Fifty years of Anderson localization, \textit{Physics Today}
\textbf{62}(8) (2009) 24--29.

\bibitem[LY1]{lieb-yngvasonPhysReport} \textsc{E. H. Lieb and J. Yngvason},
The physics and mathematics of the second law of thermodynamics, Phys. Rep.
\textbf{310} (1999) 1--96.

\bibitem[LY2]{lieb-yngvason} \textsc{E. H. Lieb and J. Yngvason}, The
mathematical structure of the second law of thermodynamics, \textit{Current
Developments in Mathematics}, \textbf{2001} (2002) 89--129.

\bibitem[NS1]{NS1} \textsc{S. R. Nagel and S. E. Schnatterly}, Frequency
dependence of the Drude relaxation time in metal films, \textit{Phys. Rev. B}
\textbf{9}(4) (1974) 1299--1303.

\bibitem[NS2]{NS2} \textsc{S. R. Nagel and S. E. Schnatterly}, Frequency
dependence of the Drude relaxation time in metal films: Further evidence for
a two-carrier model, \textit{Phys. Rev. B} \textbf{12}(12) (1975) 6002--6005.

\bibitem[PW]{PW} \textsc{W. Pusz and S. L. Woronowicz}, Passive States and
KMS States for General Quantum Systems, \textit{Commun. math. Phys.} \textbf{%
58} (1978) 273--290.

\bibitem[RS1]{ReedSimonI} \textsc{M. Reed and B. Simon}, \textit{Methods of
Modern Mathematical Physics, Vol. I: Functional Analysis, Self--Adjointness}%
. Academic Press, New York-London, 1980.

\bibitem[RS2]{ReedSimonII} \textsc{M. Reed and B. Simon}, \textit{Methods of
Modern Mathematical Physics, Vol. II: Fourier Analysis, Self--Adjointness}.
Academic Press, New York-London, 1975.

\bibitem[S]{Simon} \textsc{B. Simon}, \textit{The Statistical Mechanics of
Lattice Gases.} Princeton Univ. Press, Princeton, 1993.

\bibitem[SE]{SE} \textsc{J. B. Smith and H. Ehrenreich}, Frequency
dependence of the optical relaxation time in metals, \textit{Phys. Rev. B}
\textbf{25}(2) (1982) 923--930.

\bibitem[So]{meanfreepath} \textsc{E. H. Sondheimer}, The mean free path of
electrons in metals, \textit{Advances in Physics }\textbf{50}(6) (2001)
499--537.

\bibitem[SSV]{SSV} \textsc{R. L. Schilling, R. Song and Z. Vondra\v{c}ek},
\textit{Bernstein Functions: Theory and Applications}, De Gruyter, 2010.

\bibitem[T]{T} \textsc{M.-L. Th\`{e}ye}, Investigation of the Optical
Properties of Au by Means of Thin Semitransparent Films, \textit{Phys. Rev. B%
} \textbf{2} (1970) 3060.

\bibitem[W]{W} \textsc{I. R. F. Wagner}, PhD thesis: Algebraic Approach
towards Conductivity in Ergodic Media, M\"{u}nchen, Univ., Diss., 2013.

\bibitem[We]{Ohm-exp} \textsc{B. Weber et al.}, Ohm's Law Survives to the
Atomic Scale, \textit{Science} \textbf{335}(6064) (2012) 64--67.

\bibitem[YRMK]{Y} \textsc{S. J. Youn, T. H. Rho, B. I. Min, and K. S. Kim},
Extended Drude model analysis of noble metals, \textit{phys. stat. sol. (b)}
\textbf{244} (49 (2007) 1354--1362.
\end{thebibliography}
\end{document}